\documentclass[journal,twocolumn,10pt]{IEEEtran}
\usepackage{amsmath,epsfig}
\usepackage{amssymb}
\usepackage{booktabs}
\usepackage{amsthm}
\usepackage{graphicx}
\usepackage{caption}
\usepackage{subcaption}
\usepackage{rotating}
\usepackage{floatflt} % text flows around figures
\usepackage{paralist} % compact and in-paragraph lists
\usepackage{algorithm} % package for pseudo code
\usepackage{xcolor}
\usepackage{color}
\usepackage{epstopdf}
\usepackage[normalem]{ulem}
\usepackage{float}
\usepackage{todonotes}
\usepackage{comment}
\newcommand{\mypar}[1]{{\bf #1.}}
\theoremstyle{definition}
\newtheorem{defn}{Definition}
\newtheorem{myLem}{Lemma}
\newtheorem{myAlg}{Algorithm}
\newtheorem{myThm}{Theorem}
\newtheorem{myCorollary}{Corollary}
\newcommand{\R}{\ensuremath{\mathbb{R}}}

\newcommand{\eqlabel}[1]{ \stackrel{(#1)}{=} }

\newcommand{\leqlabel}[1]{ \stackrel{(#1)}{\leq} }

\DeclareMathOperator{\Id}{I}

\DeclareMathOperator{\BL}{BL}
\DeclareMathOperator{\BLT}{ABL}

\DeclareMathOperator{\GS}{GS}
\DeclareMathOperator{\Dd}{D}

\def\x{\mathbf{x}}

\def\xhat{\widehat{\x}}

\def\y{\mathbf{y}}

\def\u{\mathbf{u}}
\def\vv{\mathbf{v}}

\def\w{\mathbf{w}}

\def\N{\mathcal{N}}
\def\F{\mathcal{F}}
\def\V{\mathcal{V}}
\def\M{\mathcal{M}}

\DeclareMathOperator{\Adj}{A}

\DeclareMathOperator{\Pj}{P}

\DeclareMathOperator{\Um}{U}
\DeclareMathOperator{\Vm}{V}

\DeclareMathOperator{\W}{W}

\begin{document}
\title{Signal Recovery on Graphs: \\Fundamental Limits of Sampling Strategies}
\author{Siheng~Chen,~\IEEEmembership{Student~Member,~IEEE}, Rohan~Varma,~\IEEEmembership{Student~Member,~IEEE},
Aarti~Singh,
  Jelena~Kova\v{c}evi\'c,~\IEEEmembership{Fellow,~IEEE}% <-this % stops a space
  \thanks{Copyright (c) 2016 IEEE. Personal use of this material is permitted.
  
  S. Chen and R. Varma are with the Department of Electrical and Computer
    Engineering, Carnegie Mellon University, Pittsburgh, PA, 15213
    USA. Emails: sihengc, rohanv@andrew.cmu.edu. A. Singh is with the Department of Machine Learning, Carnegie Mellon University, Pittsburgh, PA, 15213
    USA. Email: aarti@cs.cmu.edu. J. Kova\v{c}evi\'c is with the
      Departments of Electrical and
      Computer Engineering and Biomedical Engineering, Carnegie Mellon University, Pittsburgh,
      PA. Email: jelenak@cmu.edu.  }% <-this % stops a space
}
\date{}
 \maketitle
\begin{abstract}
  This paper builds theoretical foundations for the recovery of a
  newly proposed class of smooth graph signals, approximately
  bandlimited graph signals, under three sampling strategies: uniform
  sampling, experimentally designed sampling and active sampling. We
  then state minimax lower bounds on the maximum risk for the
  approximately bandlimited class under these three sampling
  strategies and show that active sampling cannot fundamentally
  outperform experimentally designed sampling. We propose a recovery
  strategy to compare uniform sampling with experimentally designed
  sampling. As the proposed recovery strategy lends itself well to
  statistical analysis, we derive the exact mean square error for each
  sampling strategy. To study convergence rates, we introduce two
  types of graphs and find that (1) the proposed recovery strategy
  achieves the optimal rates; and (2) the experimentally designed
  sampling fundamentally outperforms uniform sampling for Type-2 class
  of graphs. To validate our proposed recovery strategy, we test it on
  five specific graphs: a ring graph with $k$ nearest neighbors, an
  Erd\H{o}s-R\'enyi graph, a random geometric graph, a small-world
  graph and a power-law graph and find that experimental results match
  the proposed theory well. This work also presents a comprehensive
  explanation for when and why sampling for semi-supervised learning
  with graphs works.
\end{abstract}
%\tableofcontents
%\pagebreak

%----------------------------------------
\section{Introduction}
%----------------------------------------

The massive amounts of data being generated from various sources,
including online social networks, citation networks, biological
networks, and physical infrastructures has inspired an emerging field
of research for analyzing data supported on graphs~\cite{Jackson:08,
  Newman:10}. Signal processing on graphs is a theoretical framework
for the analysis of high-dimensional data with complex, nonregular
structures~\cite{ShumanNFOV:13,SandryhailaM:14}; it extends classical
discrete signal processing to signals with such structure, models it
by a graph and signals by graph signals, generalizing concepts and
tools from classical discrete signal processing to graph signal
processing.  Recent work involves sampling of graph
signals~\cite{ChenVSK:15, AnisGO:15, WangLG:14}, recovery of graph
signals~\cite{ChenSMK:14,NarangGO:13}, representations for graph
signals~\cite{ZhuM:12, ChenLVSK:16}, uncertainty principles on
graphs~\cite{AgaskarL:13, TsitsveroBL:15}, graph dictionary
construction~\cite{ThanouSF:14}, semi-supervised learning with
graphs~\cite{ChenCRBGK:13,EkambaramFAB:13}, graph
denoising~\cite{NarangO:12, ChenSMK:14a}, community detection and
clustering on graphs~\cite{Tremblay:14, DongFVN:14, ChenO:14},
graph-based filter banks~\cite{NarangO:12,NarangO:13,TremblayB:16} and
graph-based transforms~\cite{SandryhailaM:13,HammondVG:11,NarangSO:10,
  ShumanFV:16} , among others.

In this paper, we consider the classical signal processing task of
sampling and recovery~\cite{VetterliKG:12,KovacevicP:08}. As the
bridge connecting sequences and functions, classical sampling theory
shows that a bandlimited function can be perfectly recovered from its
sampled sequence if the sampling rate is high enough. In the 60 years
since Shannon, many new regular and irregular sampling and recovery
frameworks have emerged to recover signals with different properties. For
example, finite rate of innovation sampling considers sampling and
recovery of signals that have a finite number of degrees of freedom
per unit time~\cite{VetterliMB:02,DragottiVB:07}, compressed sensing
considers sampling and recovery of sparse signals or signals that can
be sparsely represented by some coherent and redundant
dictionary~\cite{ CandesRT:06a, CandesENR:10} and subspace sampling
considers sampling and recovery of signals from a union of
subspaces~\cite{LuD:08, BlumensathD:09}.

The interest in sampling and recovery of graph signals has increased
in the last few years~\cite{WangLG:14, Pesenson:08, ZhuM:12a,
  AnisGO:14, ChenSK:15a}.  In~\cite{ChenSMK:14}, authors proposed an
algorithm to recover graph signals that have small variation under
uniform sampling, with an upper bound on recovery
error. In~\cite{ChenVSK:15,AnisGO:15}, authors proposed a sampling
theory for graph signals and showed perfect recovery for bandlimited
graph signals under experimentally designed sampling. Extensions
include a fast distributed algorithm~\cite{WangLG:15a, ChenSK:15b},
local weighted measurements~\cite{WangCG:12}, local
aggregation~\cite{MarquesSGR:15} and percolation from seeding
nodes~\cite{SegarraMLR:15}.

In this paper, we focus on smooth graph signals, that is, signals
  whose coefficients at each node are similar to the coefficients of
  its neighbors. We propose a new class of smooth graph signals,
called~\emph{approximately bandlimited} and build a theoretical
foundation to understand sampling and recovery of this class under
uniform sampling, experimentally designed sampling and active
sampling. We propose a recovery strategy to compare uniform sampling
and experimentally designed sampling; this recovery strategy is an
unbiased estimator for low-frequency components and achieves the
optimal rate of convergence under some assumptions on the graph
structures. In spirit, our work follows previous work that studied the
theoretical capabilities of passive and active sampling for recovering
functions from samples~\cite{CastroWN:05,CastroN:08}; the difference
is that we consider a discrete setting and deal with irregular
structures. For a smooth function, active sampling, experimentally
designed sampling and uniform sampling have the same performance from
a statistical perspective~\cite{KorostelevT:03, CastroWN:05}. For
approximately bandlimited graph signals, however, we will see that
while active sampling achieves the same rate of convergence as
experimentally designed sampling, experimentally designed sampling
fundamentally outperforms uniform sampling when the graph is
irregular.

To validate the recovery strategy, we test it on five specific graphs:
a ring graph with $k$ nearest neighbors, an Erd\H{o}s-R\'enyi graph, a
random geometric graph, a small-world graph and a power-law graph, and
show that experimental results match the proposed theory well.

\mypar{Contributions} The contributions of the paper are as follows:
We propose
\begin{itemize}
\item a new class of smooth graph signals related to existing classes
  of smooth graph signals;
\item minimax lower bounds on the recovery error under three sampling
  strategies;
\item a recovery strategy that achieves optimal rates of convergence
  on two specific types of graphs;
\item a statistical analysis of graph structures; and
\item a comprehensive explanation of when and why experimentally
  designed sampling works for semi-supervised learning with graphs.
\end{itemize}

\mypar{Outline of the paper} Section~\ref{sec:DSPG} briefly reviews
graph signal processing; Section~\ref{sec:formulation} reviews smooth
graph signal models and formulates sampling and recovery
strategies.  We propose the minimax lower bounds on the recovery
  errors in Section~\ref{sec:lower} and a recovery strategy that works
  for both uniform sampling and experimentally designed sampling in
  Section~\ref{sec:upper}. Section~\ref{sec:convergence} combines the
  results from the two previous sections and shows the optimal
  convergence rates of recovery on two types of graphs.  The proposed
recovery strategy is evaluated in Section~\ref{sec:experiments} on
five known graph classes. Section~\ref{sec:conclusions} concludes the
paper and provides pointers to future directions.

%----------------------------------------
\section{Signal Processing on Graphs}
\label{sec:DSPG}
%----------------------------------------
We now briefly review signal processing on
graphs~\cite{SandryhailaM:14}, which lays the foundation for the
proposed work.  We consider a graph $G = (\V,\Adj)$, where $\V =
\{v_0,\ldots, v_{N-1}\}$ is the set of nodes and $\Adj \in \R^{N
    \times N}$ is the graph shift, or a weighted adjacency matrix,
  representing a discrete version of the graph. The edge weight
$\Adj_{i,j}$ between nodes $v_i$ and $v_j$ is a quantitative
expression of the underlying relation between them, such as a
similarity, a dependency, or a communication pattern. To guarantee
that the shifted signal is properly scaled for comparison with the
original one~\cite{SandryhailaM:131}, we normalize the graph shift
such that $|\lambda_{\max} (\Adj)| = 1$. After the node order is
fixed, the graph signal is written as a vector
\begin{equation}
  \label{eq:graph_signal}
  \nonumber
  \x \ = \ \begin{bmatrix}
    x_0 & x_1 & \ldots & x_{N-1}
  \end{bmatrix}^T \in \R^N.
\end{equation}
 
The Jordan decomposition of $\Adj$ is~\cite{SandryhailaM:131}
\begin{equation}
  \label{eq:eigendecomposition}
  \Adj=\Vm \Lambda \Um,
\end{equation}
where the generalized eigenvectors of $\Adj$ form the columns of
matrix $\Vm$ , $\Um = \Vm^{-1}$ (the norm of each column is normalized
to one), and the eigenvalue matrix $\Lambda\in\R^{N\times N}$ is
  the block diagonal matrix of corresponding eigenvalues $\lambda_0,
  \, \lambda_1, \ \ldots, \, \lambda_{N-1}$ of $\Adj$ in descending
  order ($1 = \lambda_0 \geq \lambda_1 \geq \ldots, \, \geq
  \lambda_{N-1} \geq -1)$. These eigenvalues represent frequencies on
the graph~\cite{SandryhailaM:131}.

The~\emph{graph Fourier transform} of $\x \in \R^N$ is
\begin{equation}
  \label{eq:graph_FT}
  \widehat{\x} = \Um \x,
\end{equation}
 and the~\emph{inverse graph Fourier transform} is
\begin{displaymath}
 \x  =  \Vm  \widehat{\x}  = \sum_{k=0}^{N-1}  \widehat{x}_k \vv_k,
\end{displaymath}
where $\vv_k$ is the $k$th column of $\Vm$ and $\widehat{x}_k $ is the
$k$th component in $\widehat{\x}$. The vector $\widehat{\x}$
in~\eqref{eq:graph_FT} represents the signal's expansion in the
eigenvector basis and describes the frequency components of the graph
signal $\x$. The inverse graph Fourier transform reconstructs the
graph signal by combining graph frequency components. In general,
$\Vm$ is not orthonormal; to restrict its behavior, we assume that
\begin{equation*}
  \alpha_1  \left\| \x  \right\|_2^2\leq  \left\| \Vm \x \right\|_2^2 \leq  \alpha_2 \left\| \x  \right\|_2^2, ~~{\rm for~all}~\x \in \R^N,
\end{equation*}
where $\alpha_1, \alpha_2 > 0$, that is, $\Vm$ is a Riesz basis with
stability constants $\alpha_1, \alpha_2$~\cite{VetterliKG:12}.  When
$\Adj$ represents an undirected graph, it is symmetric, we have
$\Um = \Vm^T$ and both $\Um$ and $\Vm$ are orthonormal. For
  simplicity, we consider $\Adj$ to represent an undirected graph in
  this paper, that is, we assume $\alpha_1= \alpha_2 = 1$, but the
  proposed graph signal models and sampling strategies apply equally
  well for directed graphs. Note that there exist other versions of
the graph Fourier transform based on different approaches to
normalize the adjacency matrix, such as the eigenvector matrix of the
graph Laplacian matrix~\cite{ShumanNFOV:13}; our proposed methods work
for all the versions of the graph Fourier transforms.
Table~\ref{table:parameters} lists the notations used in the paper.

\begin{table}[h]
  \footnotesize
  \begin{center}
    \begin{tabular}{@{}lll@{}}
      \toprule
      {\bf Symbol}  & {\bf Description} & {\bf Dimension}\\
      \midrule \addlinespace[1mm]
      $\V$ &  set of graph nodes &   $N$ \\
      $ \Adj $ &  graph shift (adjacency matrix) &  $N \times N$\\ 
      $ \x $ &  graph signal &  $N$\\
      $\Vm$  & graph Fourier basis &  $N \times N$\\             	  $ \widehat{\x}$ &  graph signal in the frequency domain &  $N$\\
      $\Psi $ & sampling operator &  $m \times N$\\ 
	  $ \M $ &  sample indices &  \\
	  $\x_\M $ & sampled signal coefficients of $\x$&  $m$\\ 
	  $\widehat{\x}_{(K)} $ & first $K$ coefficients of $\widehat{\x}$&  $K$\\ 
	  $\widehat{\x}_{(-K)}$ &  last $N-K$ coefficients of $\widehat{\x}$ &  $N-K$\\ 
	  $\Vm_{(K)} $ & first $K$ columns of $\Vm$&  $N \times K$\\ 
	  $\Vm_{(-K)} $ & last $N-K$ columns of $\Vm$&  $N \times (N-K)$\\ 
	  $\pi $ & sampling score &  $N$\\
      \bottomrule
    \end{tabular}
  \end{center}
  \caption{\label{table:parameters}
    Key notation used in the paper.
 }
\end{table}

%----------------------------------------
\section{Problem Formulation}
\label{sec:formulation}
%----------------------------------------
We now review three classes of smooth graph signals and show
connections between them. We then describe the sampling and recovery
strategies of interest, connecting this work to the previous work on
sampling theory on graphs.

\subsection{Graph Signal Model}
\label{sec:model}
 We consider a graph
signal to be smooth when coefficients at each node are
  similar to the coefficients of its neighbors. We have previously defined two classes of smooth graph signals in~\cite{ChenSMK:14,ChenVSK:15}.
\begin{defn}
  \label{df:GS}
  A graph signal $\x \in \R^N$ is~\emph{globally smooth} on a graph
  $\Adj \in \R^{N \times N}$ with parameter $\eta \geq 0 $, when
  \begin{equation}
    \label{eq:GS}
    \left\|  \x - \Adj \x \right\|_2^2 \ \leq \ \eta \left\| \x \right\|_2^2.
  \end{equation}
  Denote this class of graph signals by $\GS_{\Adj}( \eta)$.
\end{defn} 
In the above, $\Adj \x$ is the shifted version of $\x$ and $\x -
  \Adj \x$ gives the first-order difference~\cite{SandryhailaM:131}.
Since we normalized
the graph shift such that $|\lambda_{\max} (\Adj)| = 1$, 
\begin{eqnarray}
\nonumber
&&  \left\| \x - \Adj \x \right\|_2^2 \ = \   \left( \left\| \x \right\|_2 +  \left\| \Adj \x \right\|_2 \right)^2
\\
\label{eq:number4}
& = & \left( \left\| \x \right\|_2 +  \left\| \Adj  \right\|_2 \left\| \x  \right\|_2  \right)^2
\ \leq \   4 \left\| \x \right\|_2^2.
\end{eqnarray}
Thus, when $\eta \geq 4$, all graph signals satisfy~\eqref{eq:GS}.
\begin{defn}
  \label{df:BL}
  A graph signal $\x \in \R^N$ is~\emph{bandlimited} on a graph $\Adj \in \R^{N \times N}$ with parameter $K \in \{0, 1, \cdots,
  N-1\}$, when the graph frequency components satisfy
  \begin{displaymath}
    \widehat{x}_k  \ = \ 0 \quad {\rm for~all~}  \quad k \geq K.
  \end{displaymath}
  Denote this class of graph signals by $\BL_{\Adj}(K)$.
\end{defn}
Note that the original definition in~\cite{ChenVSK:15} just requires
$\xhat$ to be $K$-sparse, meaning that $\xhat$ is not necessarily
lowpass. Here, instead, we fix the ordering of frequencies and
requires the bandlimited graph signals to be lowpass. The following
theorem details the relationship between these two classes.
\begin{myThm}
  \label{thm:BLvsGS}
  For any $K \in \{0, 1, \cdots, N-1\}$, $\BL_{\Adj}(K)$ is a subset
  of $\GS_{\Adj}(\eta)$, when $\eta \geq (1-\lambda_{K-1})^2$.
\end{myThm}
\begin{proof}
  To show when $\BL_{\Adj}(K) \subseteq \GS_{\Adj}(\eta)$, let $\x \in
  \BL_{\Adj}(K)$, that is,
  %\begin{displaymath}
    $
    \x = \sum_{k=0}^{K-1} \widehat{x}_k \vv_k.
    $
  %\end{displaymath}
  Then,
  \begin{eqnarray}
    \label{eq:GSR}
    && \left| x_i - \sum_{j \in \mathcal{N}_i} \Adj_{i,j} x_j  \right|
    \nonumber
    \\
    \nonumber
    & = &   \left| \left(\sum_{k=0}^{K-1} \widehat{x}_k \vv_k \right)_i - \sum_{j \in \mathcal{N}_i} \Adj_{i,j} \left(\sum_{k=0}^{K-1} \widehat{x}_k \vv_k \right)_j \right|
    \\
    \nonumber
    & = &   \left| \sum_{k=0}^{K-1} \widehat{x}_k \left(  (\vv_k)_i - \sum_{j \in \mathcal{N}_i} \Adj_{i,j} (\vv_k)_j \right) \right|
    \\
    \nonumber
    & = &   \left| \sum_{k=0}^{K-1} \widehat{x}_k \left(  (\vv_k)_i - (\Adj \vv_k)_i \right) \right|
    \\
    \nonumber
    & \eqlabel{a} &   \left| \sum_{k=0}^{K-1} \widehat{x}_k  (1 - \lambda_k)  (\vv_k)_i \right|
    \\
    & \leqlabel{b} &  (1 - \lambda_{K-1})  \left| \sum_{k=0}^{K-1} \widehat{x}_k  (\vv_k)_i  \right| \ = \  (1 - \lambda_{K-1})  |x_i|,
  \end{eqnarray}
   where $\mathcal{N}_i$ denotes the neighbors of the $i$th node
    (basically all $j$ for which $\Adj_{i,j} \neq 0$), (a) follows
    from the fact that $\vv_k$ and $\lambda_k$ are the $k$th
    eigenvector and eigenvalue of $\Adj$, respectively and thus $(\Adj
    \vv_k)_i = \lambda_k (\vv_k)_i$; and (b) from the ordering of
    eigenvalues $\lambda_{k} \geq \lambda_{k-1}$. From the above, we
  see that for bandlimited graph signals, the signal coefficient at
  each node is close to the weighted average of all its neighbors; in
  other words, bandlimited graph signals are smooth locally, which
  implies global smoothness.

    Summing~\eqref{eq:GSR} over all nodes, we get
  \begin{eqnarray*}
    && \left\|  \x - \Adj \x \right\|_2^2  \ = \ \sum_{i=0}^{N-1} |x_i - \sum_{j \in \mathcal{N}_i} \Adj_{i,j} x_j |^2
    \\
    & \leq & \sum_{i=0}^{N-1}   (1 - \lambda_{K-1})^2  |x_i|^2 \ = \   (1 - \lambda_{K-1})^2 \left\| \x \right\|_2^2,
  \end{eqnarray*}
  proving that the bandlimited graph signals form a subset of globally
  smooth graph signals. The converse is not true, however, as globally
  smooth graph signals can have isolated high-frequency components and
  are thus not bandlimited.
\end{proof}

While the recovery of globally smooth graph signals has been studied
in~\cite{ChenSMK:14} (leading to graph signal inpainting), the
criterion of global smoothness is quite loose, making it hard to
provide further theoretical insight~\cite{SharpnackS:10}. Similarly,
while the recovery of bandlimited graph signals has been studied
in~\cite{ChenVSK:15} (leading to sampling theory on graphs), the
requirement of bandlimitedness is rather restrictive, making it
impractical in real-world applications.

We thus propose a third class of smooth graph signals that relaxes the requirement of bandlimitedness, but still promotes
smoothness\footnote{ We proposed the approximately bandlimited class
  in~\cite{ChenVSK:15a}.}.
\begin{defn}
  \label{df:BLT}
  A graph signal $\x \in \R^N$ is~\emph{approximately bandlimited} on
  a graph $\Adj \in \R^{N \times N}$ with parameters $\beta \geq 1$
  and $\mu \geq 0$ , when there exists a $K \in \{0, 1, \cdots, N-1\}$
  such that the graph Fourier transform $\widehat{\x}$ satisfies
  \begin{equation}
    \label{eq:BLT}
    \sum_{k = K}^{N-1} (1+k^{2\beta}) \widehat{x}_k^2 \leq \mu \left\| \x \right\|_2^2.
  \end{equation}
  Denote this class of graph signals by $\BLT_{\Adj}(K, \beta, \mu)$.
\end{defn}
The approximately bandlimited class allows for a tail after the first
$K$ frequency components, whose shape and decay are controlled by
  $\mu$ and $\beta$; the smaller the $\mu$, the less energy from the
  high-frequency components is allowed in the tail, and the larger the
  $\beta$, the higher the penalty on the energy from those
  high-frequency components. The class of $\BLT_{\Adj}(K)$ is similar
to the ellipsoid constraints in previous
literature~\cite{Johnstone:94}, where all the frequency components are
considered in the constraints; in other words, $\BLT_{\Adj}(K)$ poses
fewer restrictions on the low-frequency components. Many real
  graph signals exhibit the approximately bandlimited property;
    for example, Figures~\ref{fig:temp} and~\ref{fig:wind} show that
   the temperature readings across the U.S and wind
  speeds across Minnesota are approximately bandlimited graph
  signals. We have found that the approximately bandlimited class is
  more powerful than the bandlimited class when representing real
  graph signals.

\begin{figure}[htb]
  \begin{center}
    \begin{tabular}{cc}
      \includegraphics[width=0.5\columnwidth]{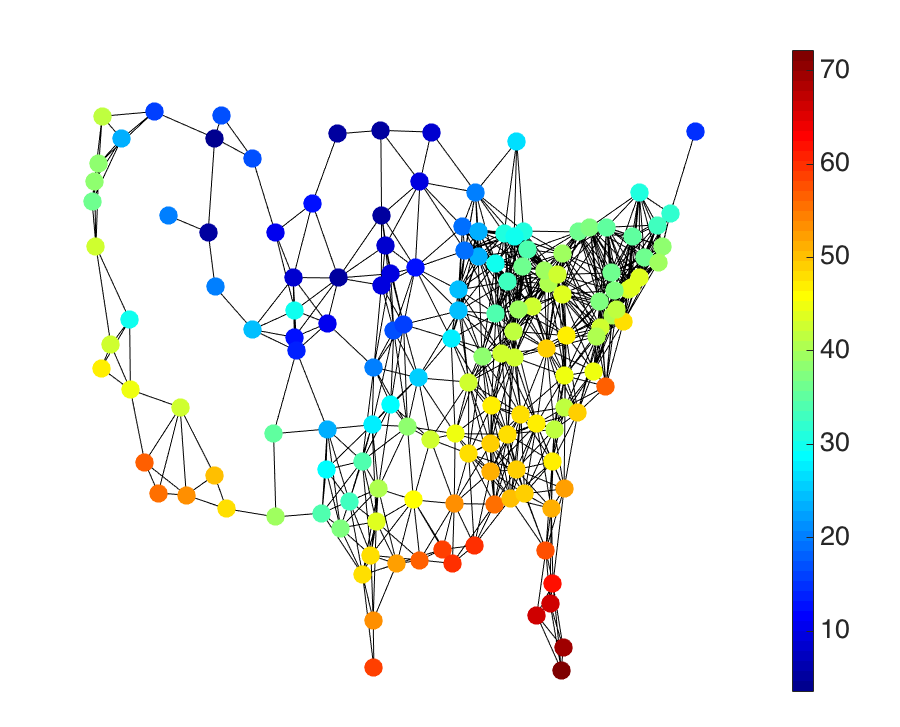}  & \includegraphics[width=0.4\columnwidth]{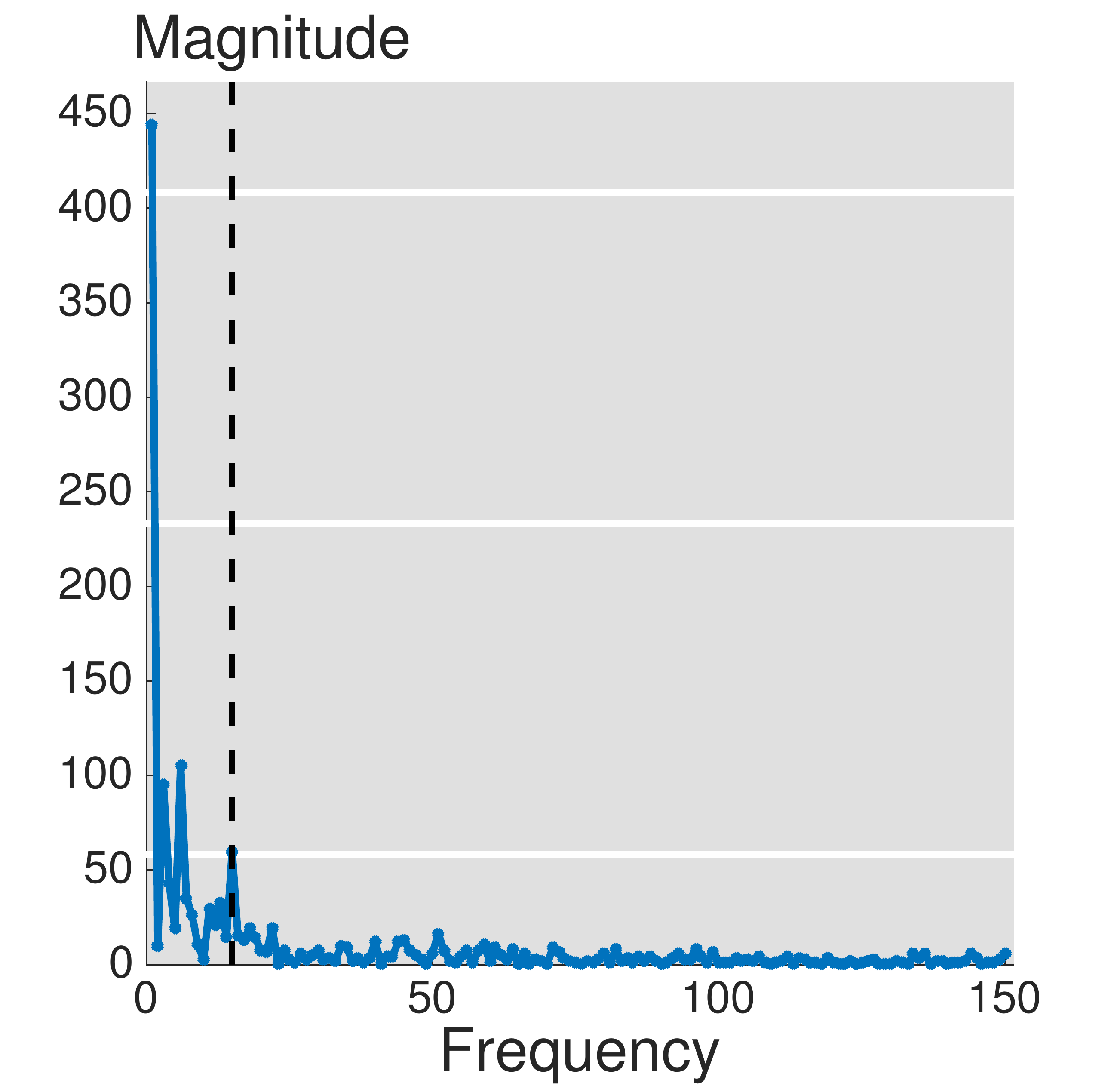} 
      \\
      {\small (a)  Temperature. } & {\small (b) Frequency content. } 
    \end{tabular}
  \end{center}
  \caption{\label{fig:temp} Temperature readings across the U.S
    is an approximately bandlimited graph signal. After the first 
      ten  frequency components (black dashed
    line), energy decays fast.  }
\end{figure}

\begin{figure}[htb]
  \begin{center}
    \begin{tabular}{cc}
\includegraphics[width=0.5\columnwidth]{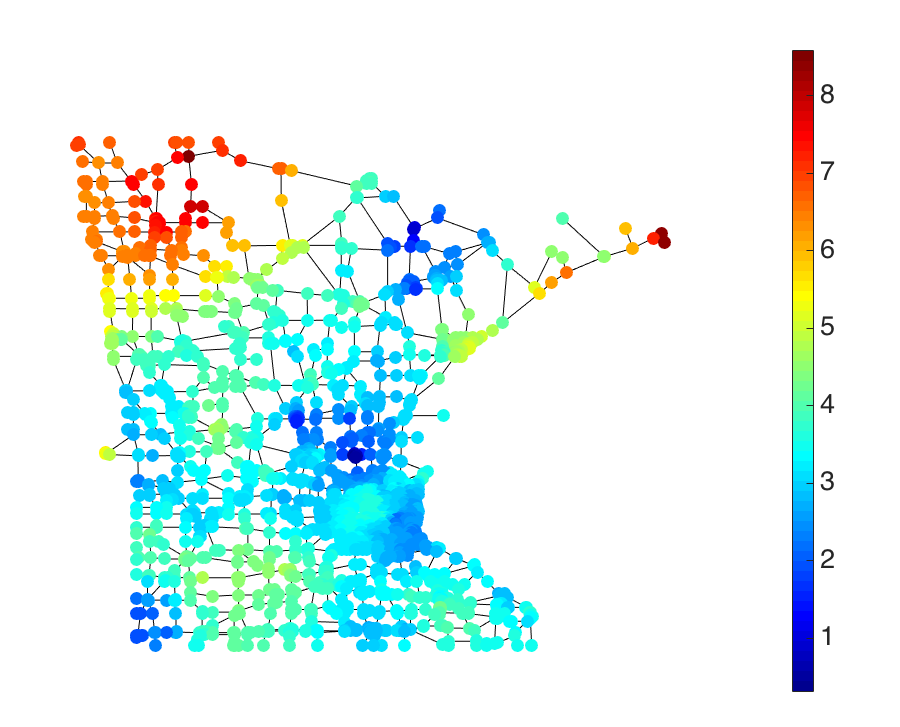}  & \includegraphics[width=0.4\columnwidth]{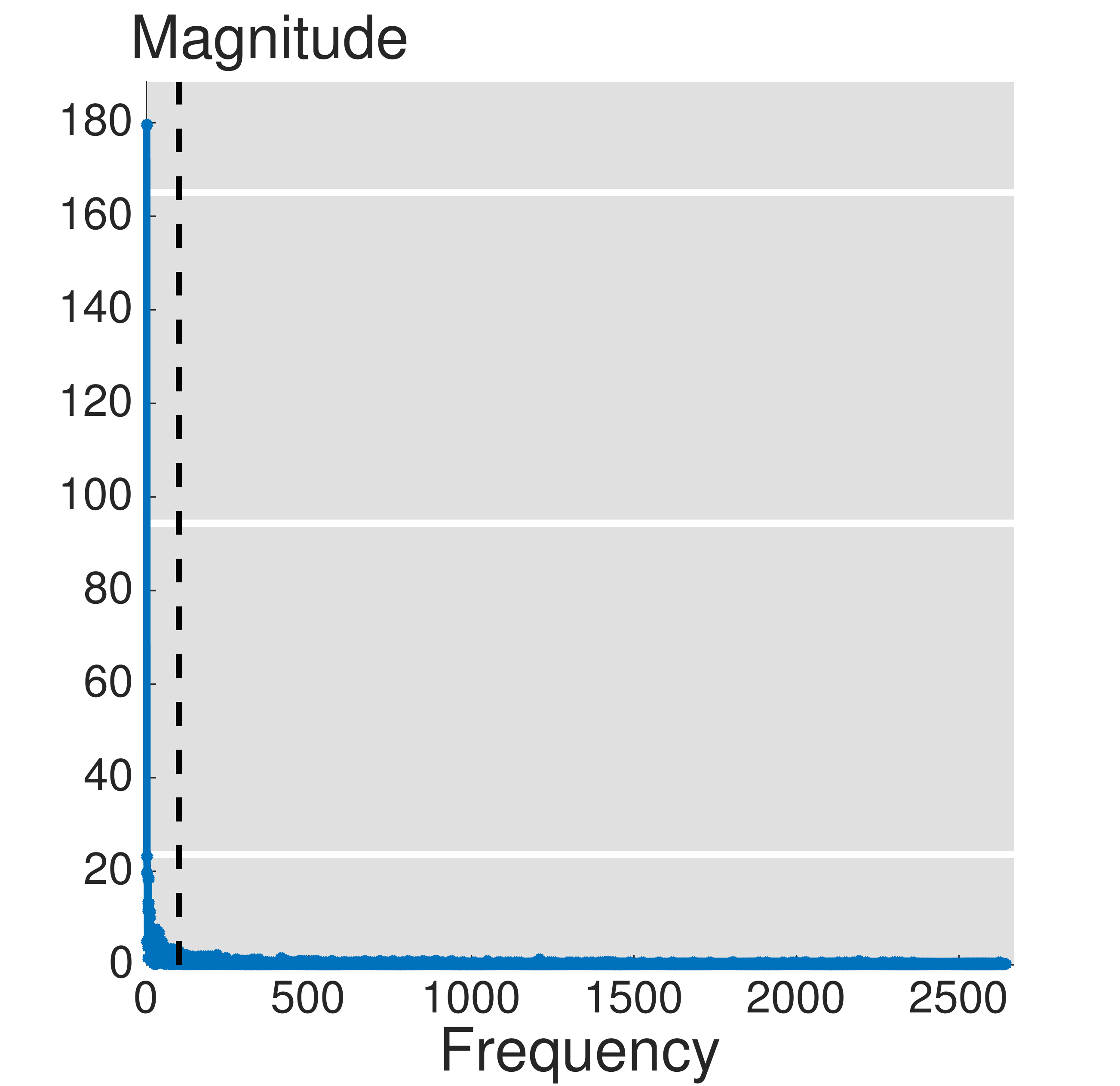} 
\\
      {\small (a)  Wind. } & {\small (b) Frequency content. } 
\end{tabular}
  \end{center}
  \caption{\label{fig:wind} Wind speed across Minnesota is an
    approximately bandlimited graph signal. After the first 100
    frequency components ( black dashed line), energy decays
    fast.  }
\end{figure}

The following theorem details the relationship between $\BLT_{\Adj}(K,
\beta, \mu)$ and $\GS_{\Adj}(\eta)$.
\begin{myThm}
  \label{thm:BLKvsGS}
With $\beta \geq 1$, $\mu, \eta \geq 0$ and $K \in \{0, 1, \cdots,
    N-1\}$,
  \begin{itemize}
  \item $\BLT_{\Adj}(K, \beta, \mu)$ is a subset of
    $\GS_{\Adj}(\eta)$, when
    \begin{equation*}
      \eta \geq  \left( 1-\lambda_{K-1}+ \sqrt{    \frac{ 4 \mu} { (1+K^{2\beta})}  } \right)^2;
    \end{equation*}   
  \item $\GS_{\Adj}(\eta)$ is a subset of $\BLT_{\Adj}(K, \beta,
    \mu)$, when
    \begin{equation*}
      \mu \geq \frac{ 1+(N-1)^{2\beta}}{1-\lambda_K} \eta.
    \end{equation*} 
  \end{itemize}
\end{myThm} 
From Theorem~\ref{thm:BLKvsGS}, we see that depending on the
parameters, $\GS_{\Adj}(\eta)$ can be a subset of $\BLT_{\Adj}(K,
\beta, \mu)$ or vice versa.
\begin{proof}
  To show when  $\BLT_{\Adj}(K, \beta, \mu) \subseteq \GS_{\Adj}(\eta)$, let $\x \in \BLT_{\Adj}(K, \beta,
  \mu)$. Then,
  \begin{eqnarray*}
    && \left\|  \x - \Adj \x \right\|  
     =  \left\| \sum_{k=0}^{N-1} \widehat{x}_k \vv_k - \Adj \sum_{k=0}^{N-1} \widehat{x}_k \vv_k   \right\|
    \\
    & \leqlabel{a} & \left\| \sum_{k=0}^{K-1} \widehat{x}_k \vv_k - \Adj  \sum_{k=0}^{K-1} \widehat{x}_k \vv_k \right\| 
     \\  && 
+ \left\| \sum_{k=K}^{N-1} \widehat{x}_k \vv_k - \Adj  \sum_{k=K}^{N-1}  \widehat{x}_k \vv_k \right\| 
    \\
    & \leqlabel{b}  & (1-\lambda_{K-1})  \left\| \x \right\| + \left\| \sum_{k=K}^{N-1} \widehat{x}_k \vv_k - \Adj  \sum_{k=K}^{N-1}  \widehat{x}_k \vv_k \right\| 
    \\
    & \eqlabel{c}  & (1-\lambda_{K-1})  \left\| \x \right\| + \left\| \sum_{k=K}^{N-1} (1-\lambda_k) \widehat{x}_k \vv_k \right\|
    \\
    &  \eqlabel{d} & (1-\lambda_{K-1})  \left\| \x \right\| + \sqrt{  \sum_{k=K}^{N-1} (1-\lambda_k)^2 \widehat{x}_k^2   }
\end{eqnarray*}
  \begin{eqnarray*}
    & = & (1-\lambda_{K-1})  \left\| \x \right\| + \sqrt{    \sum_{k=K}^{N-1} \frac{ (1-\lambda_k)^2 } { (1+k^{2\beta})} (1+k^{2\beta}) \widehat{x}_k^2   }
    \\
    &  \leq  & (1-\lambda_{K-1})  \left\| \x \right\| + 
    \\
    & + & 
    \sqrt{ \left(\max_{k \in \{K, \cdots, N-1 \} } \frac{ (1-\lambda_k)^2 } { (1+k^{2\beta})}\right) \sum_{k=K}^{N-1}(1+k^{2\beta}) \widehat{x}_k^2 }
    \\
    &  \leqlabel{ e} &  (1-\lambda_{K-1})  \left\| \x \right\| + 
    \sqrt{ \frac{ 4 } { (1+K^{2\beta})} \sum_{k=K}^{N-1}(1+k^{2\beta}) \widehat{x}_k^2  }
    \\
    & \leqlabel{f} & \left(  1-\lambda_{K-1}+ \sqrt{    \frac{ 4 \mu} { (1+K^{2\beta})}   }  \right) \left\| \x \right\| ,
  \end{eqnarray*}
  where (a) follows from the triangle inequality; (b) from
    Theorem~\ref{thm:BLvsGS}; (c) from $\Adj \vv_k = \lambda_k \vv_k$;
    (d) from the fact that $\left\| \Vm \alpha \right\|_2 =
    \|\alpha\|_2$ by Parseval's equality since $\Vm$ is orthonormal,
    with $\alpha_i = (1-\lambda_i) \widehat{x}_i$, for $i=K, \, K+1,
    \, \ldots, \, N-1$ and $0$ otherwise; (e) from $-1 \leq \lambda_K
    \leq 1$; and (f) from $\x \in \BLT_{\Adj}(K, \beta, \mu)$.

  To show the second statement when $\GS_{\Adj}(\eta) \subseteq
  \BLT_{\Adj}(K, \beta, \mu)$, let $\x \in \GS_{\Adj}(\eta)$. Then,
  \begin{eqnarray*}
    && \sum_{k = K}^{N-1} (1+k^{2\beta}) \widehat{x}_k^2 
     = \sum_{k = K}^{N-1} \frac{  1+k^{2\beta} } { (1-\lambda_k)^2 }  (1-\lambda_k)^2 \widehat{x}_k^2 
    \\
    & \leq & \left(\max_{k \in \{K, \cdots, N-1 \} } \frac{  1+k^{2\beta} }{(1-\lambda_k)^2 }\right) \sum_{k = K}^{N-1}   (1-\lambda_k)^2  \widehat{x}_k^2 
    \\
    & = &  \frac{  1+(N-1)^{2\beta} }{(1-\lambda_K)^2 } \sum_{k = K}^{N-1}   (1-\lambda_k)^2  \widehat{x}_k^2
    \\
    & \leqlabel{a} & \frac{  1+(N-1)^{2\beta} }{(1-\lambda_K)^2 } \eta \left\| \x \right\|^2,
  \end{eqnarray*}
  where (a) follows from $\sum_{k = K}^{N-1} (1-\lambda_k)^2
  \widehat{x}_k^2 \leq \sum_{k = 0}^{N-1} (1-\lambda_k)^2
  \widehat{x}_k^2 = \left\| \x -\Adj \x \right\|_2^2 \leq \eta \left\|
    \x \right\|_2^2$, where the last inequality follows from
  Definition~\ref{df:GS}.
\end{proof}
From Theorem~\ref{thm:BLKvsGS}, we see that $\BLT_{\Adj}(K, \beta,
\mu)$ is not only more general than $\BL_{\Adj}(K)$, but describes
$\GS_{\Adj}(\eta)$ in a more controlled fashion. In this paper, we
focus on $\BLT_{\Adj}(K, \beta, \mu)$, and study the recovery
performance of this class under three sampling strategies.

%----------------------------------------
\subsection{Sampling \& Recovery Strategies}
%----------------------------------------
We consider the procedure of sampling and recovery as follows: we
sample $m$ coefficients from a graph signal $\x \in \R^N$ with noise
to produce a noisy sampled signal $\y \in \R^m$, that is,
\begin{eqnarray}
\label{eq:sampling}
\y  & = & \Psi \x + \epsilon \ \equiv \  \x_\M + \epsilon,
\end{eqnarray}
where $\epsilon \sim \N( 0, \sigma^2 \Id_{m \times m})$, and $\M =
(\M_1, \ldots, \M_{m})$ denotes the sequence of sample indices called
the~\emph{sampling set}, with $\M_i \in \{0,1, \cdots, N-1\}$, $|\M| =
m$ and the sampling operator $\Psi$ is a linear mapping from $\R^N$ to
$\R^m$, defined as
\begin{equation}
\label{eq:Psi}
 \Psi_{i,j} = 
  \left\{ 
    \begin{array}{rl}
      1, & j = \M_i;\\
      0, & \mbox{otherwise}.
  \end{array} \right. 
\end{equation}
We then interpolate $\y$ to get $\x^* \in \R^N$, which recovers $\x$
either exactly or approximately.
 
We consider three different sampling strategies:
\begin{itemize}
\item \emph{uniform sampling} means that sample indices are chosen
  from $\{0, 1, \ldots, N-1\}$ independently and randomly;
\item \emph{experimentally designed sampling} means that sample
  indices can be chosen beforehand; and
\item \emph{active sampling} means that sample indices can be chosen
  as a function of the sample points and the samples collected up to
  that instance, that is, $\M_i$ depends only on $\{\M_j, y_j\}_{j <
    i}$.
\end{itemize}
It is clear that uniform sampling is a subset of experimentally
designed sampling, which is again a subset of active sampling.

{\bf The goal in this paper is to study the fundamental limitations of
  these three sampling strategies when recovering approximately
  bandlimited graph signals.}  This study is related to many
  real-world applications. For example, in semi-supervised learning,
  datasets are modeled as a graph with data samples as nodes and
  similarities between those data samples as edges. Features and
  labels associated with data samples form approximately bandlimited
  graph signals. We aim to select data samples as labeled data and
  recover the labels for the unlabeled data. The sampling strategy
  helps select the most informative data samples and minimizes the
  recovery error. Some other applications include route planning based
  on wind speed~\cite{JiCVK:15}, sensor position
  selection~\cite{SakiyamaTTO:16} and compressive spectral
  clustering~\cite{TremblayPGV:16}.

%----------------------------------------
\section{Fundamental Limits of Sampling Strategies}
\label{sec:lower}
%----------------------------------------
In this section, we study the fundamental limitations of the three
sampling strategies for recovering $\BLT_{\Adj}(K, \beta, \mu)$ by
showing minimax lower bounds. We do this by following the minimax
  decision rule and finding tight lower bounds for the minimax risk
  over all possible recovery strategies~\cite{KorostelevT:93}. In
  other words, we try to minimize the recovery error in the worst case
  and use a tight lower bound to describe this minimax error.

We start by introducing some notation~\cite{CastroWN:05}.
\begin{defn}
  For a recovery strategy $(\x^*, \M)$, and a vector $\x \in
  \R^N$,  the recovery strategy risk is
  \begin{equation*}
    R( \x^*, \M, \x)  \ = \  \mathbb{E}_{\x, \M} [d^2( \x^*, \x)],
  \end{equation*}
  where $\mathbb{E}_{\x, \M}$ is the expectation with respect to the
  probability measure of $\{ x_i, y_i \}_{i \in \M}$ and $d(\x^*,
    \x)$ is the error metric. Here we use the $\ell_2$ norm $\left\|
      \x^* - \x \right\|_2$.  The maximum risk of a recovery strategy
  is $\sup_{\x \in \R^N} R( \x^*, \M, \x)$.
\end{defn} 

The lower bounds we present will be of the form
\begin{equation}
  \label{eq:lower}
  \inf_{(\x^*, \M) \in \Theta} \sup_{\x \in \R^N} \mathbb{E}_{\x, \M} [d^2( \x^*, \x)] \geq  c \phi_m^2, \text{ for all } m \geq m_0,
\end{equation}
where  $\inf$ is the infimum (greatest lower
  bound) and $\sup$ is the supremum (least upper bound), 
$m$ is the number of samples, $m_0 \in \mathbb{N}$, $c > 0$ is a
constant, $\phi_m$ is a positive sequence converging to zero, and
$\Theta$ is the set of all recovery strategies. The sequence
$\phi_m^2$ is denoted as a lower rate of convergence.

The upper bounds on the maximum risk are usually obtained through
explicit recovery strategies, as will be shown in
Section~\ref{sec:upper}.  The upper bounds we present will be of the
form
\begin{equation}
  \label{eq:upper}
  \inf_{(\x^*, \M) \in \Theta} \sup_{\x \in \R^N} \mathbb{E}_{\x, \M} [d^2( \x^*, \x)] \leq  C \phi_m^2, \text{ for all }m \geq m_0,
\end{equation}
where $C > 0$ is a constant. If~\eqref{eq:lower} and~\eqref{eq:upper}
both hold, then $\phi_m$ is said to be the optimal rate of convergence
and there is no recovery strategy that asymptotically outperforms the
proposed one.  When talking about optimal rates of convergence, we
bound the sequence by a polynomial and make statements of the form:
Given $\gamma_1 < \gamma < \gamma_2$, a rate of convergence $\phi_m^2$
is equivalent to $m^{-\gamma}$ if and only if $m^{-\gamma_2} <
\phi_m^2 < m^{-\gamma_1} $ for $n$ large enough.

Note that the general bounds are for arbitrary graphs and thus involve
parameters that depend on the graph structure; given the graph
structure, we can specify the parameters and show explicit rates of
convergence, as we will do in Section~\ref{sec:convergence}.
 
Let $\Vm_{(2,K)}$ be the sub-matrix of $\Vm$, consisting of the $K$th
to the $(2K-1)$th columns of $\Vm$.

\begin{myThm}
  \label{thm:lower}
  Given the class $\BLT_{\Adj}(K, \beta, \mu)$: \\
  (1) Under uniform sampling,
  \begin{eqnarray*}
    && \inf_{(\x^*, \M) \in \Theta_{\rm u} } \sup_{\x \in \BLT_{\Adj}(K, \beta, \mu)} \mathbb{E}_{\x, \M} \left(  \frac{ \left\|  \x^* - \x \right\|_2^2 } { \left\|  \x \right\|_2^2  } \right) 
    \\
    & \geq &  \max_{K \leq \kappa_0 \leq N}  \frac{c_1 \mu  } {\kappa_0^{2\beta}}  \left( 1 -   \frac{ c \mu  \left\| \x \right\|_2^2 } { \sigma^2 \kappa_0^{2\beta+2} N} \left\| \Vm_{(2,\kappa_0)} \right\|_{F}^2 m \right),
  \end{eqnarray*}
  where $c_1 >0$ , $0 < c < 1$, and $\Theta_{\rm u}$ denotes the set
  of all recovery strategies based on uniform sampling;
  \\
  (2) Under experimentally designed sampling,
  \begin{eqnarray*}
    && \inf_{(\x^*, \M) \in \Theta_{\rm e} } \sup_{\x \in \BLT_{\Adj}(K, \beta, \mu)} \mathbb{E}_{\x, \M}   \left(  \frac{ \left\|  \x^* - \x \right\|_2^2 } { \left\|  \x \right\|_2^2  } \right) 
    \\
    & \geq &  \max_{K \leq \kappa_0 \leq N} \frac{c_1 \mu } {\kappa_0^{2\beta}}  \left( 1 - \frac{ c \mu \left\| \x \right\|_2^2  } { \sigma^2 \kappa_0^{2\beta+2}}  \left\| \Vm_{(2,\kappa_0)} \right\|_{\infty,2}^2 m \right),
  \end{eqnarray*}
  where $c_1 >0$ , $0 < c < 1$, and $\Theta_{\rm e}$ denotes the set
  of all recovery strategies based on experimentally designed
  sampling;
  \\
  (3) Under active sampling,
  \begin{eqnarray*}
    && \inf_{(\x^*, \M) \in \Theta_{\rm a}} \sup_{\x \in \BLT_{\Adj}(K, \beta, \mu)} \mathbb{E}_{\x, \M}  \left(  \frac{ \left\|  \x^* - \x \right\|_2^2 } { \left\|  \x \right\|_2^2  } \right)  
    \\
    & \geq &   \max_{K \leq \kappa_0 \leq N} \frac{c_1 \mu  } {\kappa_0^{2\beta}}  \left( 1 - \frac{ c \mu  \left\| \x \right\|_2^2 } { \sigma^2 \kappa_0^{2\beta+2}}  \left\| \Vm_{(2,\kappa_0)} \right\|_{\infty,2}^2 m \right),
  \end{eqnarray*}
  where $c_1 >0$ , $0 < c < 1$, and $\Theta_{\rm a}$ denotes the set
  of all recovery strategies based on active sampling.
\end{myThm} 
See Appendix A in the supporting document for the proof of these
results.  From Theorem~\ref{thm:lower}, we see that experimentally
designed sampling has the same minimax lower bound as active sampling,
which means that collecting the feedback before taking samples does
not fundamentally improve the recovery performance. We also see that
the three minimax lower bounds depend on the properties of $\Vm_{(2,
  \kappa_0)}$, which depend on the graph structure. When each row of
$\Vm_{(2,\kappa_0)}$ has roughly similar energy, $ \left\|
  \Vm_{(2,\kappa_0)} \right\|_{F}^2$ and $ N \left\|
  \Vm_{(2,\kappa_0)} \right\|_{\infty, 2}^2$ are similar; when the
energy is concentrated in a few rows, $ N \left\| \Vm_{(2,\kappa_0)}
\right\|_{\infty, 2}^2$ is much larger than $ \left\|
  \Vm_{(2,\kappa_0)} \right\|_{F}^2$; in other words, the minimax
lower bound for experimentally designed sampling is tighter than that
for uniform sampling. This happens in many real-world graphs that have
complex, irregular structure. The minimax lower bounds thus show the
potential advantage of experimentally designed sampling and active
sampling over uniform sampling. We will elaborate on this in
Sections~\ref{sec:convergence} and~\ref{sec:experiments}.

%----------------------------------------
\section{Recovery Strategy}
\label{sec:upper}
%----------------------------------------
In the previous section, we presented the minimax lower bounds for
each of the three sampling strategies and showed that active sampling
cannot fundamentally perform better than experimentally designed
sampling. We now propose a recovery strategy for both uniform sampling
and experimentally designed sampling and evaluate its statistical
properties.

\subsection{Algorithm}
To analyze uniform sampling and experimentally designed sampling in a
similar manner, we consider~\emph{sampling score-based sampling} that
unifies both sampling strategies. Sampling score-based sampling means
that the sample indices are chosen from an importance sampling
distribution that is proportional to some sampling score. Let
$\{\pi_{i}\}_{i=1}^N$ be a set of sampling scores, where $\pi_i$
denotes the probability to choose the $i$th sample in each random
trial. When the sampling score for each node is the same, we get uniform sampling; when the sampling score for each node
is designed based on the graph structure,  we get
experimentally designed sampling.

For $\BLT_{\Adj}(K)$, most of the energy is concentrated in the first
$K$ frequency components and the graph signal can be approximately
recovered by using those. We consider the following recovery strategy
to estimate those frequency components by projecting samples onto the bandlimited space spanned by the first $K$
  frequency components.

\begin{myAlg}
  \label{alg:subsampleProj}
  We sample a graph signal $m$ times. Each time, we independently
  choose a node $\M_j = i$, $j = 1, \ldots, m$, with probability
  $\pi_{i}$, and take a measurement $y_{\M_j}$. Let $\Psi \in \R^{m
    \times N}$ be the sampling operator, $\Vm_{(K)} \in \R^{N
      \times K}$ be the first $K$ columns of $\Vm$, and $\Dd \in
  \R^{N \times N}$ be a diagonal rescaling matrix with $\Dd_{i,i} =
  1/\sqrt{ m \pi_i }$. We recover the original graph signal by
    solving the following optimization problem: 
  \begin{eqnarray}
    \x^*_{\rm SP} & = &  \Vm_{(K)} \arg \min_{\widehat{\x}_{(K)}}  \left\|  \Psi^T \Psi \Dd^2  \Psi^T \y - \Vm_{(K)} \widehat{\x}_{(K)} \right\|_2^2
     \nonumber
    \\ 
    \label{eq:SubsampleProj}
    \\
     \label{eq:SubsampleProj_sol}
    & = &  \Vm_{(K)} \Um_{(K)} \Psi^T \Psi  \Dd^2  \Psi^T \y,
  \end{eqnarray}
  where $\y$ is a vector representation of the
  samples~\eqref{eq:sampling}.
\end{myAlg}

We call Algorithm~\ref{alg:subsampleProj}~\emph{sampled projection}
(SampleProj) because we reweight and project the samples onto the
bandlimited space. $\Psi^T \y \in \R^N$ rescales the sampled signal $\y \in \R^m$ through zero padding. The rescaling matrix $\Dd^2$ compensates for non-uniform weights in sampling and equalizes samples from different sampling scores. The term $\Psi^T \Psi \Dd^2 \Psi^T \y$ denotes the equalized and zero-padding samples.  The objective function~\eqref{eq:SubsampleProj} minimizes the distance between our estimate and the samples. The intuition behind the solution~\eqref{eq:SubsampleProj_sol} is that we project the equalized and zero-padding samples onto a bandlimited space spanned by $\Vm_{(K)}$ to remove aliasing. The term $\Um_{(K)}  \Psi^T \Psi \Dd^2   \Psi^T \y$ is an unbiased estimator for the first $K$
frequency components, which will be shown later. Thus, in expectation,
the recovered graph signal~$\x^*_{\rm SP} $ is a linear approximation
of any graph signal~$\x$. When $\x$ is bandlimited, $\x^*_{\rm SP} $
perfectly recovers $\x$; when $\x$ is approximately bandlimited,
$\x^*_{\rm SP} $ has a bias due to the existence of
high-frequency components.
 
It is also intuitive to consider the following recovery
strategy~\cite{MaMY:15},
\begin{eqnarray}
  \label{eq:SubsampleLS}
  \nonumber
  \x^*_{\rm LS}	& = &  
  \Vm_{(K)}   \arg \min_{ \widehat{\x}_{(K)} }  
  \left\| \Dd \Psi^T \y -  \Dd \Psi^T \Psi \Vm_{(K)} \widehat{\x}_{(K)} \right\|_2^2
  \nonumber
  \\
  \nonumber
  & = &  \Vm_{(K)} \left( \Dd^2 \Psi^T \Psi \Vm_{(K)} \right)^{\dagger} \Dd \Psi^T \y,
\end{eqnarray}
where we fit the sampled elements of the recovered graph signal to the
samples by solving the least squares problem and
  $(\cdot)^{\dagger}$ is the pseudo-inverse
  operator~\cite{VetterliKG:12}. The corresponding expected recovered
graph signal is
\begin{displaymath}
  \mathbb{E} \x^*_{\rm LS} = \Vm_{(K)}  \widehat{\x}_{(K)} + \Vm_{(K)} \Pj \Vm_{(-K)}  \widehat{\x}_{(-K)},
\end{displaymath}
where $\Pj = \left( \Dd^2 \Psi^T \Psi \Vm_{(K)} \right)^{\dagger} \Dd$
$\Psi^T \Psi = \left(\Um_{(K)} \Psi^T \Psi \Dd^2 \Psi^T \Psi \Vm_{(K)}
\right)^{-1}$ $\Um_{(K)} \Psi^T \Psi \Dd^2 \Psi^T \Psi $,
  $\Vm_{(-K)} $ chooses the last $K$ columns of $\Vm$, and
  $\xhat_{(-K)}$ chooses the last $K$ columns of $\x$. When $\x$ is
bandlimited, $\x^*_{\rm LS}$ perfectly recovers $\x$; when $\x$ is
approximately bandlimited, however, $\x^*_{\rm LS}$ is a mixture of
low-frequency and high-frequency components and it is hard to show its
statistical properties. Moreover, it is less computationally
  efficient to compute $\x^*_{\rm LS}$ than $\x^*_{\rm SP}$ because
of the presence of the inverse term.

Since the definitions of the bandlimited class and the
approximately bandlimited class, and the proposed sampling strategies
are all based on the graph Fourier transform instead of on graph
shift, all the proposed methods work for other
versions of the graph Fourier transform as well.

\subsection{Statistical Analysis}
\label{sec:statistical}
We study the statistical properties of
Algorithm~\ref{alg:subsampleProj} by providing the bias, covariance,
mean square error (MSE), and an optimal sampling distribution.

The following lemma shows that the
sampled projection estimator is an unbiased estimator of the first $K$
frequency components for any sampling scores.
\begin{myLem} 
  \label{lem:unbias}
  The sampled projection estimator with bandwidth $K$ and arbitrary
  sampling scores is an unbiased estimator of the first $K$ frequency
  components, that is,
  \begin{equation*}
    \mathbb{E} \x^* \ = \ \Vm_{(K)} \Um_{(K)} \x,
    \quad \text{for all } \x,
  \end{equation*}
  where $\x^*$ is the solution of Algorithm~\ref{alg:subsampleProj}.
\end{myLem}

Lemma~\ref{lem:covariance} gives the exact covariance of the sampled
projection estimator.
\begin{myLem} 
  \label{lem:covariance}
  The covariance of sampled projection estimator $\x^*$ has the
  following property:
  \begin{eqnarray*}
    {\rm  Tr} ( {\rm Covar}  \left[  \x^*   \right] ) 
    & = &  \mathbb{E}  \left\|  \x^* - \mathbb{E}   \left[   \x^*  \right]   \right\|_2^2 
    \\
    & = &  {\rm  Tr} \left( \Um_{(K)}  \W_C \Vm_{(K)} \right) - \frac{1}{m} \left\|  \widehat{\x}_{(K)} \right\|_2^2,
  \end{eqnarray*}
  where Tr($\cdot$) is the trace operator and $\W_C$ is a
  diagonal matrix with $\left(\W_C \right)_{i,i}= (x_{i}^2 +
  \sigma^2)/(m \pi_{i}) $.
\end{myLem}

Theorem~\ref{thm:mse} shows the exact MSE of sampled projection
estimator and an upper bound.
\begin{myThm}
  \label{thm:mse}
  For $\x \in \BLT_{\Adj}(K, \beta, \mu)$, let $\x^*$ be the sampled
  projection estimator with bandwidth $\kappa \geq K$. Then,
  \begin{eqnarray}
    && \mathbb{E} \left\|  \x^* - \x  \right\|_2^2  
   \nonumber \\ \nonumber
    & = &   \left\| \Vm_{(-\kappa)} \widehat{\x}_{(-\kappa)} \right\|_2^2 +   {\rm Tr} \left( \Um_{(\kappa)}  \W_C \Vm_{(\kappa)} \right) - \frac{1}{m} \left\|  \widehat{\x}_{(\kappa)} \right\|_2^2 
    \\
    \label{eq:MSE_upper}
    & \leq  & \frac{\mu }{1 + \kappa^{2\beta}}  \left\| \x \right\|_2^2 + {\rm Tr} \left( \Um_{(\kappa)}  \W_C \Vm_{(\kappa)} \right).
  \end{eqnarray}
\end{myThm}
We merge the proofs of Lemmas~\ref{lem:unbias},~\ref{lem:covariance}
and Theorem~\ref{thm:mse} in Appendix B in the supporting
document. The main idea follows from the bias-variance tradeoff.  The
bias term $\mu/(1 + \kappa^{2\beta}) \left\| \x \right\|_2^2$ comes
from the high-frequency components and ${\rm Tr} \left( \Um_{(\kappa)}
  \W_C \Vm_{(\kappa)} \right)$ comes from the covariance. The
  last inequality in Theorem~\ref{thm:mse} comes from relaxing the
  bias term and omitting a constant, which has nothing to do with the
  sampling scores. Thus, optimizing the upper bound over sampling
  scores is equivalent to optimizing the exact MSE.

We next study the MSEs of sampled projection estimator based on both
uniform sampling and experimentally designed sampling.
\begin{myCorollary}
  \label{cor:random}
  The upper bound of MSE of uniform sampling is
  \begin{eqnarray*}
    \mathbb{E} \left\| \x^* - \x \right\|_2^2
    & \leq & \frac{\mu }{1 + \kappa^{2\beta}}  \left\| \x \right\|_2^2 + \frac{N}{m} \sum_{k=1}^{\kappa} \sum_{i=1}^N \Um_{k,i}^2  \left( x_i^2 + \sigma^2 \right).
\end{eqnarray*}
\end{myCorollary}
The upper bound just specifies that the sampling score be $\pi_i =
1/N$ for all $i$. For experimentally designed sampling, we are allowed
to study the graph structure and design sample indices. In the
following corollary, we show a set of optimal sampling scores of the
sampled projection estimator by minimizing the upper bound of MSE.

\begin{myCorollary}
  \label{cor:optimal}
  The optimal sampling score of the sampled projection estimator with
  bandwidth $\kappa \geq K$ is
  \begin{eqnarray*}
    \pi_i & \propto & \sqrt{  \left( \sum_{k=1}^{\kappa}\Um_{k,i}^2 \right)  \left( x_i^2 + \sigma^2 \right)  }.
  \end{eqnarray*}
  The corresponding upper bound of MSE is
  \begin{eqnarray*}
    && \mathbb{E} \left\|     \x^* - \mathbb{E} \x^*  \right\|_2^2 
    \\
    & \leq & \frac{\mu }{1 + \kappa^{2\beta}}  \left\| \x \right\|_2^2 +  \frac{1}{m} \left( \sum_{i=1}^N \sqrt{  \sum_{k=1}^{\kappa} \Um_{k,i}^2  \left( x_i^2 + \sigma^2 \right)  }  \right)^2.
  \end{eqnarray*}
\end{myCorollary}

\begin{proof}
  To obtain the optimal sampling score for the estimator, we minimize
  the MSE given in Theorem~\ref{thm:mse} and solve the following
  optimization problem.
  \begin{eqnarray*}
    &&  \min_{\pi_{i}}  {\rm Tr} \left( \Um_{(\kappa)}  \W_C \Vm_{(\kappa)} \right),  
    \\
    && {\rm subject~to}~ \sum_{i}  \pi_{i} = 1, \pi_{i} \geq 0.
  \end{eqnarray*}
  The objective function is the variance term of the MSE derived
    in Theorem~\ref{thm:mse}. Since the bias term has nothing to do
    with the sampling score, minimizing the variance term is
    equivalent to minimizing the MSE. The constraints require $\{
  \pi_i\}_{i=1}^N$ to be a valid probability distribution. The
  Lagrangian function is then
  \begin{eqnarray*}
    L( \pi_{i}, \eta_0, \eta_{i}  ) & = &   \sum_{i=1}^N \left( \frac{ x_i^2 +  \sigma^2}{ m \pi_{i} }  
      \sum_{k=1}^{\kappa} \Um_{k,i}^2  \right)
    \\
    && + \eta_0 \left( \sum_{i=1}^N  \pi_i - 1 \right) + \sum_{i=1}^N  \eta_i \pi_i,
  \end{eqnarray*}
  where $\eta_0, \eta_i$ are Lagrangian multipliers. We set the
  derivative of the Lagrangian function to zero,
  \begin{eqnarray*}
    \frac{d L}{d \pi_i} & = &  -  \frac{ x_i^2 +  \sigma^2 }{ m \pi_l^2 }  
    \sum_{k=1}^{\kappa} \Um_{k,i}^2  + \eta_0 + \eta_i = 0,
  \end{eqnarray*}
  and obtain the optimal sampling score 
  \begin{eqnarray}
    \label{eq:opt_sampling_score}
    \pi_i & \propto & \sqrt{  \left( \sum_{k=1}^{\kappa}\Um_{k,i}^2 \right)  \left( x_i^2 + \sigma^2 \right)  }.
  \end{eqnarray}
  Substituting the optimal sampling score $\pi_i$
  to the upper bound of the MSE~\eqref{eq:MSE_upper},
  \begin{eqnarray*}
    && \mathbb{E} \left\|     \x^* - \mathbb{E} \x^*  \right\|_2^2
    \ \stackrel{(a)}{\leq} \ \frac{\mu }{1 + \kappa^{2\beta}}  \left\| \x \right\|_2^2 
    \\
    &+& \frac{1}{m} \sum_{i=1}^{N}   \frac{  \sum_{k=1}^{\kappa}  \Um_{k,i}^2 (x_{i}^2+\sigma^2 ) }{ \sqrt{  \sum_{k=1}^{\kappa} \Um_{k,i}^2   \left( x_i^2 + \sigma^2 \right)  } } \sum_{i=1}^N \sqrt{  \sum_{k=1}^{\kappa} \Um_{k,i}^2  \left( x_i^2 + \sigma^2 \right)  } 
    \\
    & = & \frac{\mu }{1 + \kappa^{2\beta}}  \left\| \x \right\|_2^2 +  \frac{1}{m} \left( \sum_{i=1}^N \sqrt{  \sum_{k=1}^{\kappa} \Um_{k,i}^2  \left( x_i^2 + \sigma^2 \right)  }  \right)^2,
  \end{eqnarray*}
 where (a) follows from substituting the optimal sampling score~\eqref{eq:opt_sampling_score} into the upper bound of the MSE~\eqref{eq:MSE_upper}. 
\end{proof}

We see that the optimal sampling score includes a trade-off between signal and  noise. In practice, we
cannot access $\x$ and thus need to approximate the ratio between each
$x_i$ and $\sigma^2$. For active sampling, we can collect
 feedback to approximate each signal coefficient $x_i$;
for experimentally designed sampling, we approximate beforehand; one
way is to use the graph structure to sketch the shape of $\x$. We
first use the Cauchy-Schwarz inequality to bound $x_i$,
\begin{eqnarray*}
  |x_i|  & = &  | \vv_i^T \widehat{\x} | 
  \  \leq \  \left\| \vv_{i} \right\|_2  \left\| \widehat{\x} \right\|_2 
  \ = \  \left\| \vv_{i} \right\|_2  \left\| \x \right\|_2,
\end{eqnarray*}
where $\vv_{i} $ is the $i$th row of $\Vm$. Thus, in the upper bound,
we have a tradeoff between signal and noise: when
the signal-to-noise ratio (SNR) $\left\| \x \right\|_2/\sigma^2$ is
small, the approximate optimal sampling score is
\begin{eqnarray*}
  \pi_i & \propto &  \left\| \vv_{i, (K)} \right\|_2 \ = \ \sqrt{ \sum_{k=1}^{\kappa}\Um_{k,i}^2 },
\end{eqnarray*}
which is the square root of the leverage score of $\Vm_{(K)}$. 

When the SNR $\left\| \x \right\|_2/\sigma^2$ is large, the
approximate optimal sampling score is
\begin{eqnarray*}
  \pi_i & \propto & \sqrt{  \left( \sum_{k=1}^{\kappa}\Um_{k,i}^2 \right)  \left\| \vv_{i} \right\|_2^2  }.
\end{eqnarray*}

For approximately bandlimited signals, when $\beta$ is large and $\mu$
is small, the main energy concentrates in the first $K$ frequency
components, $\left\| \vv_{i} \right\|_2$ is concentrated in $\left\|
  \vv_{i, (K)} \right\|_2$, where $\vv_{i, (K)}$ is the first $K$
elements in $\vv_{i}$. Specifically,
\begin{eqnarray*}
  |x_i|  & = &  | \vv_i^T \widehat{\x} | = | \vv_{i, (K)}^T \widehat{\x}_{(K)} + \vv_{i, (-K)}^T \widehat{\x}_{(-K)}  |
  \\
  &  \leq &  \left\| \vv_{i, (K)} \right\|_2  \left\| \widehat{\x}_{(K)}  \right\|_2 +  \left\| \vv_{i, (-K)} \right\|_2  \left\| \widehat{\x}_{(-K)}  \right\|_2
  \\
  &  \leq &  \left\| \vv_{i, (K)} \right\|_2  \left\| \widehat{\x}_{(K)}  \right\|_2 +  \left\| \vv_{i, (-K)} \right\|_2   \sqrt{ \frac{ \mu }{1 + K^{2\beta}} } \left\| \x \right\|_2
  \\
  &  \approx &  \left\| \vv_{i, (K)} \right\|_2  \left\| \widehat{\x}_{(K)}  \right\|_2.
\end{eqnarray*}
In this case, when the SNR $\left\| \x \right\|_2/\sigma^2$ is large,
the approximate optimal sampling score is
\begin{eqnarray*}
  \pi_i & \propto &   \left\| \vv_{i, (K)} \right\|_2^2  \ = \   \sum_{k=1}^{\kappa}\Um_{k,i}^2,
\end{eqnarray*}
which is the leverage score of $\Vm_{(K)}$.

\subsection{Relation to the Sampling Theory on Graphs}
Sampling theory on graphs considers a bandlimited graph signal under
the experimentally designed sampling~\cite{ChenVSK:15}. It shows that
for a noiseless bandlimited graph signal, 
experimentally designed sampling guarantees perfect recovery while
uniform sampling cannot, which also implies that active sampling
cannot perform better than experimentally designed sampling. The
recovery strategy is to solve the following optimization problem,
\begin{eqnarray}
  \label{eq:st}
  \x^*_{\rm ST} = \arg \min_{\x \in \BL_{\Adj}(K) } \,  \left\|  \Psi \x -  \y \right\|_2^2 
%  \\ \nonumber
  = \Vm_{(K)}  (\Psi  \Vm_{(K)}  )^{\dagger} \y,
\end{eqnarray}
where $\Psi$ is the sampling operator~\eqref{eq:Psi} and $\y$ is a
vector representation of the samples~\eqref{eq:sampling}.  When
  the original graph signal is bandlimited, the
  estimator~\eqref{eq:st} is unbiased and its MSE comes solely from
  the variance term caused by noise, that is,
\begin{eqnarray*}
  && \mathbb{E} \left\| \x^*_{\rm ST} - \x \right\|_2^2  
  \\ 
  & = &   \mathbb{E} \left\|  \Vm_{(K)}  (\Psi  \Vm_{(K)}  )^{\dagger} (\Psi \x + \epsilon)   - \x \right\|_2^2  
  \\
  & =  & \mathbb{E} \left\|  \Vm_{(K)}  (\Psi  \Vm_{(K)}  )^{\dagger}  \epsilon \right\|_2^2 \ = \   \mathbb{E} \left\| (\Psi  \Vm_{(K)}  )^{\dagger}  \epsilon \right\|_2^2.
\end{eqnarray*}
In~\cite{ChenVSK:15}, the authors propose an optimal sampling operator
based on minimizing $\left\| (\Psi \Vm_{(k)} )^{\dagger}
\right\|_2^2$. The sample set is deterministic and the procedure is
computationally efficient when the sample size is small. Compared with
sampling score-based sampling, however, the optimal sampling operator
is computationally inefficient when the sample size is large. When the
original graph signal is not bandlimited, similarly to $\x^*_{\rm
  LS}$, the solution of~\eqref{eq:st} is biased, that is,
\begin{eqnarray*}
  \mathbb{E} \x^* & = & \Vm_{(K)}  (\Psi  \Vm_{(K)}  )^{\dagger} \mathbb{E} (\Psi \x + \epsilon)
  \\ 
  & = & \Vm_{(K)}  (\Psi  \Vm_{(K)}  )^{\dagger} \Psi \left(  \Vm_{(K)} \xhat_{(K)}  + \Vm_{(-K)} \xhat_{(-K)}  \right)
  \\ 
  & = & \Vm_{(K)}  \xhat_{(K)}  + \Vm_{(K)}  (\Psi  \Vm_{(K)}  )^{\dagger}  \Vm_{(-K)} \xhat_{(-K)}.
\end{eqnarray*}
 We see that the high-frequency components $\Vm_{(-K)} \xhat_{(-K)}$ are projected on the low-frequency space spanned by $\Vm_{(K)}$, which causes aliasing.

% -----------------------------------------
\section{Optimal Rates of Convergence}
\label{sec:convergence}
% -----------------------------------------
In this section, we introduce two types of graphs and show that the
proposed recovery strategies achieve the optimal rates of convergence
on both. Type-1 graphs model regular graphs, where the
  corresponding graph Fourier bases are not sparse and elements have
  similar magnitudes; examples are circulant graphs and
  nearest-neighbor graphs~\cite{SandryhailaKP:11a}.  Since the energy
  is evenly spread over the graph Fourier basis, each element contains
  similar amount of information; for such graphs, experimentally
  designed sampling performs similarly to uniform sampling. Type-2
  graphs model irregular graphs, where the corresponding graph Fourier
  bases are sparse and elements do not have similar magnitudes;
  examples are small-world graphs and scale-free
  graphs~\cite{Newman:10}. Since the energy is concentrated in a few
  elements in the graph Fourier basis, these elements are more
  informative and should be selected. For such graphs, experimentally
  designed sampling outperforms uniform sampling.

% -----------------------------------------
\subsection{Type-1 Graphs}
% -----------------------------------------
\begin{defn}
  \label{df:regular}
  A graph $\Adj \in \R^{N \times N}$ is of~\emph{Type-1}, when its
  graph Fourier basis satisfies
  \begin{equation*}
    | \Vm_{i,j} | = O(N^{-1/2}),~~~{\rm for~all}~ i,j = 0, 1, \cdots, N-1.
  \end{equation*}
\end{defn}
For a Type-1 graph, elements in $\Vm$ have roughly similar magnitudes,
that is, the energy is evenly spread over $\Vm$.

Based on Theorem~\ref{thm:mse}, we can specify the parameters for a
Type-1 graph and show the following results.
\begin{myCorollary}
  \label{cor:type1_upp}
  Let $\Adj \in \R^{N \times N}$ be a Type-1 graph, for the class
  $\BLT_{\Adj }(K, \beta, \mu)$.
  \begin{itemize}
  \item Let $\x^*$ be the sampled projection estimator with the
    bandwidth $\kappa \geq K$ and uniform sampling; we have
    \begin{eqnarray*}
      \mathbb{E}  \left(  \frac{  \left\| \x^* - \x \right\|_2^2  } { \left\| \x \right\|_2^2 }   \right) \ \leq \  C \, m^{-\frac{2\beta}{2\beta+1}},
    \end{eqnarray*}
    where $C$ is a positive constant,  $m$ is the number of samples~\eqref{eq:Psi},
      $\beta$ is the spectral decay factor in the approximately
      bandlimited class~\eqref{eq:BLT}, and the rate is achieved when
    $\kappa$ is of the order of $m^{1/(2\beta+1)}$ and
    upper bounded by $N$;
  \item Let $\x^*$ be the sampled projection estimator with the
    bandwidth $\kappa \geq K$ and the optimal sampling score in
    Corollary~\ref{cor:optimal}; we have
    \begin{eqnarray*}
      \mathbb{E}  \left(  \frac{  \left\| \x^* - \x \right\|_2^2  } { \left\| \x \right\|_2^2 }   \right)  \ \leq \ C \, m^{-\frac{2\beta}{2\beta+ 1}},
    \end{eqnarray*}
    where $C$ is a positive constant, and the rate is achieved when $\kappa$ is of the order of $m^{1/(2\beta+1)}$ and upper
    bounded by $N$.
  \end{itemize}
\end{myCorollary}
When $m \gg N$, we set $\kappa = N$, the bias term is then zero, and
both upper bounds are actually $C \, m^{-1}$. We see that uniform
sampling and optimal sampling score based sampling have the same
convergence rate, that is, experimentally designed sampling does not
perform better than uniform sampling for the Type-1 graphs.

Based on Theorem~\ref{thm:lower} and Corollary~\ref{cor:type1_upp}, we
conclude the following.
\begin{myCorollary}
  \label{cor:type1_opt}
  Let $\Adj \in \R^{N \times N}$ be a Type-1 graph, for the class
  $\BLT_{\Adj }(K, \beta, \mu)$.
  \begin{itemize}
  \item Under uniform sampling,
    \begin{eqnarray*}
      c \, m^{-\frac{2\beta}{2\beta+1}} \ \leq \ ~~~~~~~~~~~~~~~~~~~~~~~~~~~~~~~~~~~~~~~~
      \\
      \inf_{(\x^*, \M) \in \Theta_{\rm u} } \sup_{\x \in \BLT_{\Adj}(K, \beta, \mu)} \mathbb{E}_{\x, \M} \left(  \left\| \x^* - \x \right\|_2^2 \right) 
      \\
      \ \leq \  C \, m^{-\frac{2\beta}{2\beta+1}},
    \end{eqnarray*}
    where constant $C > c > 0$, and the rate is achieved when $\kappa$
    is  of the order of $m^{1/(2\beta+1)}$ and upper bounded by $N$;
  \item Under experimentally designed sampling,
    \begin{eqnarray*}
      c  \, m^{-\frac{2\beta}{2\beta+1}} \ \leq \
      ~~~~~~~~~~~~~~~~~~~~~~~~~~~~~~~~~~~~~~~~
      \\
      \inf_{(\x^*, \M) \in \Theta_{\rm e} } \sup_{\x \in \BLT_{\Adj} (K, \beta, \mu)} \mathbb{E}_{\x, \M} \left(  \left\| \x^* - \x \right\|_2^2 \right) 
      \\
      \ \leq \  C  \, m^{-\frac{2\beta}{2\beta+1}},
    \end{eqnarray*}
    where constant $C > c > 0$, and the rate is achieved when $\kappa$
    is  of the order of $m^{1/(2\beta+1)}$ and upper bounded by $N$.
  \end{itemize}
\end{myCorollary}
We merge the proofs of Corollaries~\ref{cor:type1_upp}
and~\ref{cor:type1_opt} in Appendix C in the supporting document.

We see that under both random and
experimentally designed sampling, the lower and
  upper bounds have the same rate of convergence, which achieves the
  optimum. In addition, random and experimentally designed sampling
  have the same optimal rate of convergence and we can thus conclude
  that experimentally designed sampling does not perform
  asymptotically better than uniform sampling for the Type-1
  graphs. The sampled projection estimator attains the optimal
  rate~of~convergence.

\subsection{Type-2 Graphs}
\begin{defn}
  \label{df:irregular}
  A graph $\Adj \in \R^{N \times N}$ is of~\emph{Type-2}, when its
  graph Fourier basis satisfies
  \begin{eqnarray*}
    | \vv_{i, T}  | = O(1), {\rm and}~| \vv_{i, T+1}  | \ll O(1),~{\rm for~all}~ i = 0, \cdots, N-1,
  \end{eqnarray*}
  where $\vv_{i, T}$ is the $T$th largest elements in the $i$th column
  of $\Vm$ and $T \ll N$ is some constant.
 \end{defn}
 A Type-2 graph requires that each column vector of $\Vm$ be
 approximately sparse. When we take a few columns from $\Vm$ to form a
 submatrix, the energy in the submatrix concentrates in a few rows of
 the submatrix. This is equivalent to the sampling scores being
 approximately sparse. Simulations show that star graphs, scale-free
 graphs and small-world graphs approximately fall into this type of
 graphs.

 Based on Theorem~\ref{thm:mse}, we can specify the parameters for a
 Type-2 graph and show the following results.
 \begin{myCorollary}
   \label{cor:type2_upp}
   Let $\Adj \in \R^{N \times N}$ be a Type-2 graph, for the class
   $\BLT_{\Adj }(K, \beta, \mu)$.
   \begin{itemize}
   \item Let $\x^*$ be the sampled projection estimator with the
     bandwidth $\kappa \geq K$ and uniform sampling; we have
     \begin{eqnarray*}
       \mathbb{E}  \left( \frac{ \left\| \x^* - \x \right\|_2^2 } { \left\| \x \right\|_2^2  } \right) \ \leq \  C \, m ^{-\frac{2\beta}{2\beta+\gamma}},
     \end{eqnarray*}
     where $C$ is a positive constant, and the rate is achieved when $\kappa$ is of the order of $m^{1/(2\beta+\gamma)}$, $\gamma =
     \log(N)/\log(\kappa) \geq 1$;

   \item Let $\x^*$ be the sampled projection estimator with the
     bandwidth $\kappa \geq K$ and optimal sampling score based
     sampling; we have
     \begin{eqnarray*}
       \mathbb{E} \left(  \frac{ \left\| \x^* - \x \right\|_2^2 } { \left\| \x \right\|_2^2  }  \right) \ \leq \  C  \, m^{-\frac{2\beta}{2\beta+1}},
     \end{eqnarray*}
     where  $C$ is a positive constant, the rate is achieved when $\kappa$ is of the order of $m^{1/(2\beta+1)}$ and upper
     bounded by $N$.
   \end{itemize}
\end{myCorollary}

Based on Theorem~\ref{thm:lower} and Corollary~\ref{cor:type2_upp}, we
conclude the following.
\begin{myCorollary}
  \label{cor:type2_opt}
  Let $\Adj \in \R^{N \times N}$ be a Type-2 graph with parameter
  $K_0$, for the class $\BLT_{\Adj }(K, \beta, \mu)$.
  \begin{itemize}
  \item Under uniform sampling,
    \begin{eqnarray*}
      c \, m^{-\frac{2\beta}{2\beta+1}} \ \leq \ ~~~~~~~~~~~~~~~~~~~~~~~~~~~~~~~~~~~~~~~~
      \\
      \inf_{(\x^*, \M) \in \Theta_{\rm u} } \sup_{\x \in \BLT_{\Adj}(K, \beta, \mu)} \mathbb{E}_{\x, \M} \left( \frac{ \left\| \x^* - \x \right\|_2^2 } {  \left\| \x \right\|_2^2  } \right) 
      \\
      \ \leq \ C \, m^{-\frac{2\beta}{2\beta+\gamma}},
    \end{eqnarray*}
    where constant $C > c > 0$, and the rate is achieved when $\kappa$
    is of the order of $m^{1/(2\beta+\gamma)}$ and $\gamma =
    \log(N)/\log(\kappa) \geq 1$;

  \item Under experimentally designed sampling, there exists a $\gamma
    > 1$,
    \begin{eqnarray*}
      c  \, m  ^{-\frac{2\beta}{2\beta+1}} \ \leq \ 
      ~~~~~~~~~~~~~~~~~~~~~~~~~~~~~~~~~~~~~~~~
      \\
      \inf_{(\x^*, \M) \in \Theta_{\rm e} } \sup_{\x \in \BLT_{\Adj} (K, \beta, \mu)} \mathbb{E}_{\x, \M} \left(   \frac{ \left\| \x^* - \x \right\|_2^2 } {  \left\| \x \right\|_2^2  }   \right)
      \\
      \ \leq \  C  \ m  ^{-\frac{2\beta}{2\beta+1}},
    \end{eqnarray*}
    where $C$ is a positive constant, the rate is achieved when $\kappa$ is  of
    the order of $m^{1/(2\beta+1)}$ and upper bounded by $N$.
  \end{itemize}
\end{myCorollary}
We merge the proofs of Corollaries~\ref{cor:type2_upp}
and~\ref{cor:type2_opt} in Appendix D in the supporting document.

We see that under both uniform and experimentally
designed sampling, the lower and upper bounds have
the same rate of convergence, which achieves the optimum. However,
experimentally designed sampling has a larger optimal rate of
convergence, and we can thus conclude
that experimentally designed sampling exhibits asymptotically better
performance than uniform sampling for a Type-2 graph. The sampled
projection estimator attains the optimal rate of convergence.

\section{Experimental Results}
\label{sec:experiments}
In this section, we validate the proposed recovery strategy on five
specific graphs: a ring graph, an Erd\H{o}s-R\'enyi graph, a random
geometric graph, a small-world graph and a power-law graph. Based on
the graph structure, we roughly label each as a Type-1 or Type-2, and
then, for each, we compare the empirical performance of the proposed
recovery strategy based on uniform and experimentally designed
sampling. For experimentally designed sampling, we use both the
leverage score of $\Vm_{(K)}$ (approximately optimal in the noiseless
case) and the square root of the leverage score of $\Vm_{(K)}$
(approximately optimal in the noisy case). In the experiments, the
graph Fourier basis is the eigenvector matrix of the
adjacency matrix. We find similar results when the graph Fourier
basis is the eigenvector matrix of the graph Laplacian
matrix.

\subsection{Simulated Graphs}
\label{sec:simulatedgraphs}
We now introduce the five graphs, each with $10,000$ nodes, used
in our experiments.

\mypar{Ring graph with $k$-nearest neighbors} A ring graph is a graph
where each node connects to its $k$-nearest neighbors. We use a ring
graph where each node connects to exactly four nearest neighbors. The
eigenvectors are similar to the discrete cosine transform and the
energy evenly spreads to each element in
$\Vm$~\cite{SandryhailaKP:11a}; this is thus a Type-1
  graph. Figure~\ref{fig:Ring} illustrates some properties of
the ring graph:    Figure~\ref{fig:Ring}(a) shows the graph plot (for easier
visualization, only $20$ nodes are shown and with enough zoom, one can clearly
  see that each node connects to exactly four nearest
  neighbors;  Figure~\ref{fig:Ring}(b)
shows the histogram of the degrees that
concentrate on 4, as expected;  and
Figure~\ref{fig:Ring}(c) shows the histogram
of the leverage scores of $\Vm_{(20)}$, which are the optimal sampling
scores when the SNR is large. Note that we set the bandwidth $K = 20$
to show the low-frequency band of the graph Fourier transform
matrix. We see that the leverage scores concentrate
around $10^{-4}$; this means that each node has the same probability
to be chosen and uniform sampling is approximately the optimal
sampling strategy.

\begin{figure*}[htb]
  \begin{center}
    \begin{tabular}{ccc}
      \includegraphics[width=0.35\columnwidth]{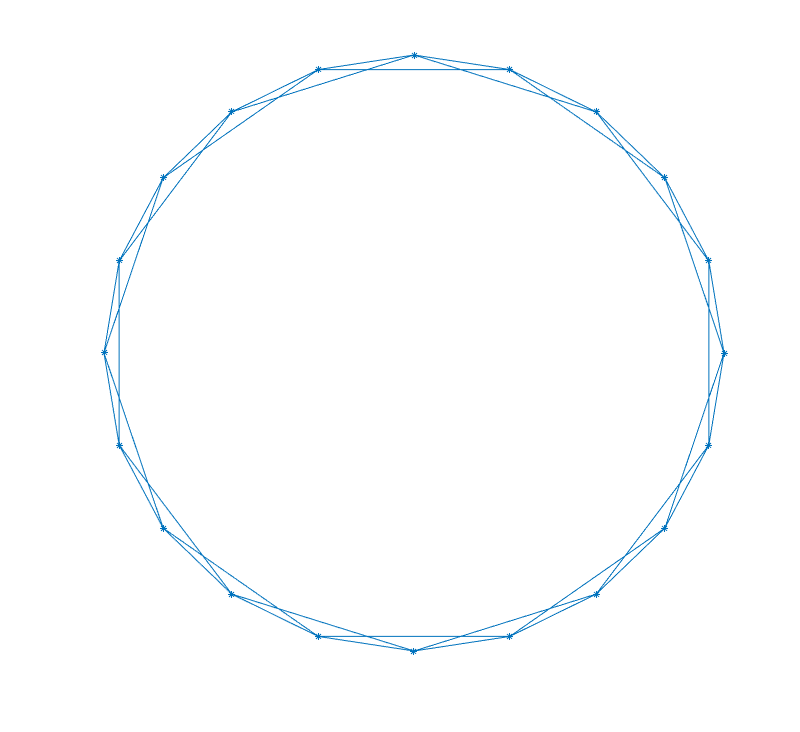}  & \includegraphics[width=0.35\columnwidth]{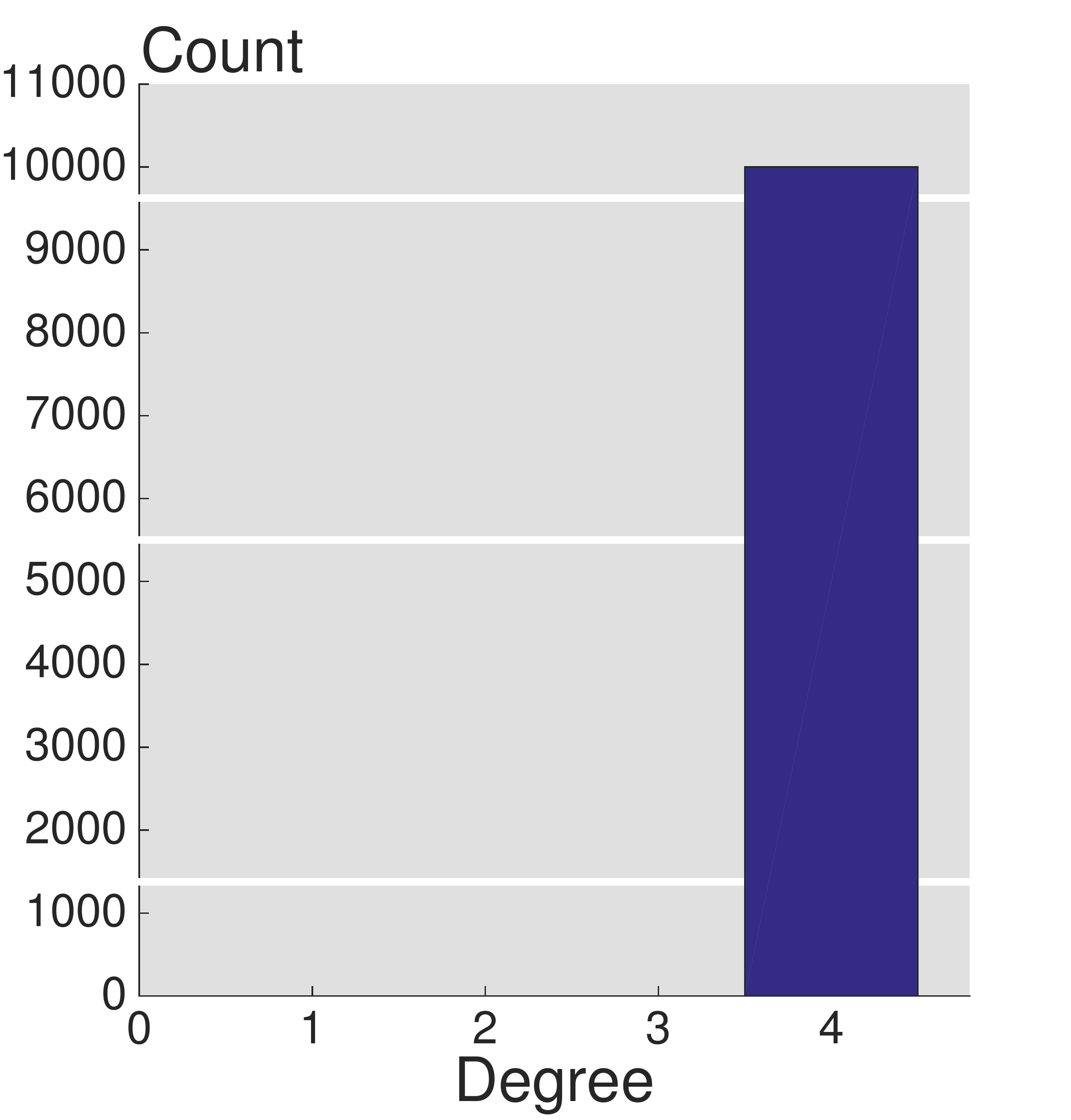}  &
      \includegraphics[width=0.35\columnwidth]{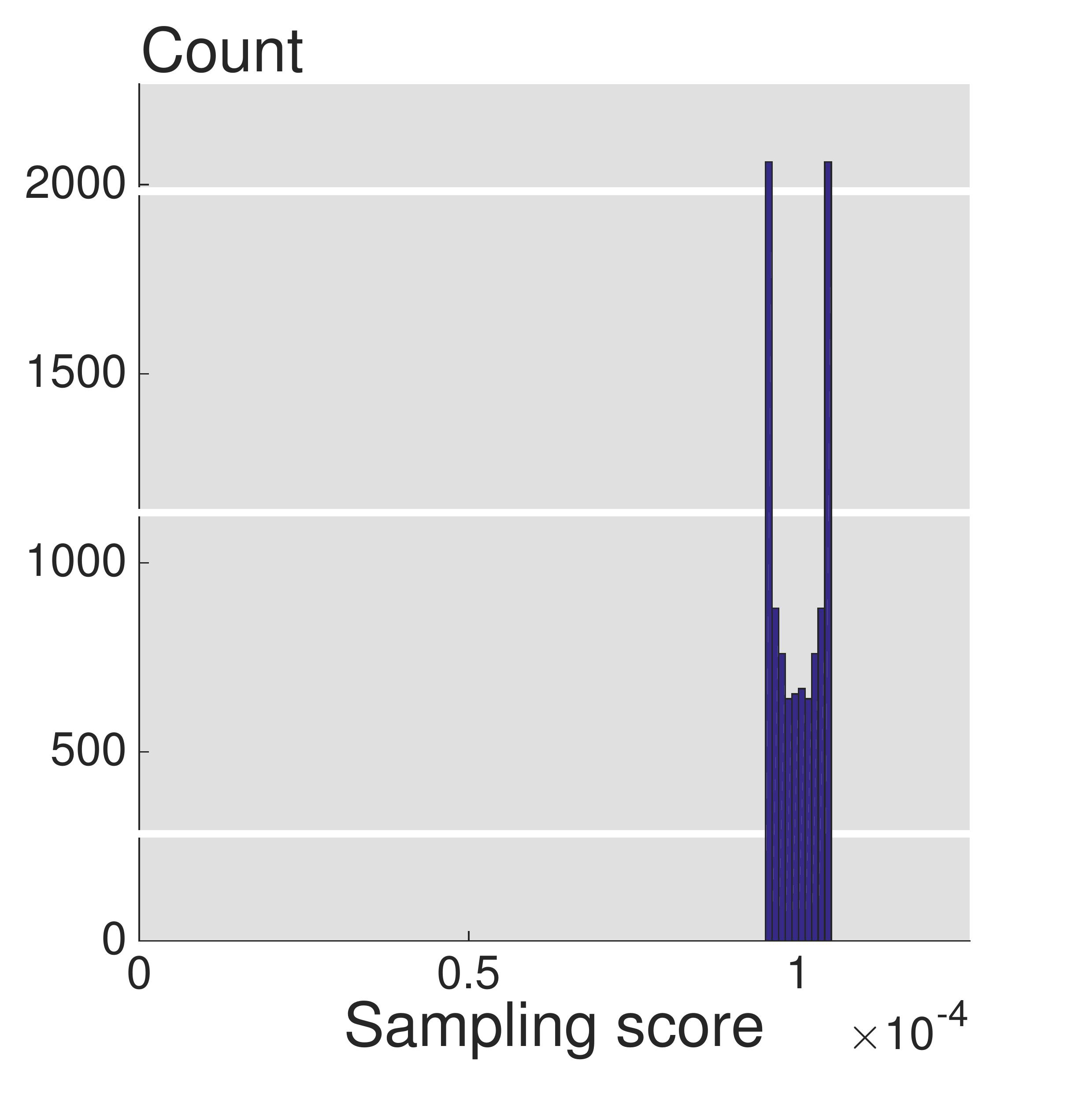} 
      \\
      {\small (a) Graph plot.} & {\small (b) Histogram of the degrees.} &  {\small (c) Histogram of the leverage scores.}
      \\  
    \end{tabular}
  \end{center}
  \caption{\label{fig:Ring} Properties of a ring graph. Plot
      (c) shows the histogram of the leverage scores, which are optimal sampling scores when the SNR is large.}
\end{figure*}

\mypar{Erd\H{o}s-R\'enyi graph} An Erd\H{o}s-R\'enyi graph is a random
graph where each pair of nodes is connected with some
probability~\cite{Newman:10}.  We use a graph where each pair of nodes
is connected with probability of $0.01$, that is, each node has 100
neighbors in expectation. Figure~\ref{fig:ER}
  illustrates some properties of the Erd\H{o}s-R\'enyi
graph.   Figure~\ref{fig:ER}(a) shows the graph plot (for easier
  visualization, only 100 nodes are shown); Figure~\ref{fig:ER}(b)
  shows the histogram of the degrees that concentrate around 100, as
  expected; and Figure~\ref{fig:ER}(c) shows the histogram of the
  leverage scores of $\Vm_{(20)}$, which are the optimal sampling scores
  when the SNR is large.  We see that the
leverage scores  concentrate around $10^{-4}$;
  this means that each node has the same probability to be chosen and
  uniform sampling is approximately the optimal sampling strategy,
  similarly to the ring graph.  Based on the above, an
  Erd\H{o}s-R\'enyi graph is approximately a Type-1 graph.

\begin{figure*}[htb]
  \begin{center}
    \begin{tabular}{ccc}
      \includegraphics[width=0.35\columnwidth]{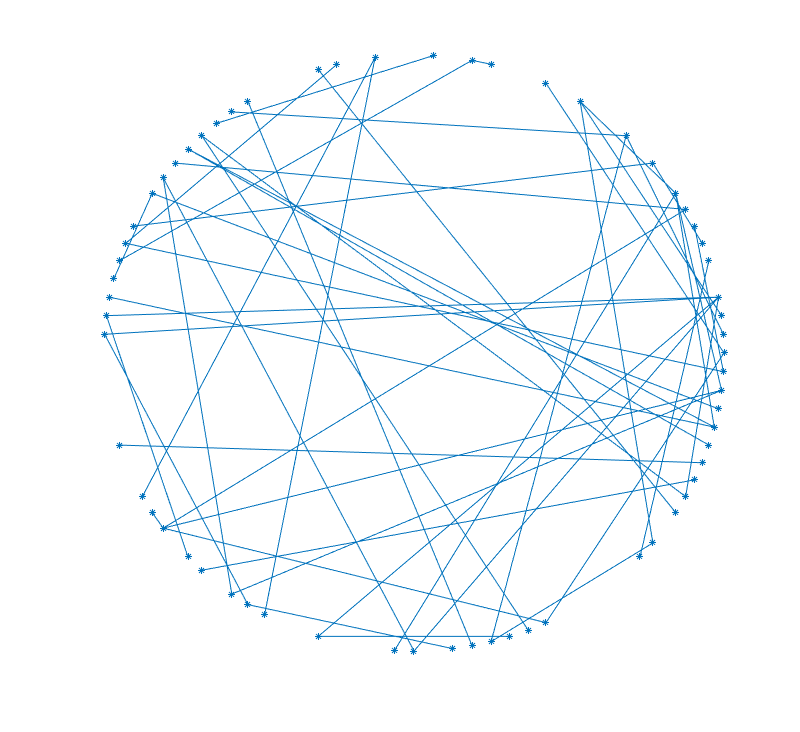}  & \includegraphics[width=0.35\columnwidth]{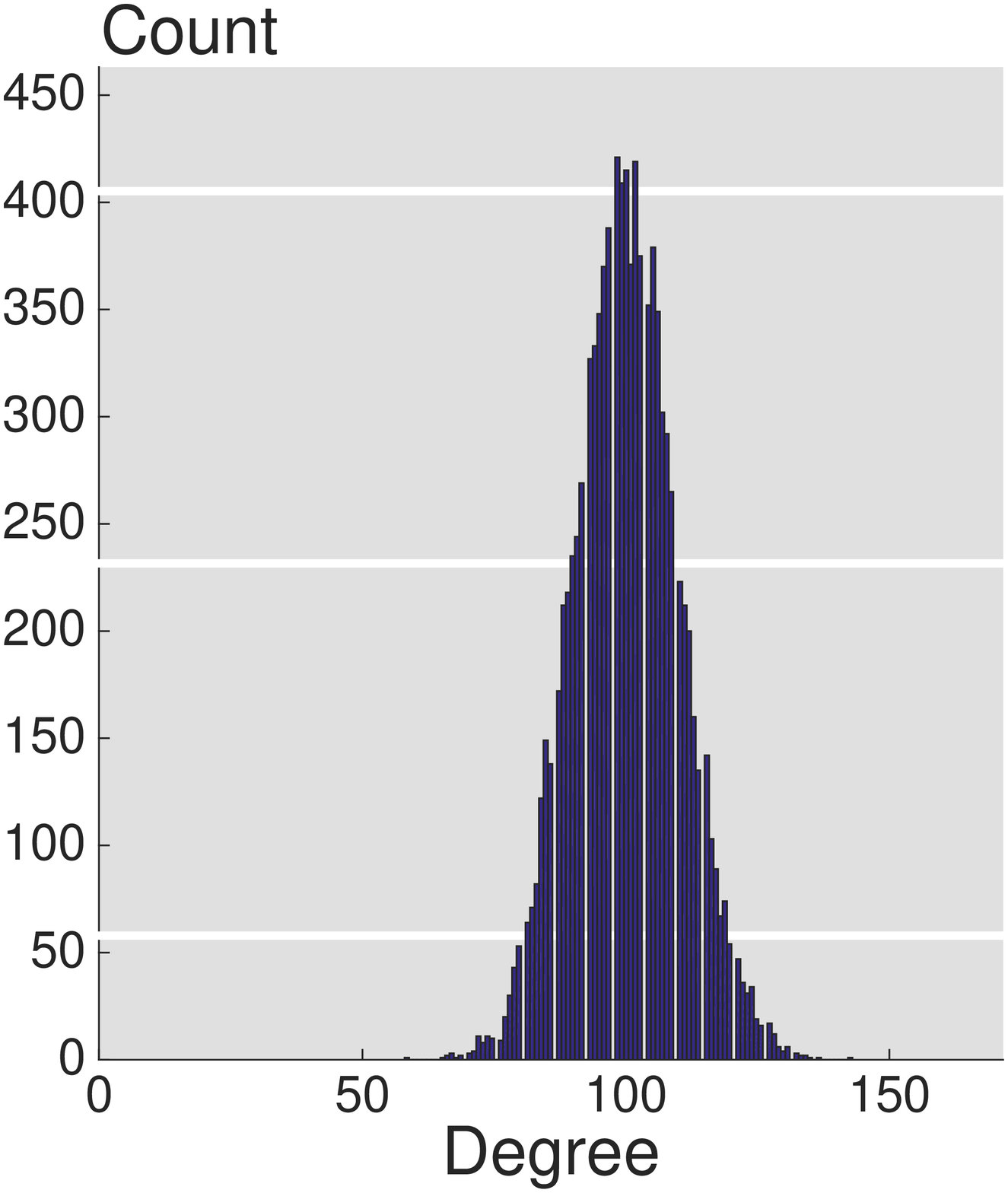}  &
      \includegraphics[width=0.35\columnwidth]{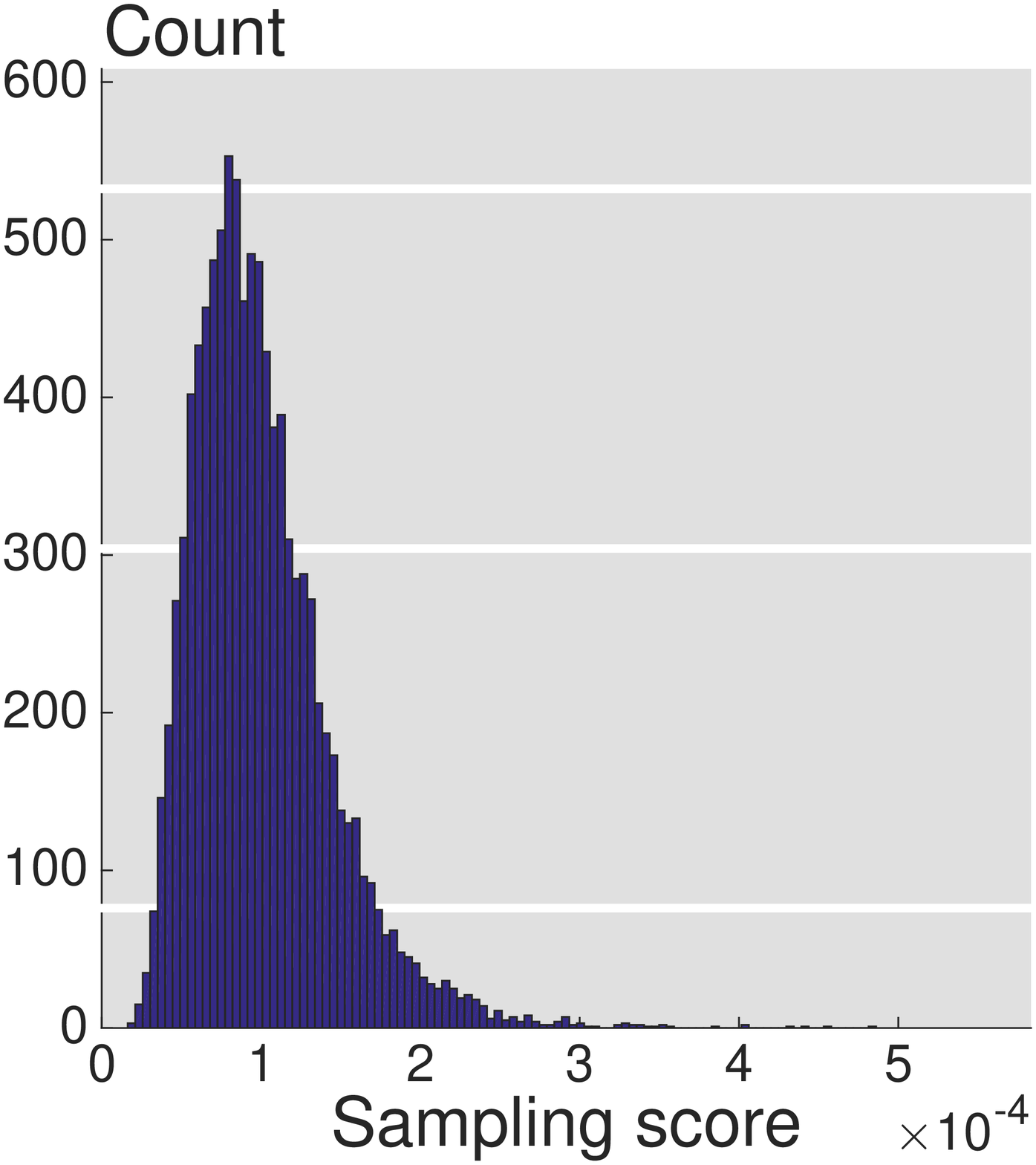} 
      \\
      {\small (a)  Graph plot.} & {\small (b) Histogram of the degrees.} &  {\small (c) Histogram of the leverage scores.}
      \\  
    \end{tabular}
  \end{center}
  \caption{\label{fig:ER} Properties of an
    Erd\H{o}s-R\'enyi graph. Plot (c) shows the histogram of the  leverage scores, which are the optimal sampling scores when the SNR  is large. }
\end{figure*}

\mypar{Random geometric graph} A random geometric graph is a spatial
graph where each of the nodes is assigned random coordinates in the
box $[0, 1]^d$ and an edge appears when the distance between two nodes
is in a given range~\cite{DallC:02}. We used a graph lying in the box
$[0, 1]^2$, and two nodes are connected when the Euclidean distance is
less than $0.03$.  Figure~\ref{fig:Geo} illustrates some
  properties of the random geometric graph: Figure~\ref{fig:Geo}(a)
  shows the graph plot; Figure~\ref{fig:Geo}(b) shows the histogram of the degrees
  that concentrate around 30, as expected (this matches
  previous assertion that given proper parameters, the degree
  distribution of a random geometric graph is the same as that of an
  Erd\H{o}s-R\'enyi graph~\cite{DallC:02}); Figure~\ref{fig:Geo}(c)
  shows the histogram of the leverage scores of $\Vm_{(20)}$, which
  are the optimal sampling scores when the SNR is large; and
  Figure~\ref{fig:Geo}(d) shows the histogram of the leverage score
  on a log-scale.   We see that the
histogram of the leverage scores is skewed; this means that a few
nodes are more important than other
nodes during sampling and have much higher probabilities to be
  chosen. In~\cite{DallC:02}, the authors show that the cluster
properties are different for a random geometric
graph and an Erd\H{o}s-R\'enyi graph. The proposed sampling scores
capture these cluster properties through  the decomposition of the graph adjacency~matrix. Based on the
  above, a random geometric graph is approximately a Type-2 graph.

\begin{figure*}[htb]
  \begin{center}
    \begin{tabular}{cccc}
      \includegraphics[width=0.35\columnwidth]{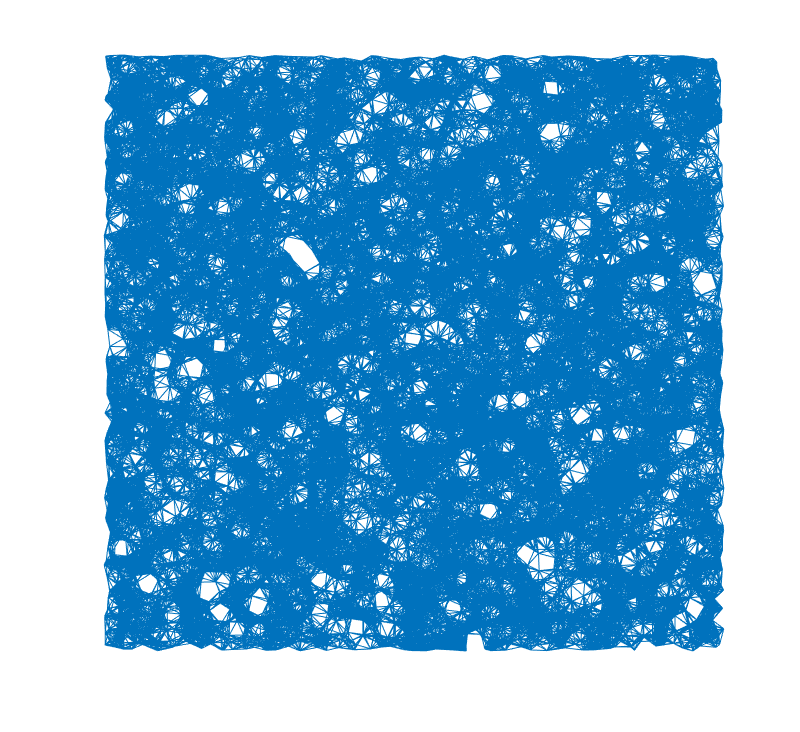}  & \includegraphics[width=0.35\columnwidth]{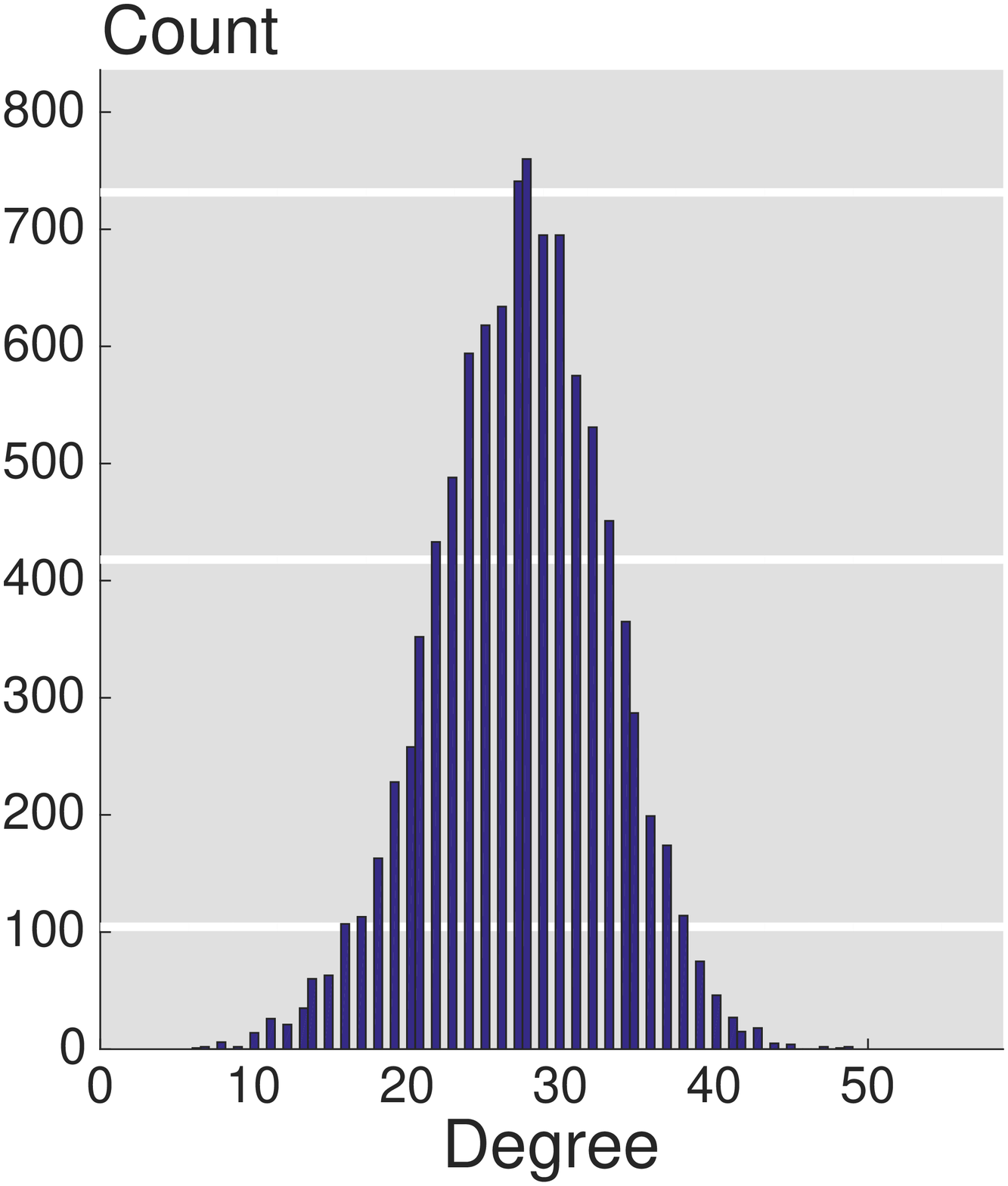}  &
      \includegraphics[width=0.35\columnwidth]{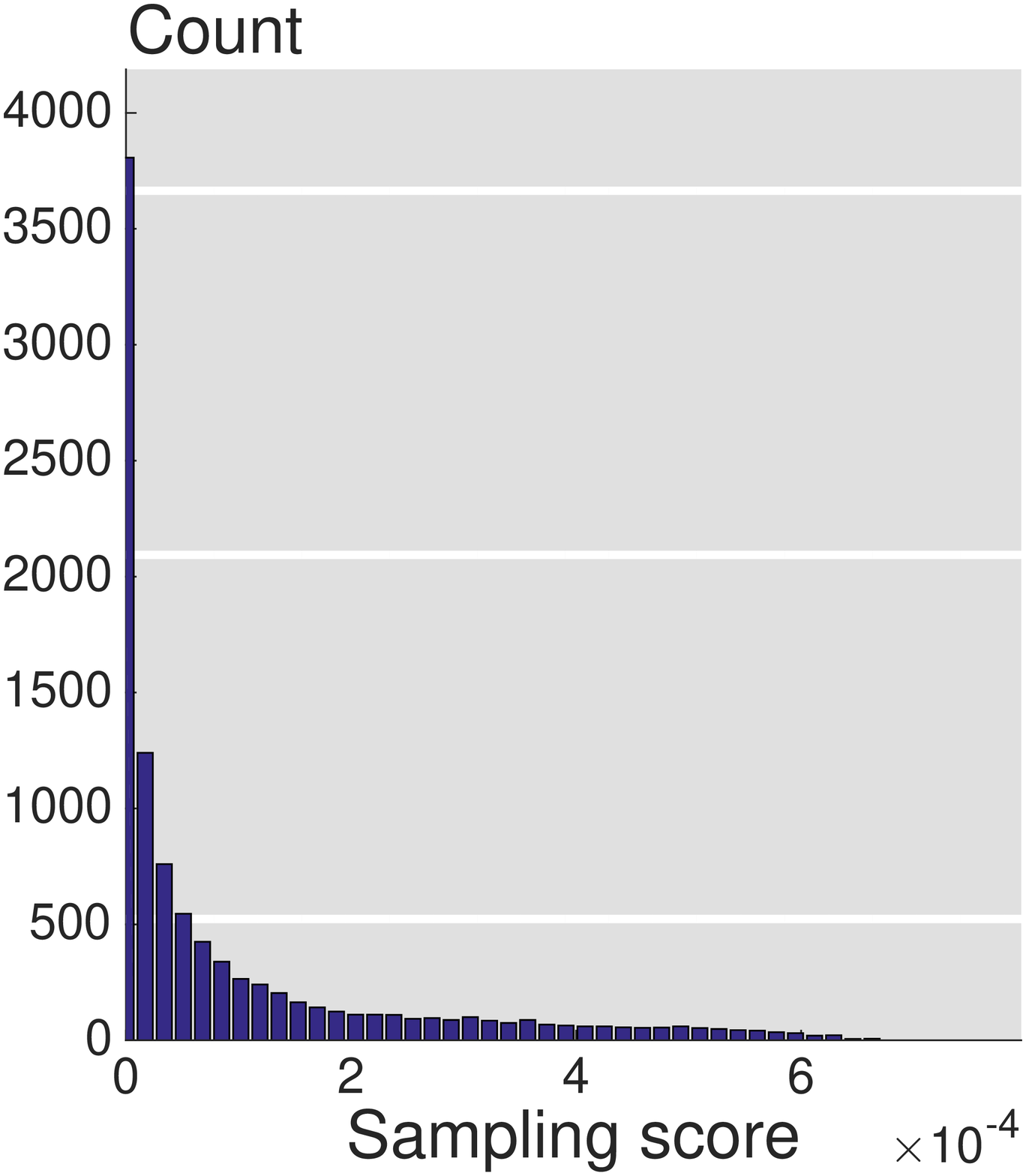}  &
      \includegraphics[width=0.35\columnwidth]{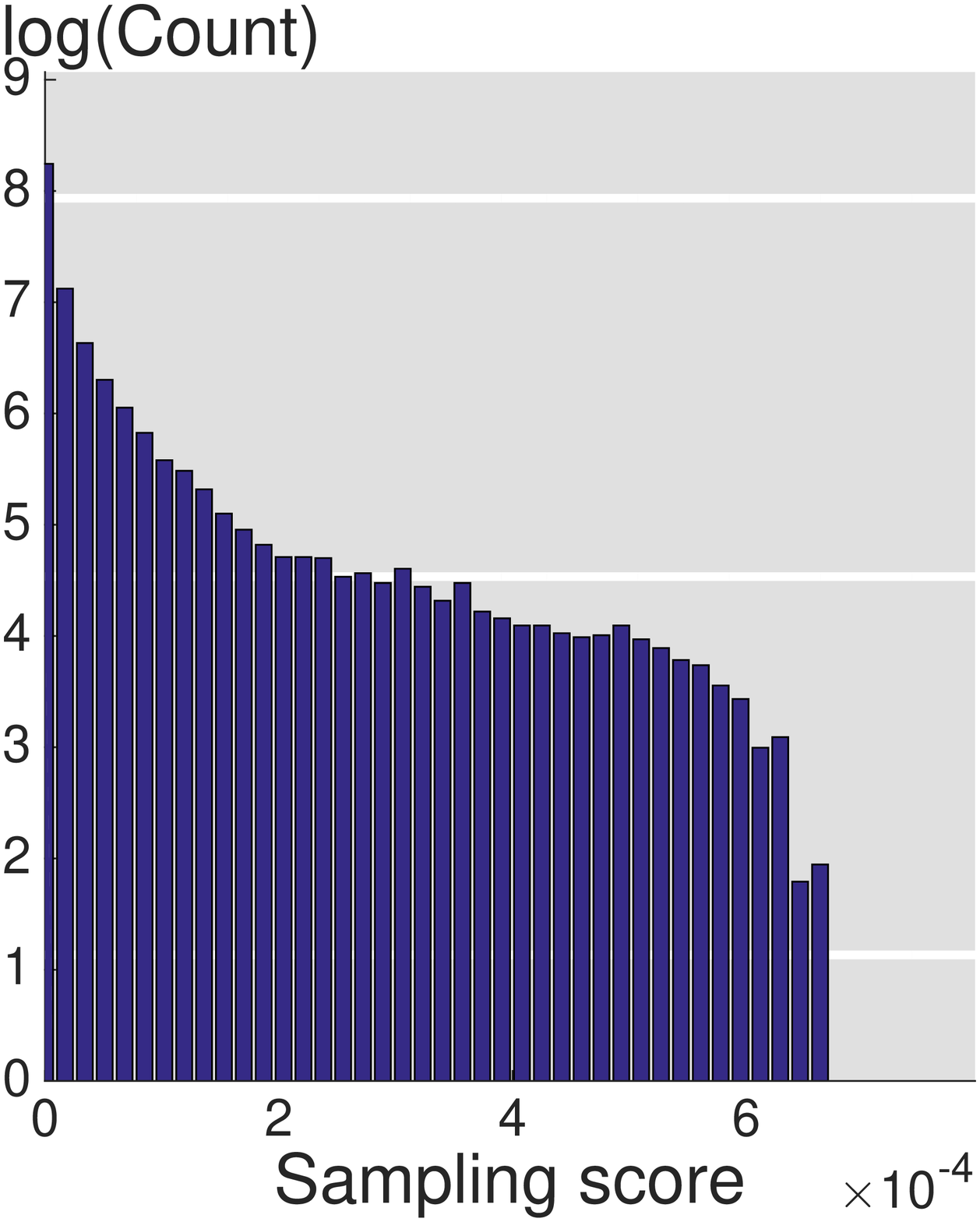}
      \\
      {\small (a)  Graph plot. } & {\small (b)  Histogram of the degrees.} &  {\small (c) Histogram of the leverage scores.} &   {\small (d) Histogram of the leverage scores }
      \\
      & & & {\small (log-scale). }
      \\  
    \end{tabular}
  \end{center}
  \caption{\label{fig:Geo} Properties of a random
    geometric graph. Plot (c) shows the histogram of the leverages
      score, which are the optimal sampling scores when the SNR is
      large. Plot (d) shows the
      log-scale histogram of the leverage scores, which confirms that the leverage scores approximately follow a power-law distribution. }
\end{figure*}

\mypar{Small-world graph} A small-world graph is a graph where most
nodes are not neighbors of one another, but can be reached from any
other node by a small number of hops
(steps)~\cite{Jackson:08,Newman:10}. It is well known that many
real-world graphs, including social networks, the connectivity of the
Internet, and gene networks show small-world graph characteristics. We
use a graph generated from the Watts–Strogatz model, where a ring
graph is first built, followed by the rewiring of the edges with probability $0.01\%$. 
  Figure~\ref{fig:SW} illustrates some properties of the small-world
  graph: Figure~\ref{fig:SW}(a) shows the graph plot;
  Figure~\ref{fig:SW}(b) shows the histogram of the degrees that
  concentrate around 2 (a few nodes have 6 neighbors because of
  rewiring); Figure~\ref{fig:SW}(c) shows the histogram of the
  leverage scores of $\Vm_{(20)}$, which are the optimal sampling scores
  when the SNR is large; and
  Figure~\ref{fig:SW}(d) shows the histogram of the leverage scores
  on a log-scale. We see that the
histogram of the leverage scores is skewed; this means that a
  few nodes are more important than others during sampling and have
  much higher probabilities to be chosen, similarly to the random
  geometric graph. Based on the above, a small-world graph is approximately a Type-2 graph.

\begin{figure*}[htb]
  \begin{center}
    \begin{tabular}{cccc}
      \includegraphics[width=0.35\columnwidth]{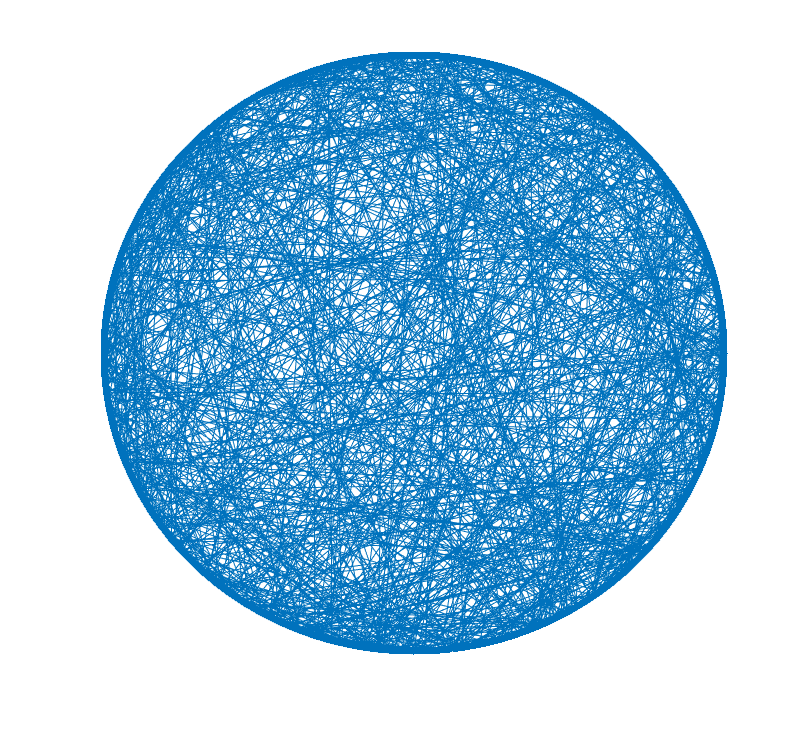}  & \includegraphics[width=0.35\columnwidth]{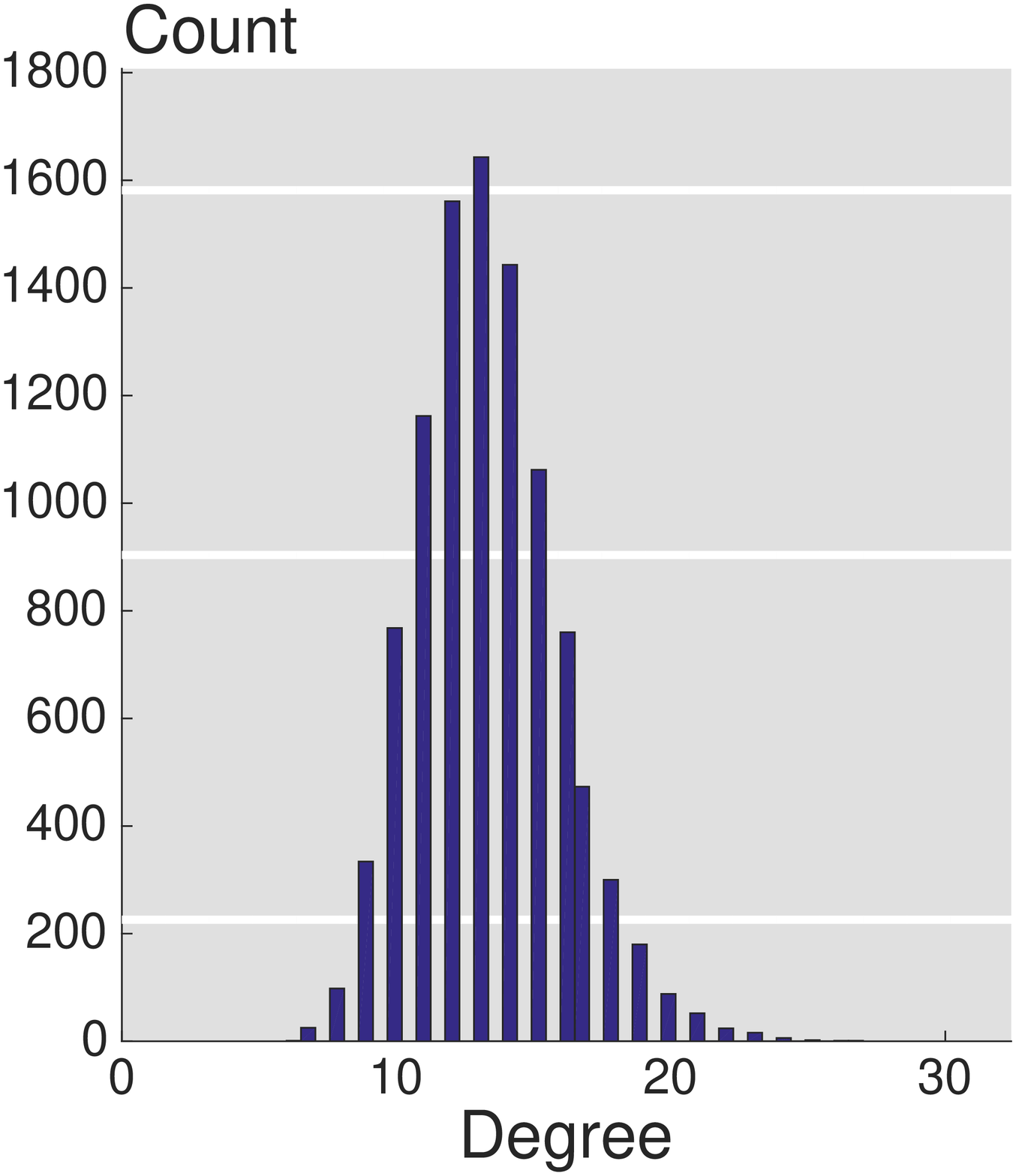}  &
      \includegraphics[width=0.35\columnwidth]{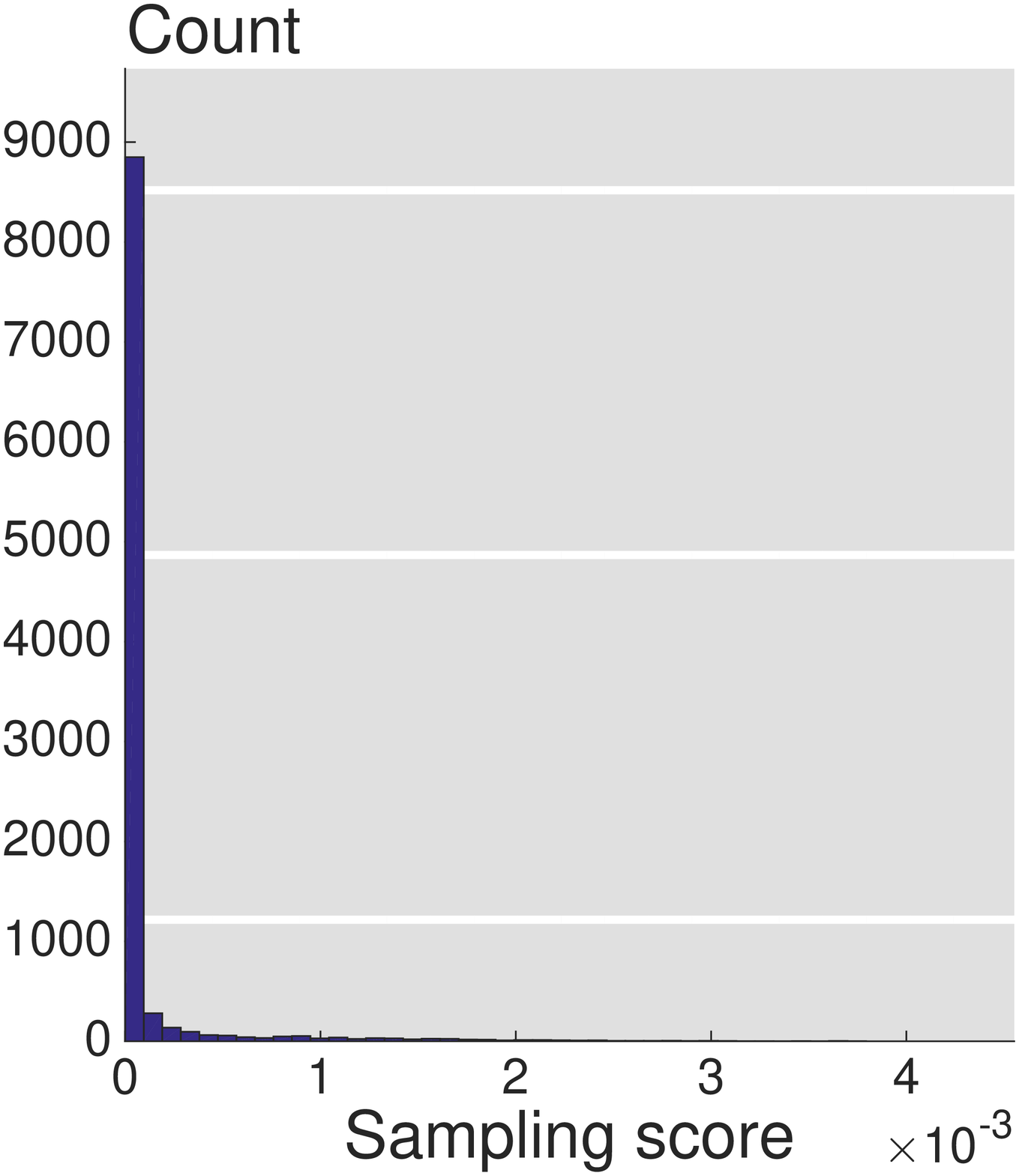} &
      \includegraphics[width=0.35\columnwidth]{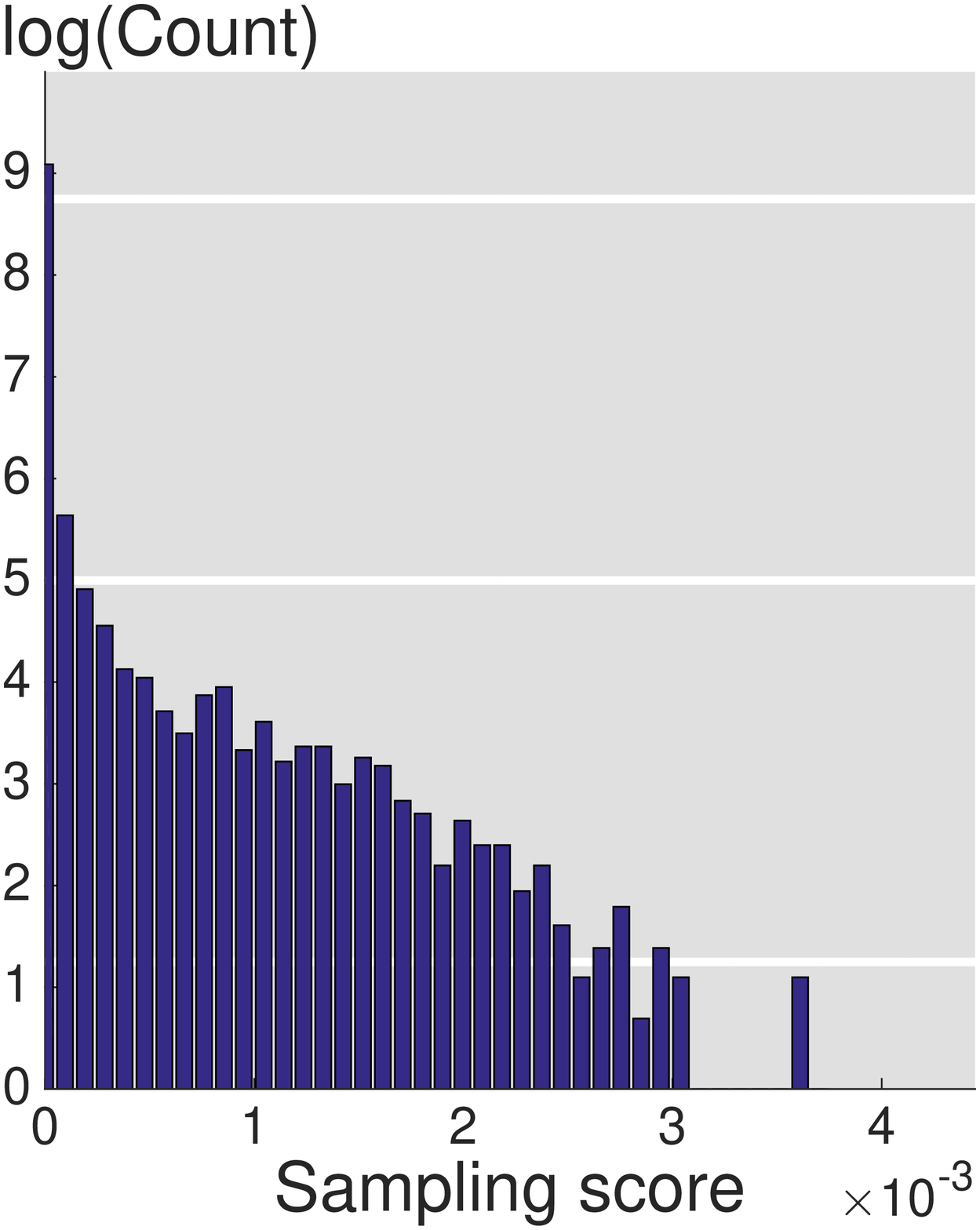} 
      \\
      {\small (a)  Graph plot. } & {\small (b)  Histogram of the degrees. } &  {\small (c)  Histogram of the leverage scores.} & 
      {\small (d)  Histogram of the leverage scores}
       \\
      & & & {\small (log-scale).}
      \\  
    \end{tabular}
  \end{center}
  \caption{\label{fig:SW} Properties of a small-world
    graph. Plot (c) shows the histogram of the leverage score,
      which is the optimal sampling score when the SNR is
      large. Plot (d) shows the
      log-scale histogram of the leverage scores, which confirms that the leverage scores approximately follow a power-law distribution. }
\end{figure*}

\mypar{Power-law graph} A power-law graph is a graph where the more
connected a node is, the more likely it is to receive new links, known
as a preferential attachment graph~\cite{Jackson:08,Newman:10}. It is
well known that the degree distribution of a preferential attachment
graph follows the power law. We use a graph generated from the
Barab\'{a}si-Albert model, where new nodes are added to the network
one at a time. Each new node is connected to one existing node with a
probability that is proportional to the number of links that the
existing nodes already have.  Figure~\ref{fig:PL} illustrates
  some properties of the small-world graph: Figure~\ref{fig:PL}(a)
  shows the graph plot; Figure~\ref{fig:PL}(b) shows the histogram of the degrees that is skewed, which clearly follows the power law;
  Figure~\ref{fig:PL}(c) shows the histogram of the leverage scores
  of $\Vm_{(20)}$, which are the optimal sampling scores when the SNR is
  large; and Figure~\ref{fig:PL}(d) shows the histogram of the
  leverage scores on a log-scale. We see that the histogram of the leverage
scores is skewed; this means that a few nodes
  are more important than others during sampling and have much higher
  probabilities to be chosen, similarly to the random geometric graph
  and the small-world graph. Based on the above, a
  power-law graph is approximately a Type-2 graph.
 
\begin{figure*}[htb]
  \begin{center}
    \begin{tabular}{cccc}
      \includegraphics[width=0.35\columnwidth]{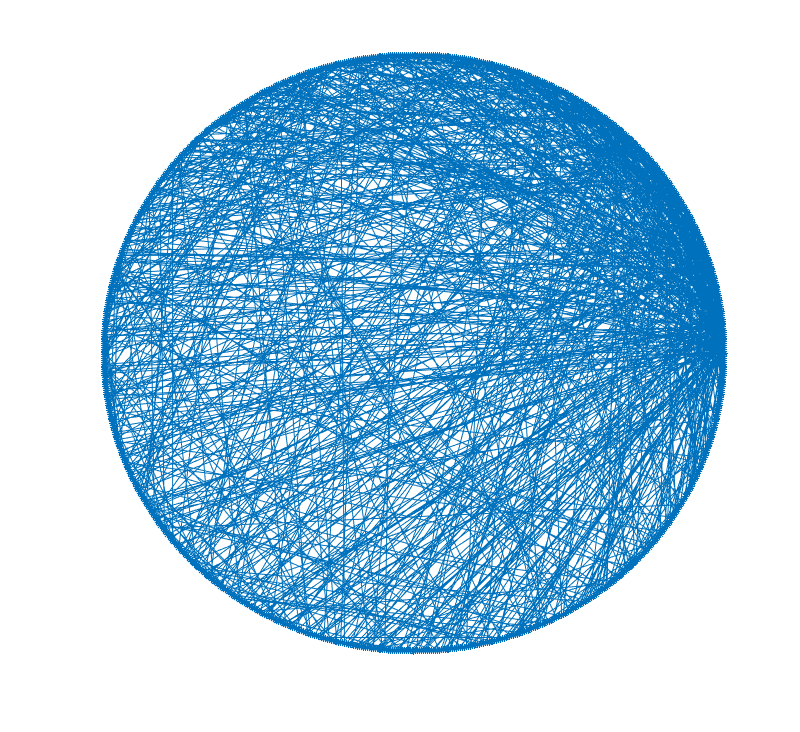}  & \includegraphics[width=0.35\columnwidth]{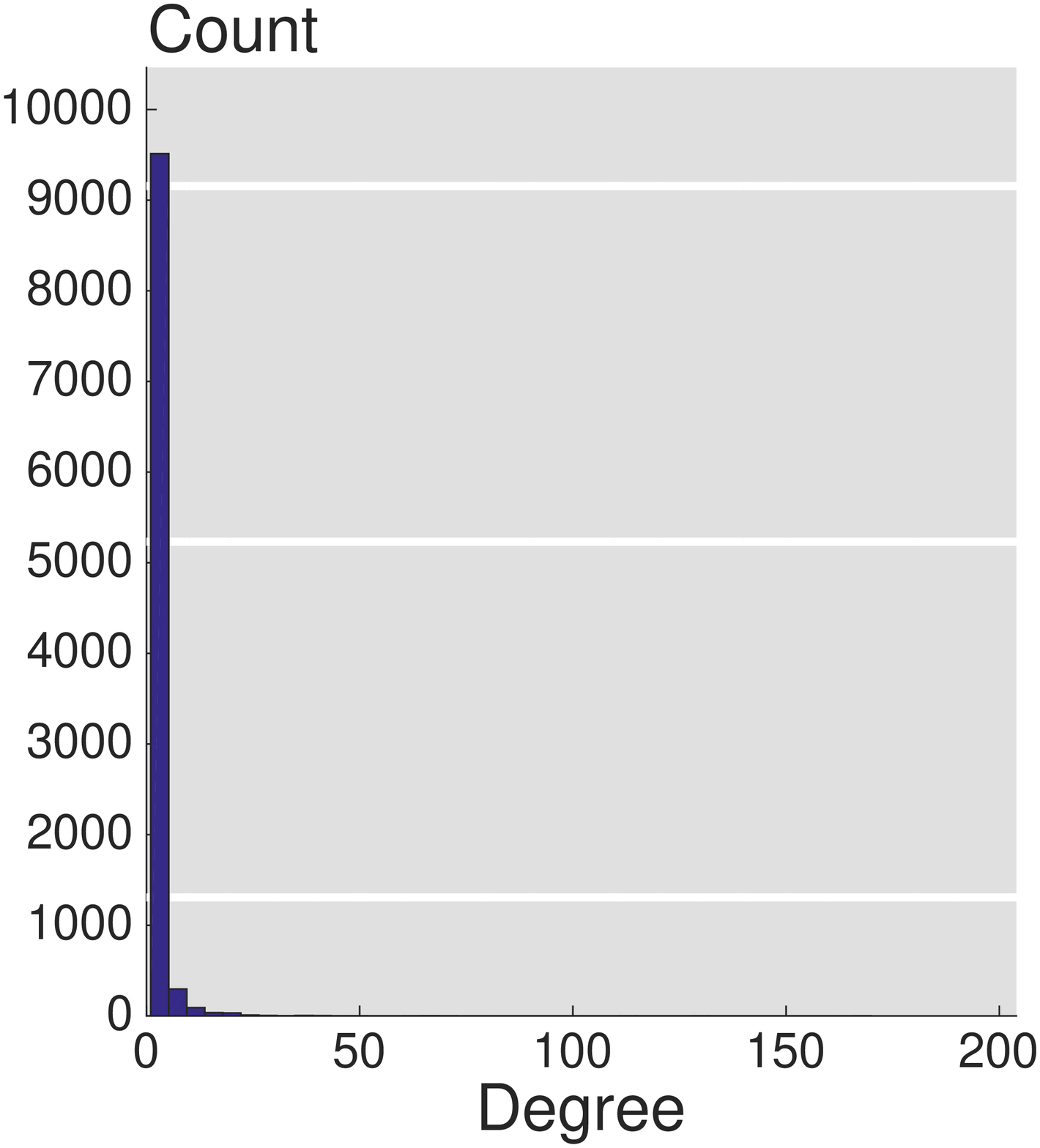}  &
      \includegraphics[width=0.35\columnwidth]{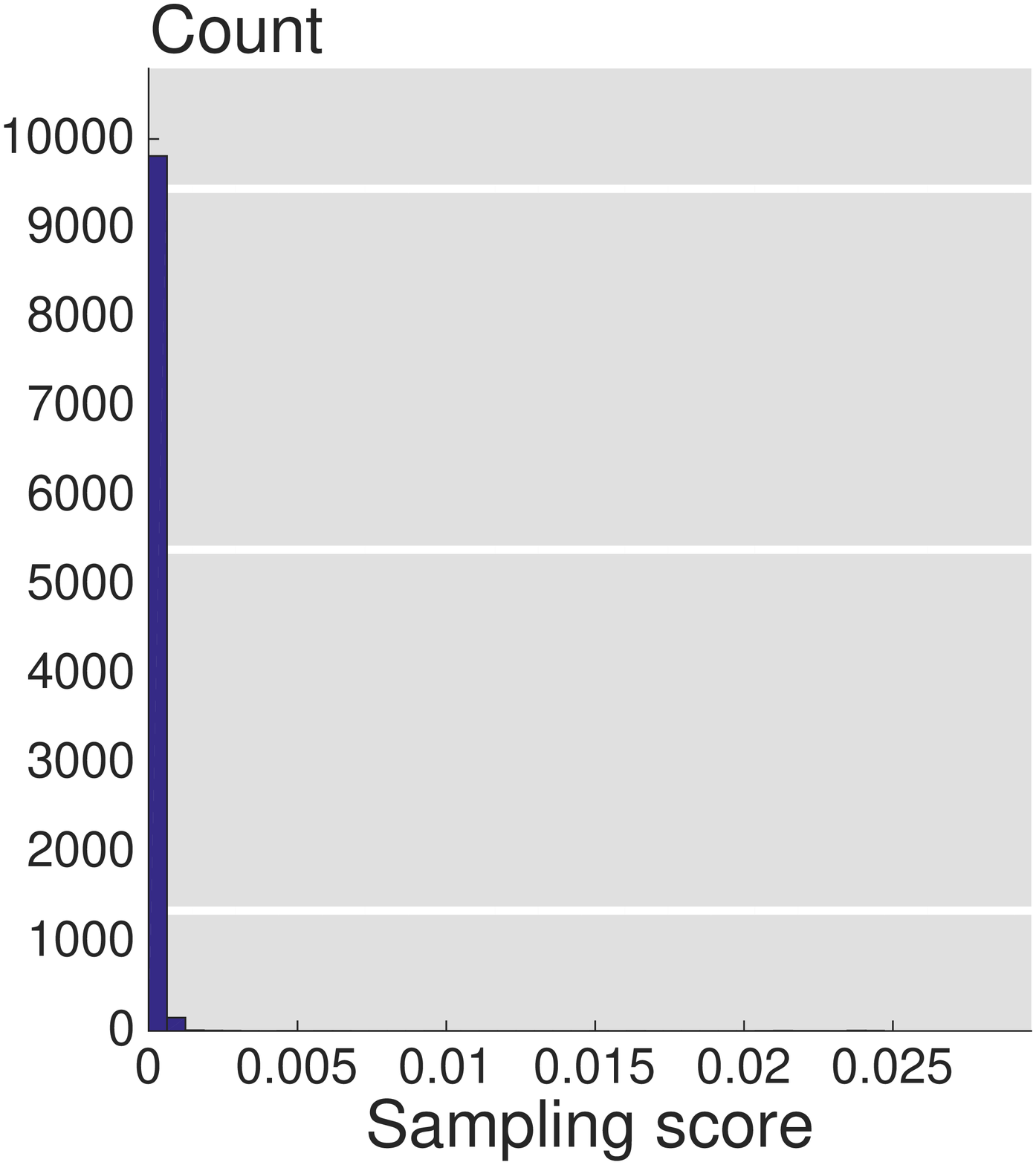} &
      \includegraphics[width=0.35\columnwidth]{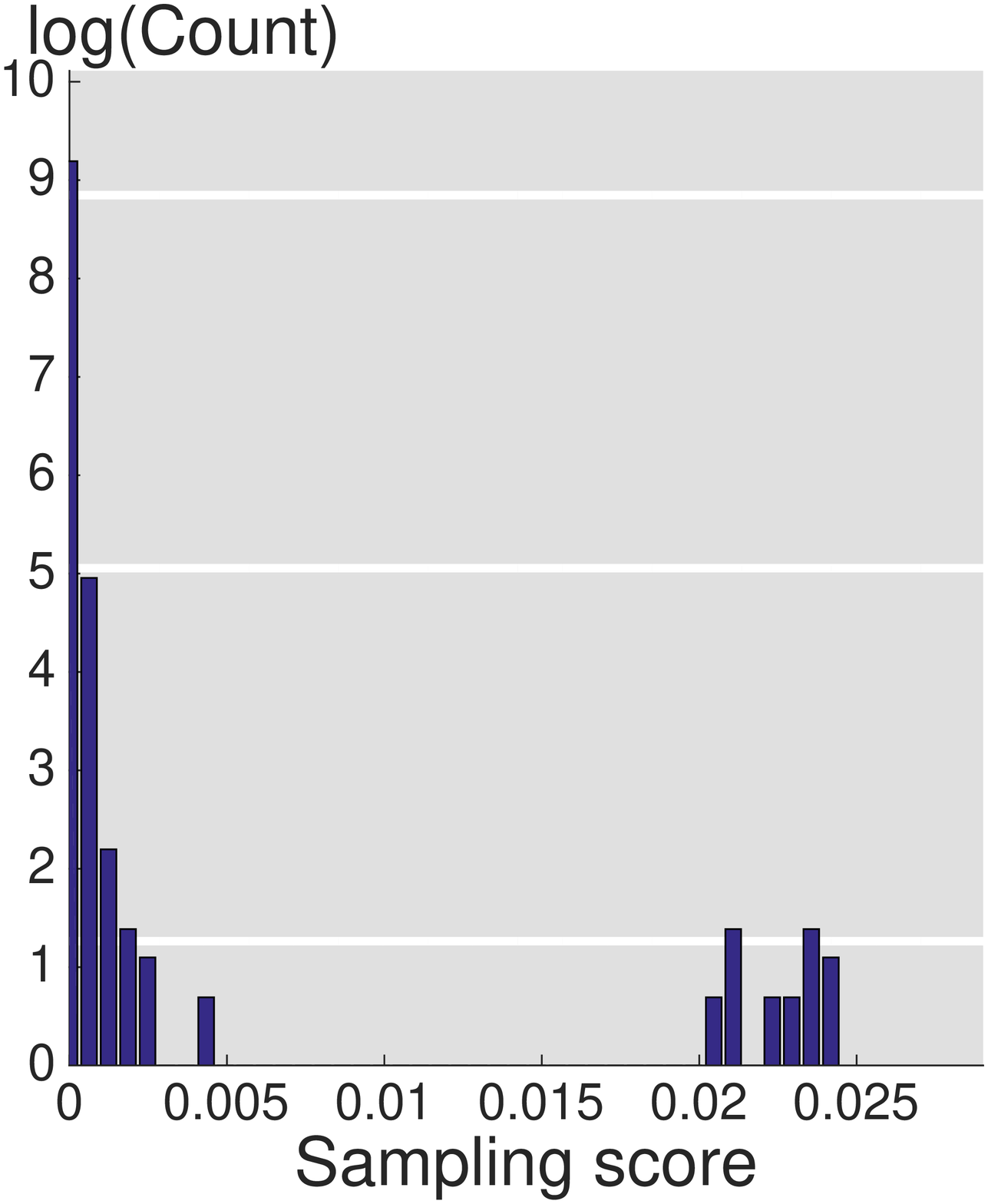} 
      \\
      {\small (a)  Graph plot. } & {\small (b) Histogram of the degrees.} &  {\small (c)  Histogram of the leverage scores.} & 
      {\small (d) Histogram of the leverage scores}
       \\
      & & & {\small (log-scale).}
      \\  
    \end{tabular}
  \end{center}
  \caption{\label{fig:PL} Properties of a power-law
    graph. Plot (c) shows the histogram of the leverage score,
      which is the optimal sampling score when the SNR is
      large. Plot (d) shows the
      log-scale histogram of the leverage scores, which confirms that the leverage scores approximately follow a power-law distribution.}
\end{figure*}

\mypar{Types} Based on our observation of the graph
Fourier transform matrices of each graph, the ring graph and the
Erd\H{o}s-R\'enyi graph approximately satisfy the requirement to be
the Type-1 graph, while the random geometric graph,
the small-world graph and the power-law graph approximately
satisfy the requirement to be the Type-2 graph. We thus expect that
the experimentally designed sampling has similar performance to uniform sampling for the ring graph and the
Erd\H{o}s-R\'enyi graph, while it outperforms uniform sampling
for the random geometric graph, the small-world graph and the power-law
graph.

\subsection{Simulated Graph Signals}
\label{sec:simulatedgraphsignals}
For each graph $\Adj$, we generate 1,000 graph signals through the
following two steps: We first generate the graph frequency components
as
\begin{eqnarray}
  \label{eq:simulation}
  \widehat{x}_k
  \left\{
    \begin{array}{ll}
      \sim \mathcal{N}(1, 0.5^2) & \mbox{when } k < K; \\
      =  {K^{2\beta}}/{k^{2\beta}} & \mbox{when } k \geq K.
    \end{array}
  \right.
\end{eqnarray}
We then normalize $\xhat$ to have unit norm, and
obtain $\x = \Vm \xhat$. It is clear that $\x \in \BLT_{\Adj } (K,
\beta, \mu)$, where $K = 10$ and $\beta$ varies as $0.5$ and
$1$. During sampling, we simulate noise $\epsilon \sim
\mathcal{N}(0, \sigma^2)$, vary the sample size $m$ from $1,000$ to $20,000$, and vary $\sigma^2$ from low noise level
$10^{-4}$ ($ \left\| \x \right\|_2/\left\| \epsilon \right\|_2 =
100$) to high noise level $0.02$ ($\left\| \x \right\|_2  / \left\| \epsilon \right\|_2 = 0.5 $). During recovery, we set the bandwidth $K =
\max( 10, m^{1/ 2 \beta+1} )$ as suggested in
Corollaries~\ref{cor:type1_opt} and~\ref{cor:type2_opt}.

\subsection{Results}
We compare four sampling strategies, including uniform sampling (in
blue), leverage score based sampling (in orange), square root of the
leverage score based sampling (in purple), and degree based sampling
(in red). Note that the last three sampling strategies all belong
  to experimentally designed sampling because they are designed
  based on the structure of the graph.  As shown in Section~\ref{sec:statistical}, leverage score based sampling is approximately optimal when the SNR is large,
square root of the leverage score based sampling is approximately the
optimal when the SNR is small. We also use the degrees
as the sampling scores because previous works show that the largest
eigenvectors of adjacency matrices often have most of their mass
localized on high-degree nodes~\cite{GohKK:01, CucuringuM:11}, which
implies the high correlation between degree and leverage score. We
evaluate the recovery performance by using the MSE,
\begin{eqnarray*} 
{\rm MSE} \ = \ \left\| \x^* - \x \right\|_2^2,
\end{eqnarray*}
where $\x^*$ is the recovered graph signal and $\x$ is the original
graph signal. The simulation results for the ring graph, the
Erd\H{o}s-R\'enyi graph, the random geometric graph, the small-world
graph and the power-law graph are shown in
Figures~\ref{fig:Ring_Err},~\ref{fig:ER_Err},~\ref{fig:Geo_Err},~\ref{fig:SW_Err}
and~\ref{fig:PL_Err}, respectively. We summarize the important points
below.

\begin{itemize}
\item All of the sampling strategies perform similarly on the ring
  graph and the Erd\H{o}s-R\'enyi graph, matching
  Corollary~\ref{cor:type1_opt}.

\item Experimentally designed sampling outperforms uniform sampling on
  the random geometric graph, the small-world graph and the power-law
  graph, matching Corollary~\ref{cor:type2_opt}. Especially for the
  small-world graph and the power-law graph, uniform sampling is much
  worse than experimentally designed sampling.

\item Leverage score based sampling outperforms all other sampling
  strategies when the noise level is low.

\item Square root of the leverage score based sampling outperforms all
  other sampling strategies when the noise level is high.

\item Degree based sampling outperforms uniform sampling because of
  its correlation to the leverage score based sampling, but for the
  small-world graph, degree based sampling is still much worse than
  leverage score based sampling.

\item When $\beta$ is larger, the recovery performance is better
  because less energy is concentrated in the high-frequency band for
  approximately bandlimited graph signals.

\item When $\sigma^2$ is smaller, the recovery performance is better
  because of less noise.

\item  The degree distribution is not a reliable indicator of when
    experimentally designed sampling outperforms uniform sampling.
    The degree distributions of the Erd\H{o}s-R\'enyi graph and the
    random geometric graph are similar, but experimentally designed
    sampling only outperforms uniform sampling on the random geometric
    graph. This implies that the first-order information provided by
    the degree is not sufficient in designing samples.
\end{itemize}

\begin{figure}[htb]
  \begin{center}
    \begin{tabular}{cc}
\includegraphics[width=0.45\columnwidth]{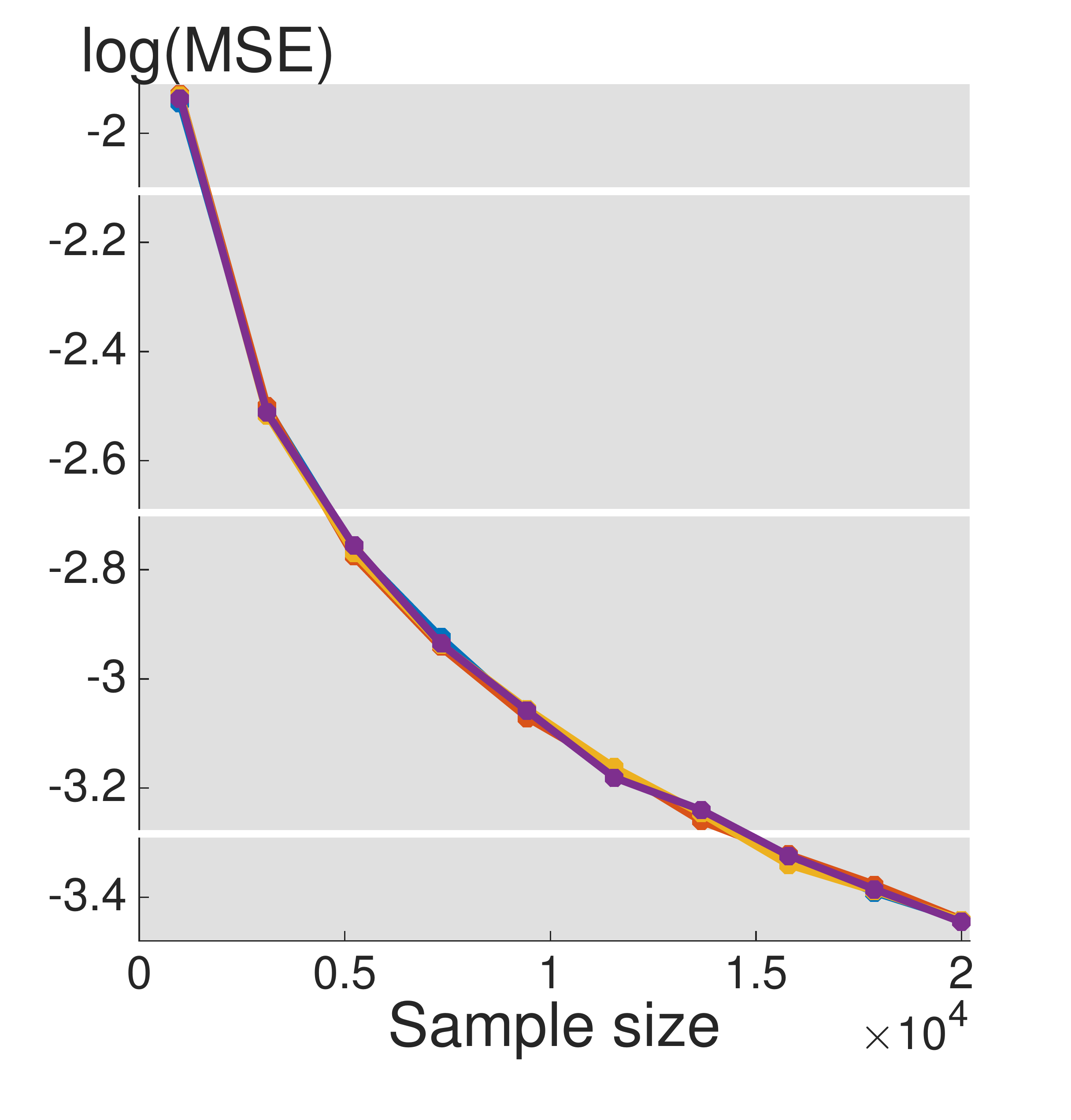}  & \includegraphics[width=0.45\columnwidth]{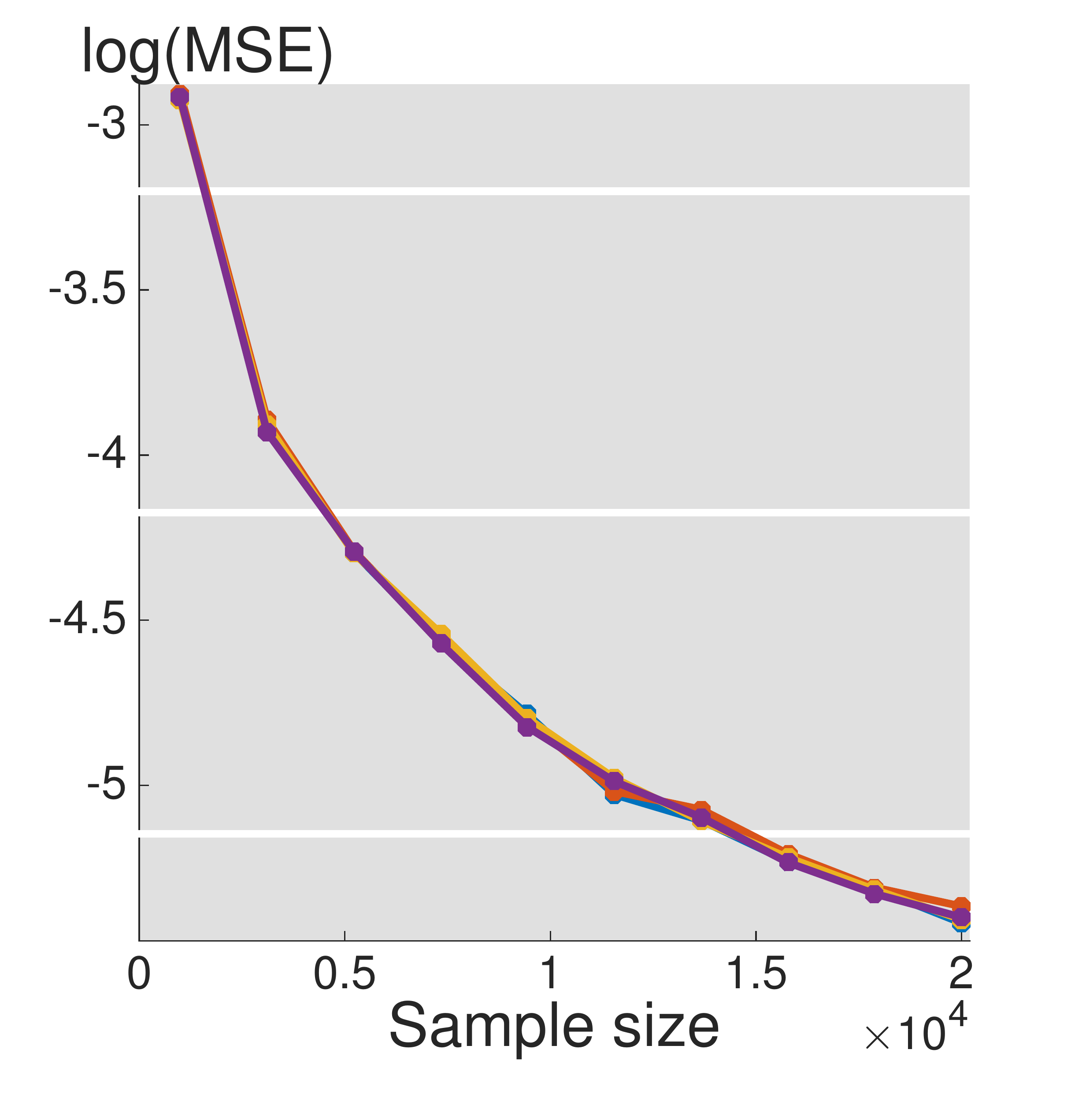} 
\\
      {\small (a)  $\beta = 0.5, \sigma^2 = 10^{-4}$ } & {\small (b) $\beta = 1, \sigma^2 = 10^{-4}$} 
\\
\includegraphics[width=0.45\columnwidth]{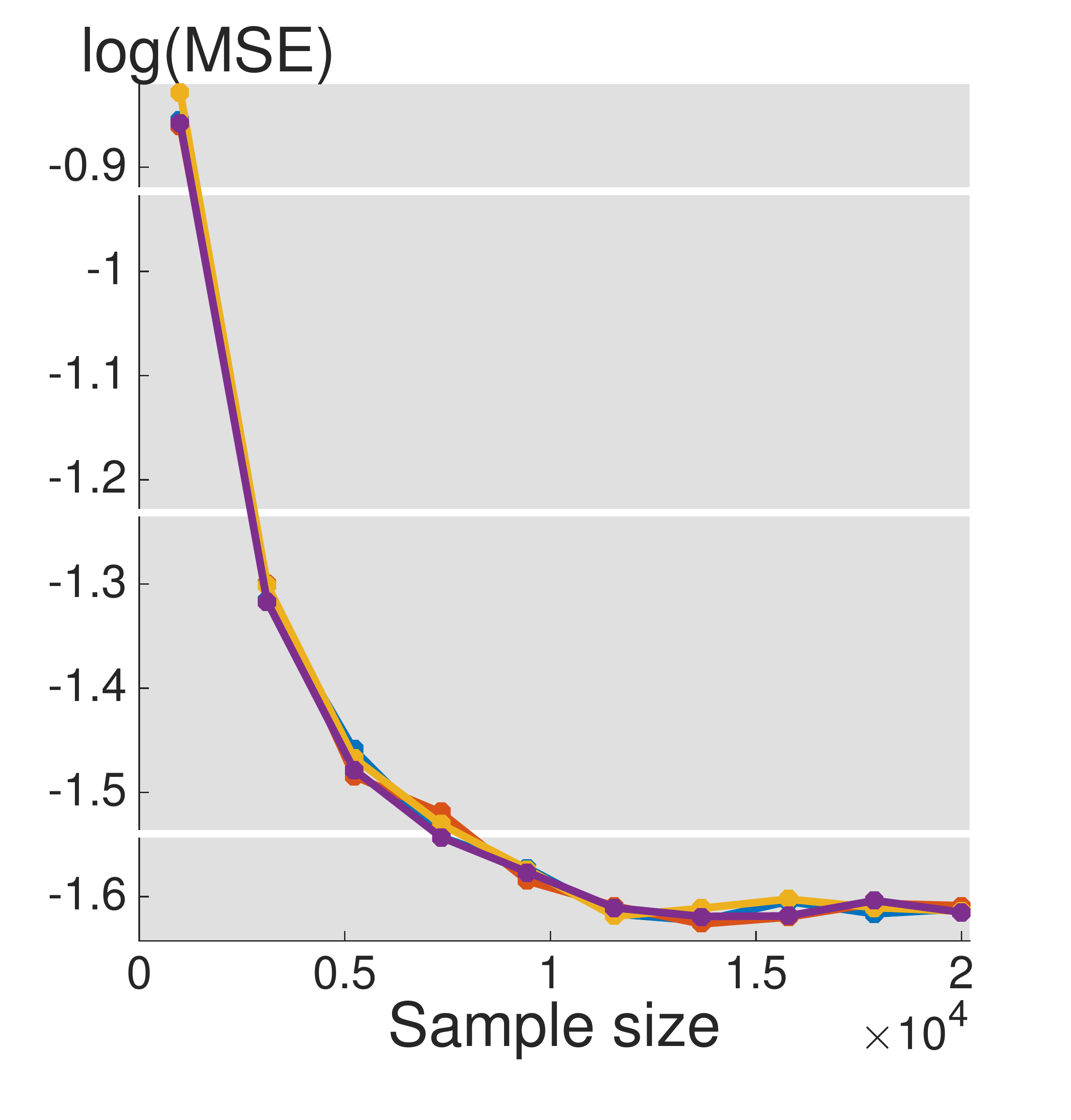}  & \includegraphics[width=0.45\columnwidth]{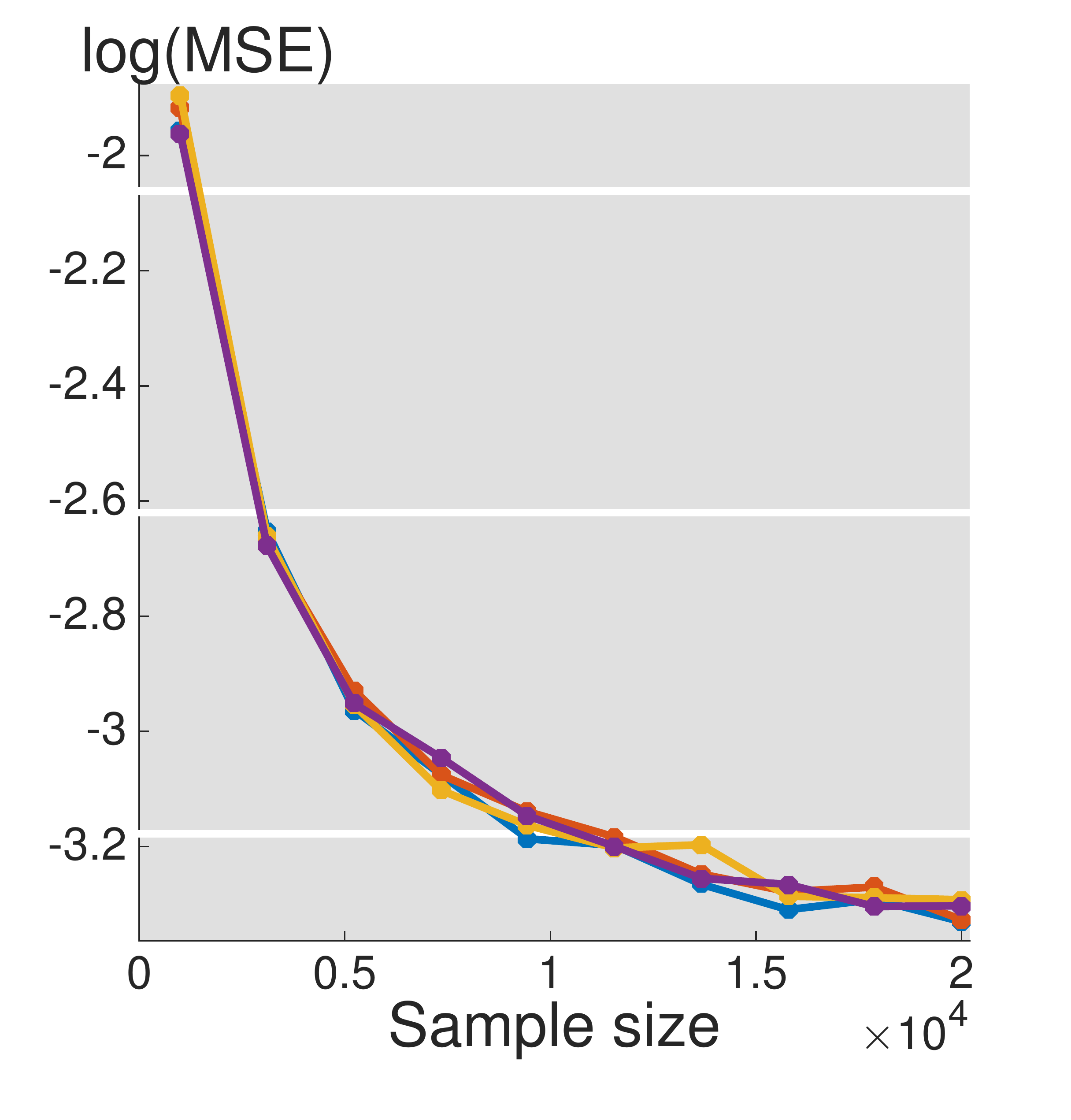} 
\\
      {\small (d)  $\beta = 0.5, \sigma^2 = 2 \times 10^{-2}$ } & {\small (e) $\beta = 1, \sigma^2 = 2 \times 10^{-2}$}
 \\  
\end{tabular}
  \end{center}
  \caption{\label{fig:Ring_Err} MSE comparison for the ring graph for
    uniform sampling (in blue), leverage score based sampling (in
    orange), square root of the leverage score based sampling (in
    purple) and degree based sampling (in red). }
\end{figure}

\begin{figure}[htb]
  \begin{center}
    \begin{tabular}{cc}
\includegraphics[width=0.45\columnwidth]{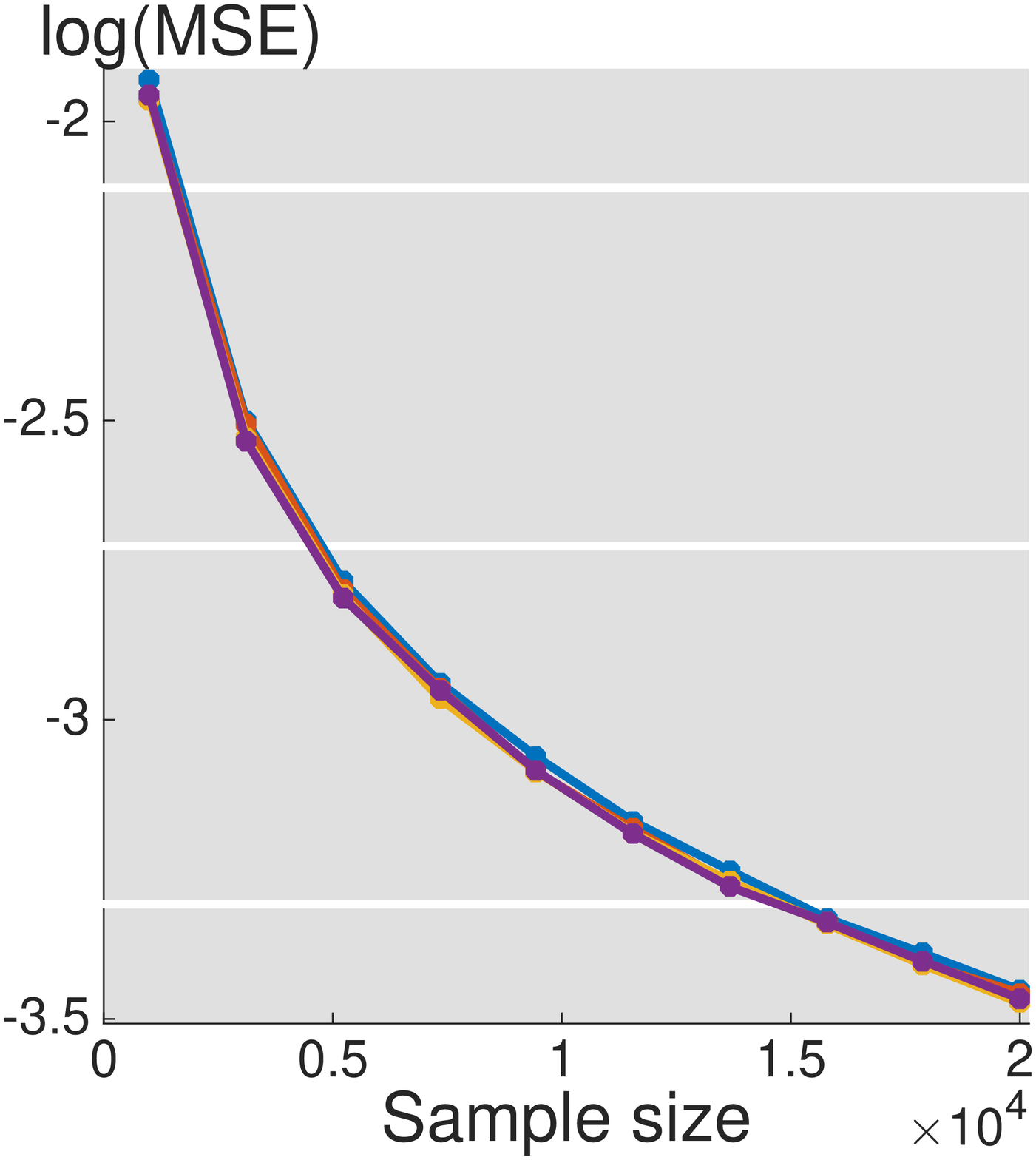}  & \includegraphics[width=0.45\columnwidth]{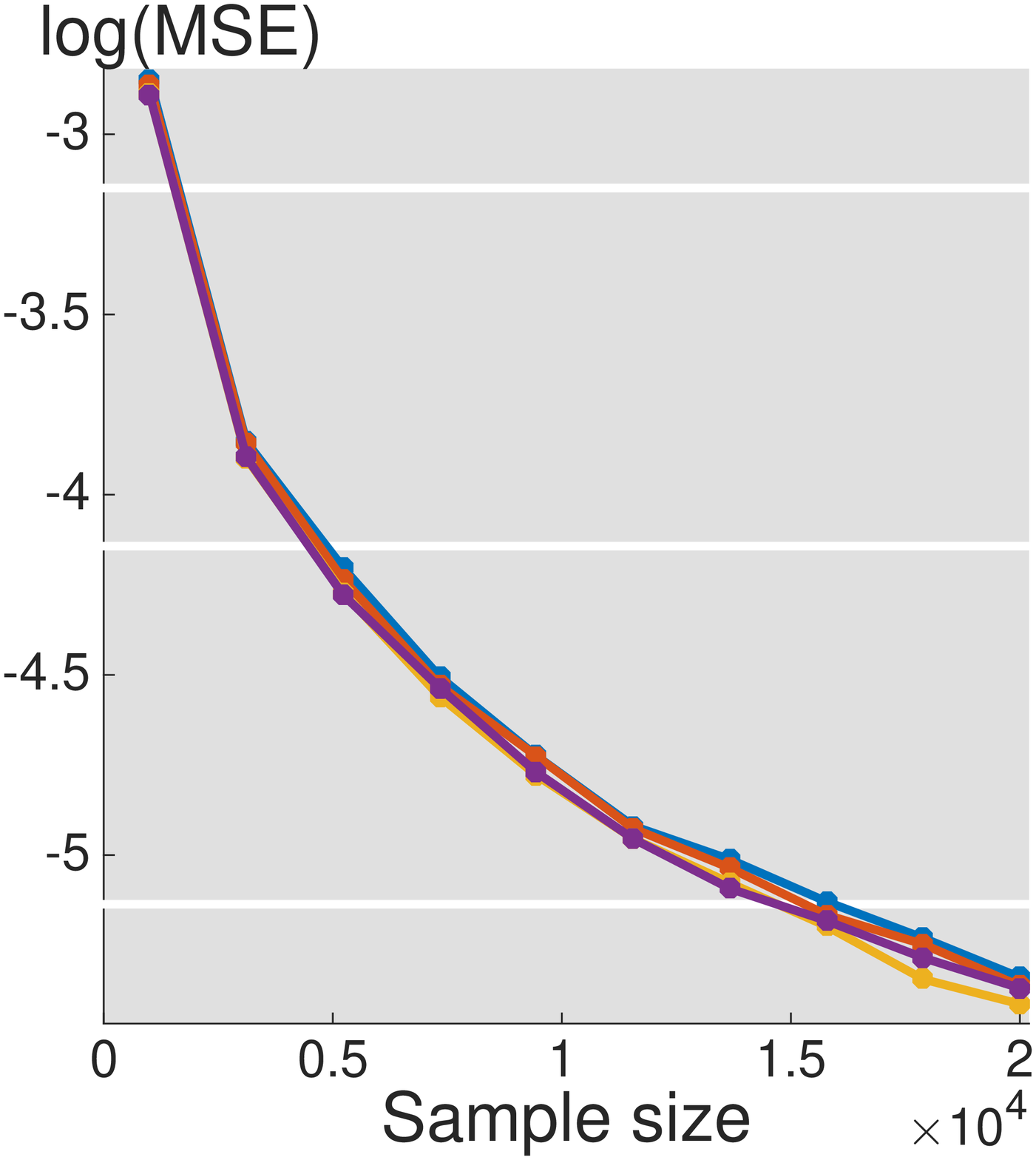} 
\\
      {\small (a)  $\beta = 0.5, \sigma^2 = 10^{-4}$. } & {\small (b) $\beta = 1, \sigma^2 = 10^{-4}$. } 
\\
\includegraphics[width=0.45\columnwidth]{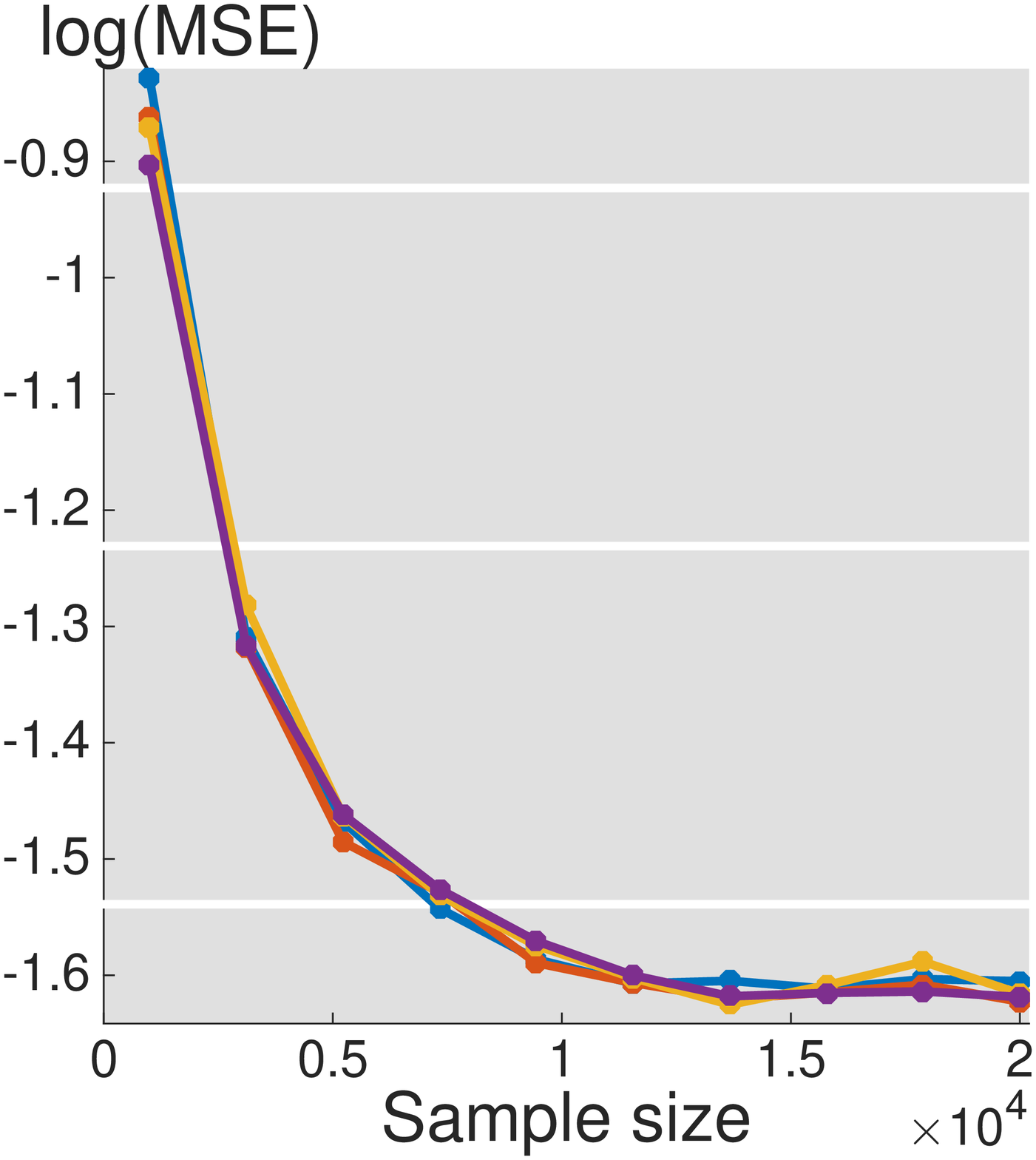}  & \includegraphics[width=0.45\columnwidth]{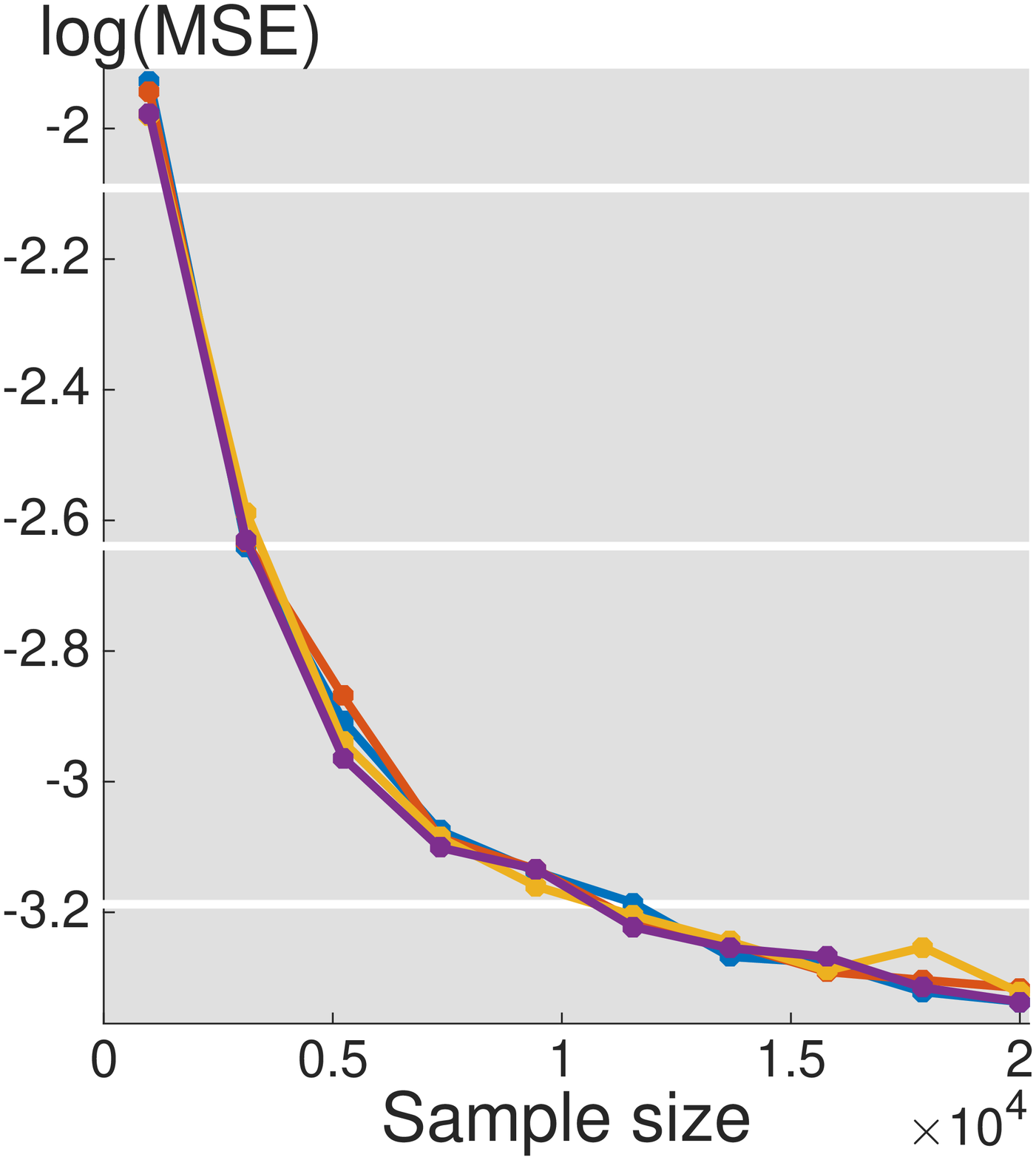} 
\\
      {\small (c)  $\beta = 0.5, \sigma^2 = 2 \times 10^{-2}$. } & {\small (d) $\beta = 1, \sigma^2 = 2 \times 10^{-2}$.}
 \\  
\end{tabular}
  \end{center}
  \caption{\label{fig:ER_Err} MSE comparison for the Erd\H{o}s-R\'enyi
    graph for uniform sampling (in blue), leverage score based
    sampling (in orange), square root of the leverage score based
    sampling (in purple) and degree based sampling (in red). }
\end{figure}

\begin{figure}[htb]
  \begin{center}
    \begin{tabular}{cc}
\includegraphics[width=0.45\columnwidth]{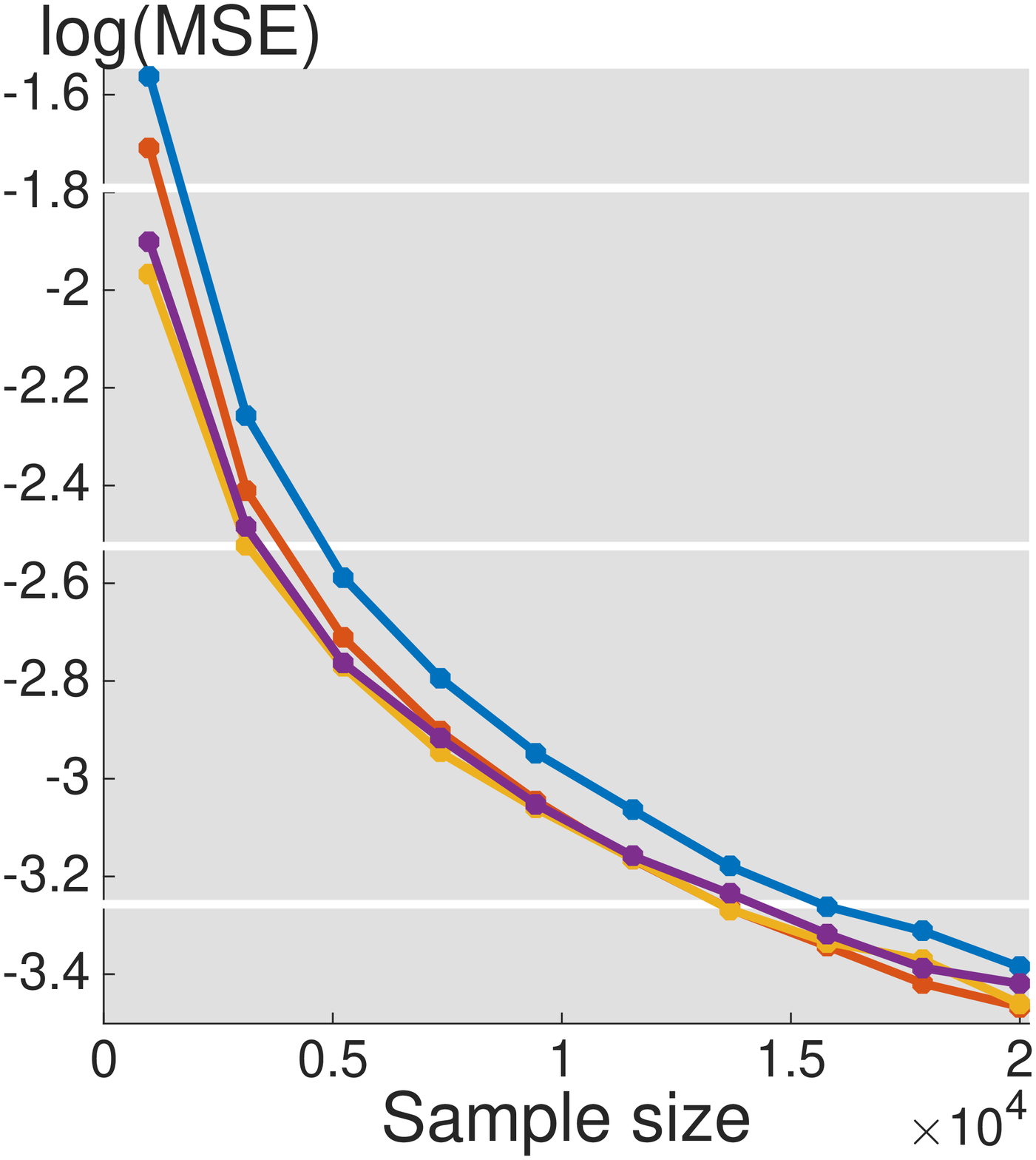}  & \includegraphics[width=0.45\columnwidth]{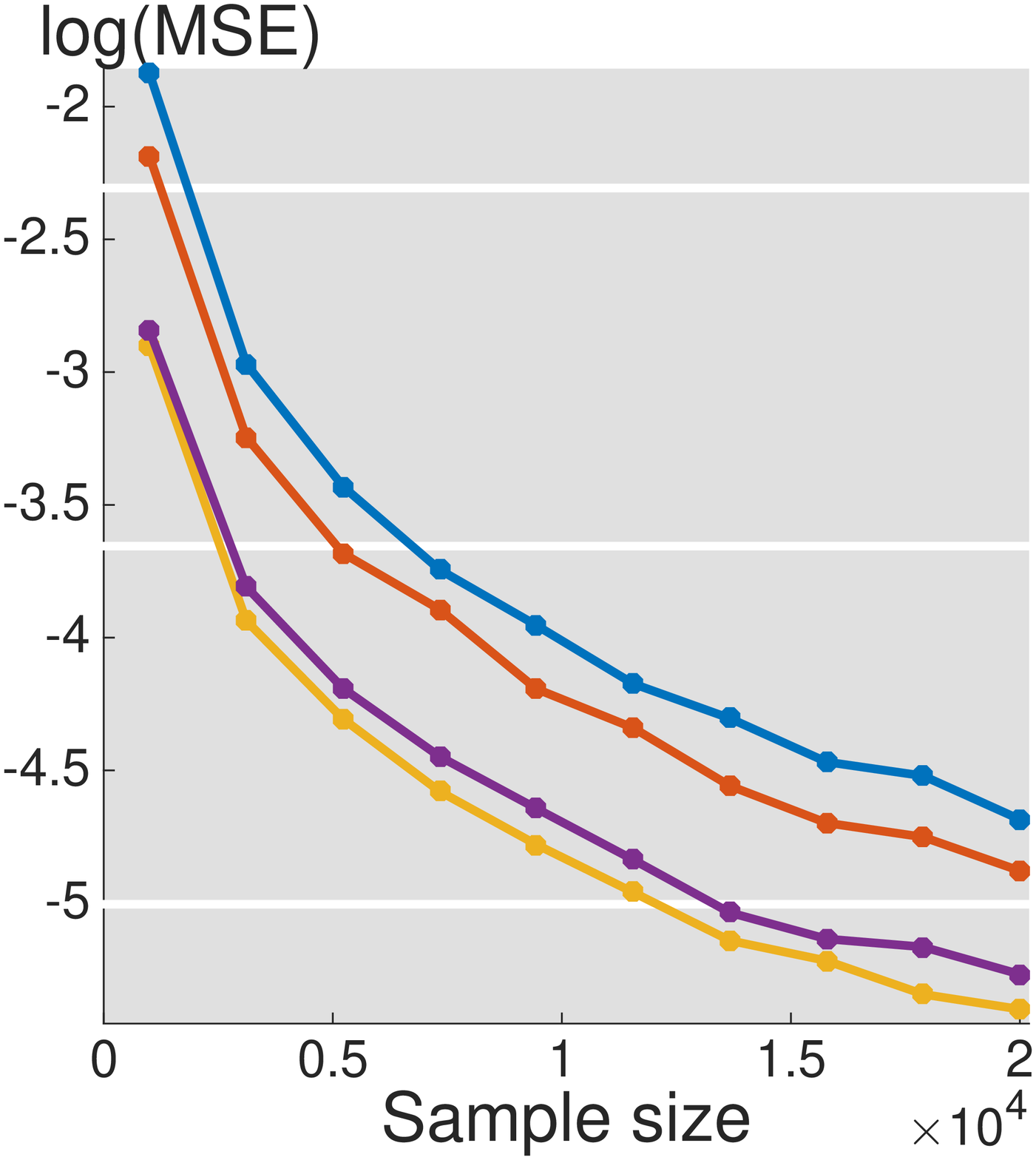} 
\\
      {\small (a)  $\beta = 0.5, \sigma^2 = 10^{-4}$. } & {\small (b) $\beta = 1, \sigma^2 = 10^{-4}$. } 
\\
\includegraphics[width=0.45\columnwidth]{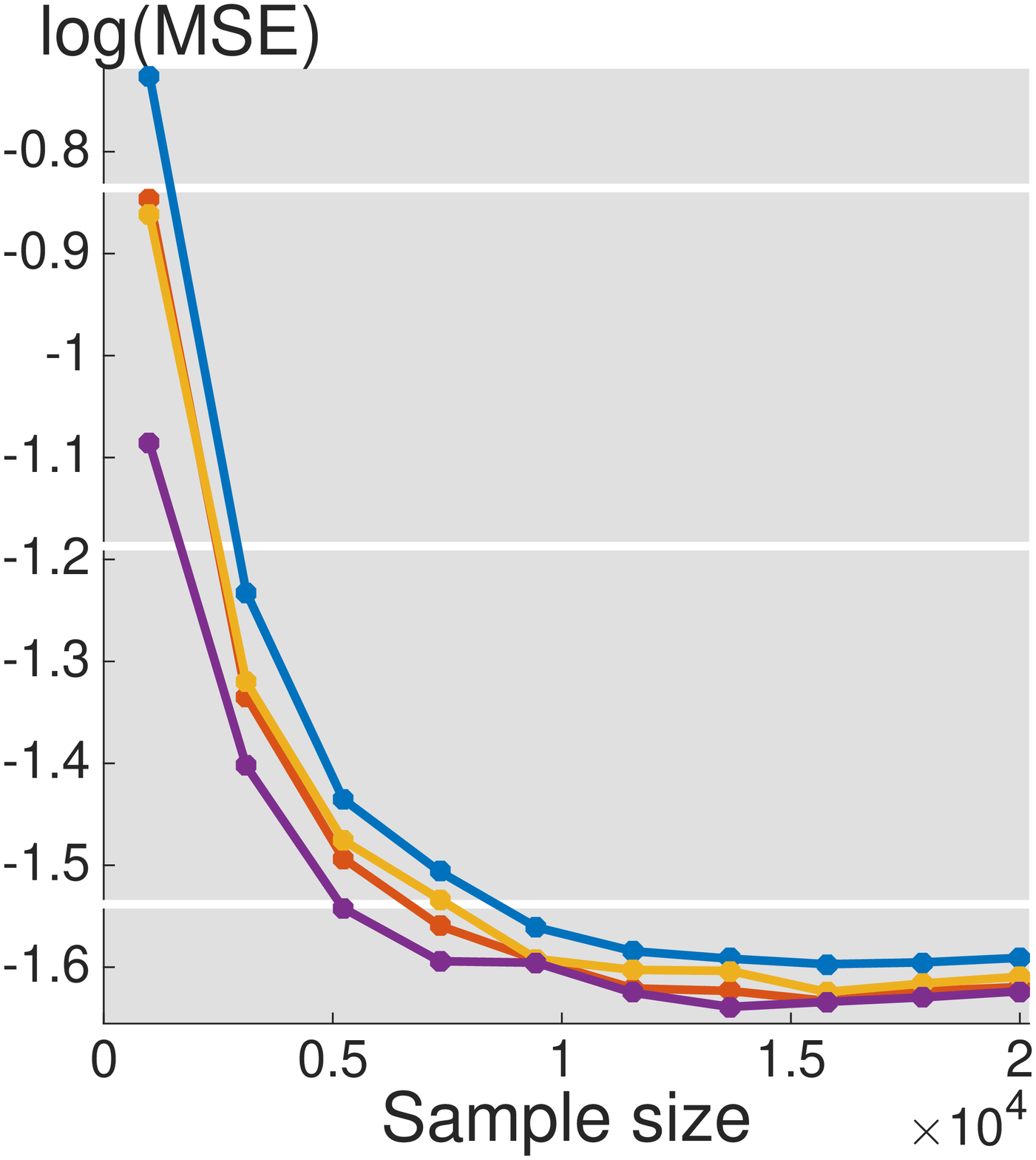}  & \includegraphics[width=0.45\columnwidth]{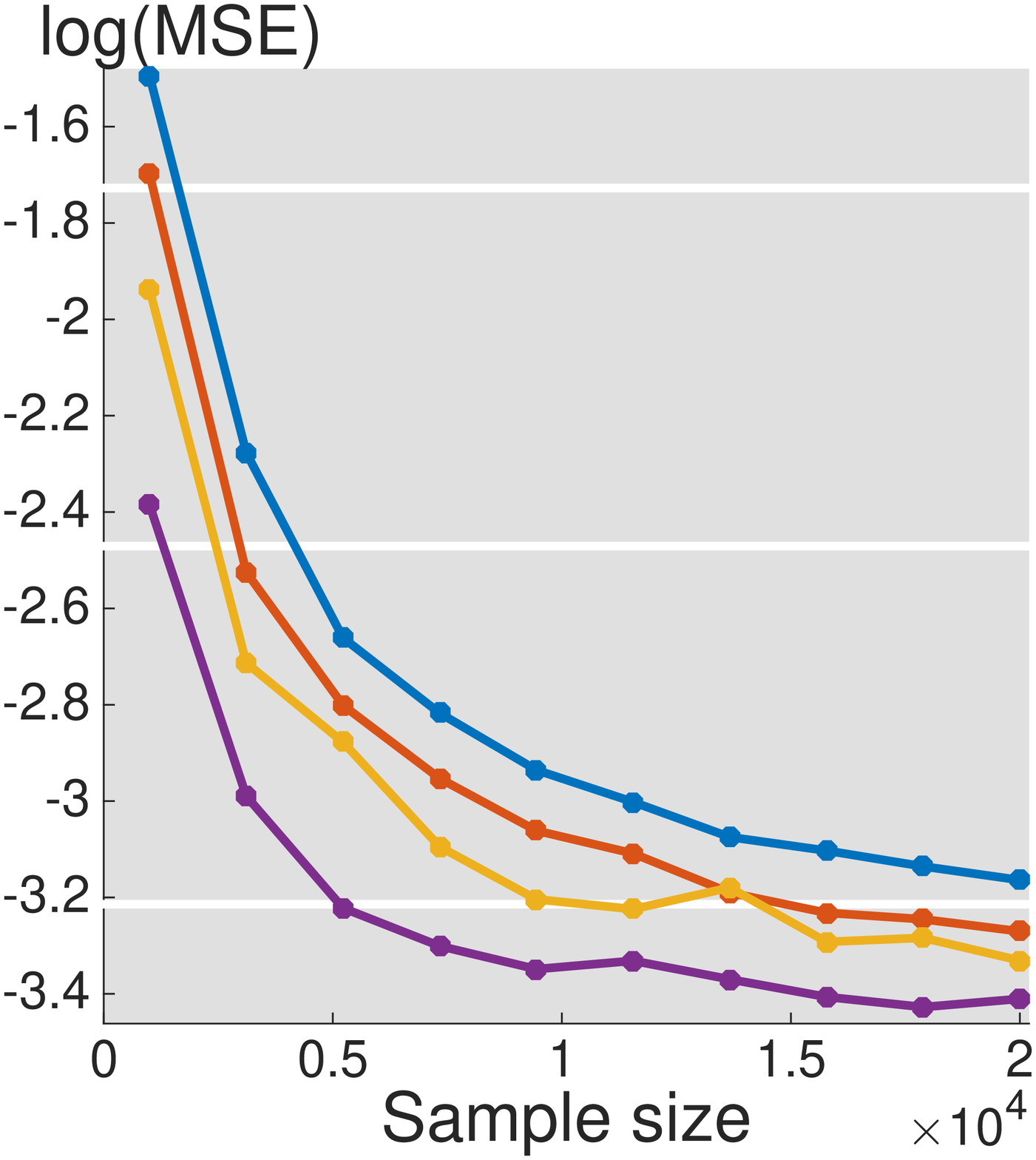} 
\\
      {\small (c)  $\beta = 0.5, \sigma^2 = 2 \times 10^{-2}$. } & {\small (d) $\beta = 1, \sigma^2 = 2 \times 10^{-2}$. }
 \\  
\end{tabular}
  \end{center}
  \caption{\label{fig:Geo_Err} MSE comparison for the random geometric
    graph for uniform sampling (in blue), leverage score based
    sampling (in orange), square root of the leverage score based
    sampling (in purple) and degree based sampling (in red).}
  \vspace{-5mm}
\end{figure}

\begin{figure}[htb]
  \begin{center}
    \begin{tabular}{cc}
\includegraphics[width=0.45\columnwidth]{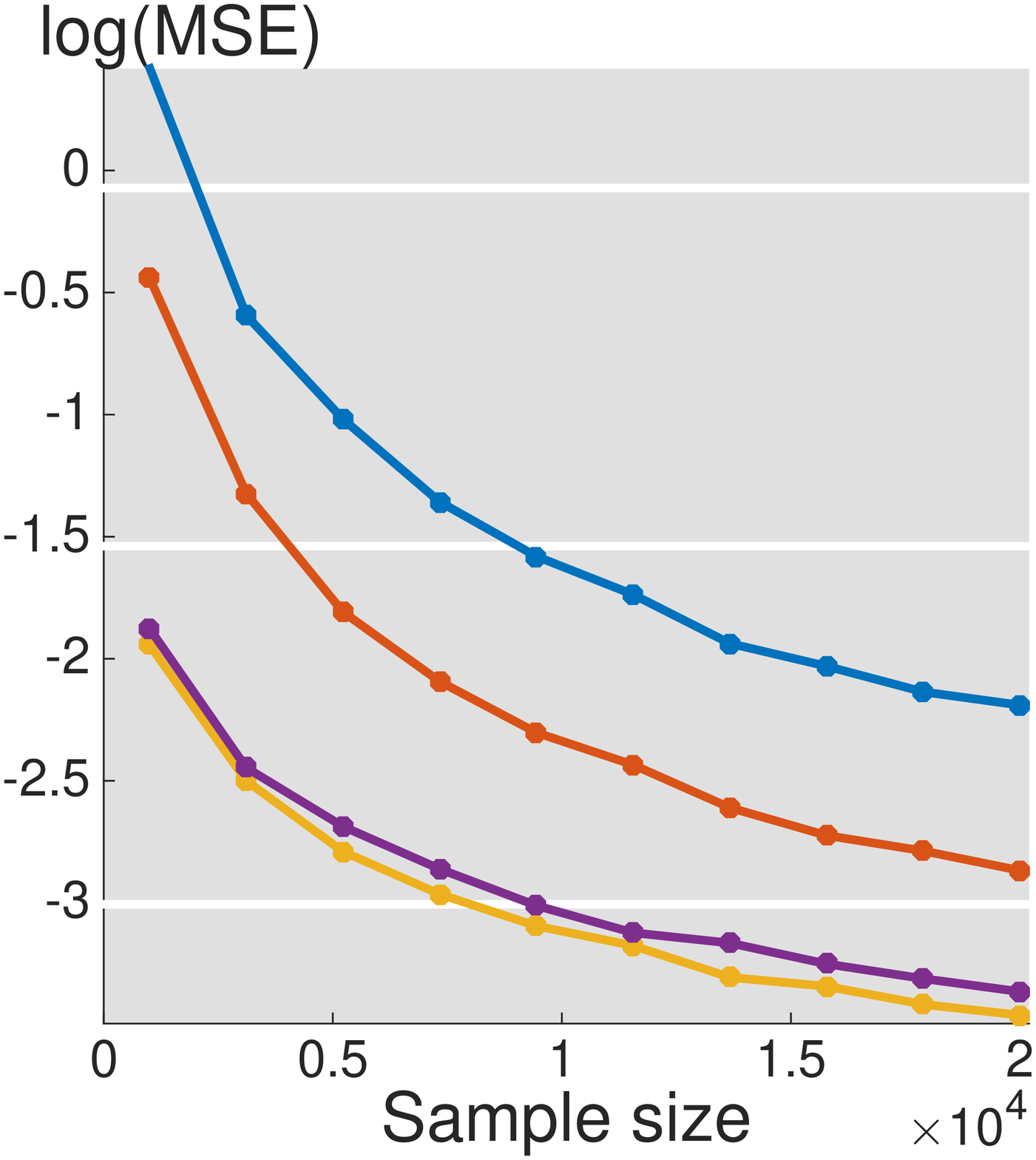}  & \includegraphics[width=0.45\columnwidth]{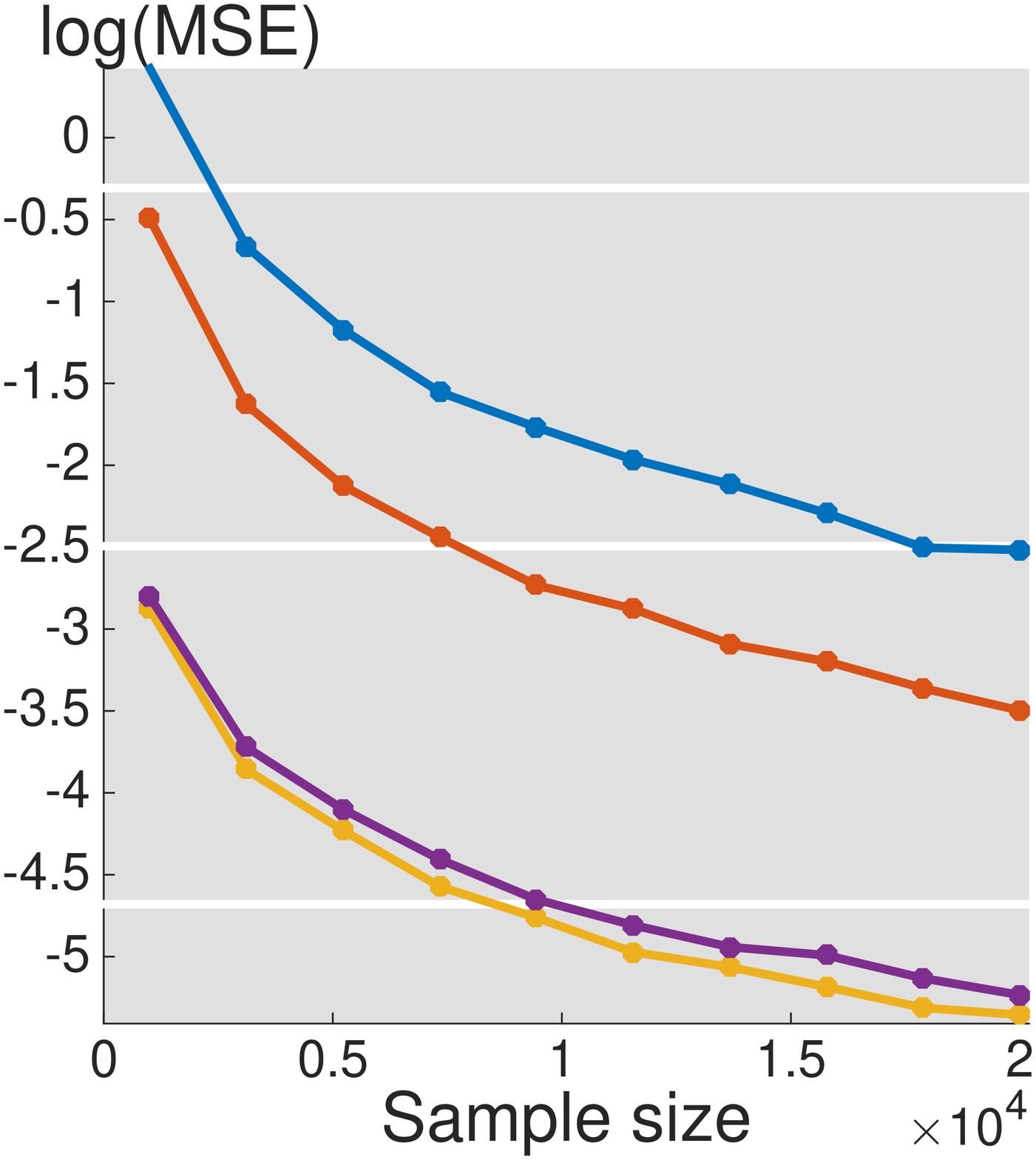} 
\\
      {\small (a)  $\beta = 0.5, \sigma^2 = 10^{-4}$. } & {\small (b) $\beta = 1, \sigma^2 = 10^{-4}$. } 
\\
\includegraphics[width=0.45\columnwidth]{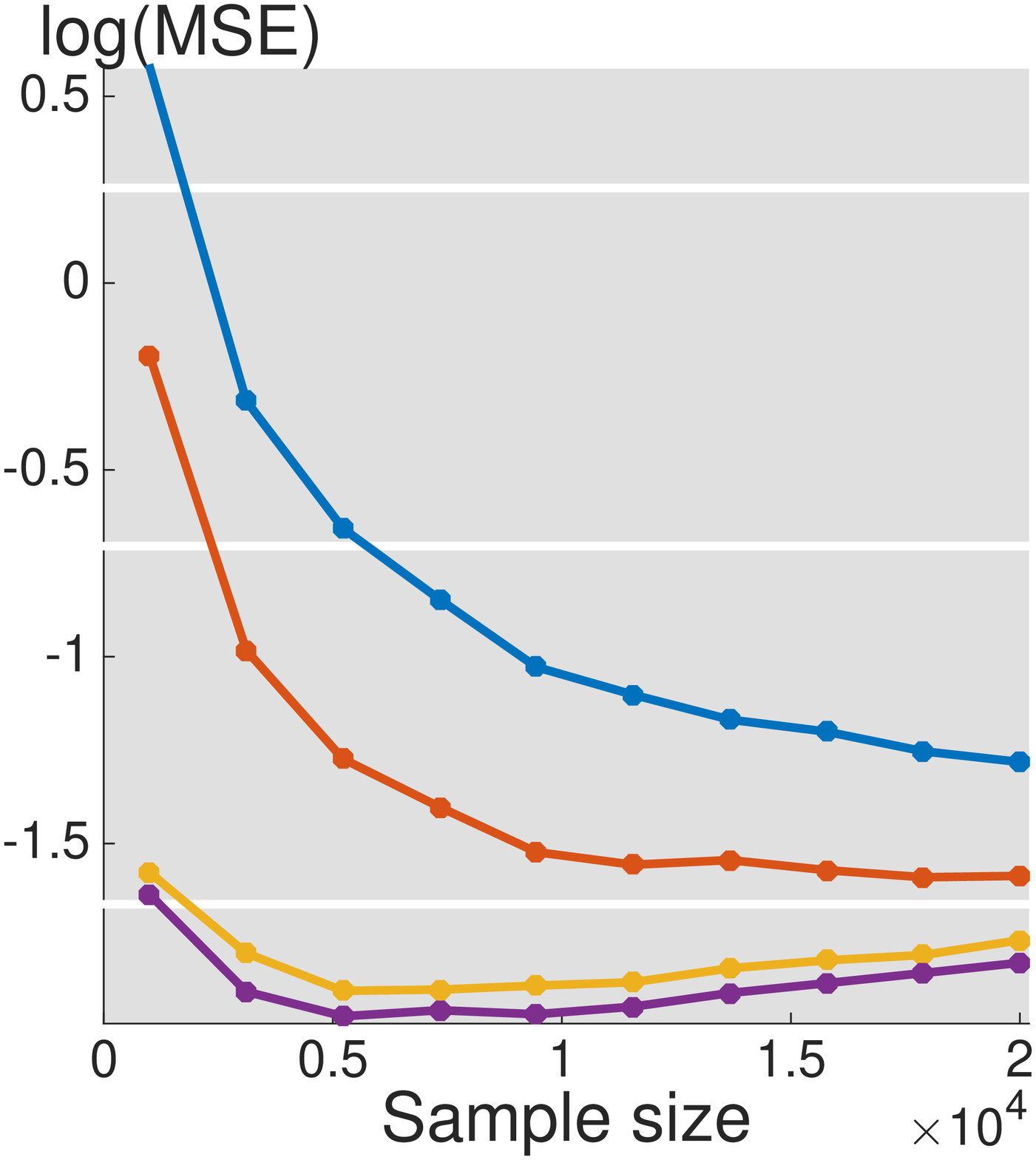}  & \includegraphics[width=0.45\columnwidth]{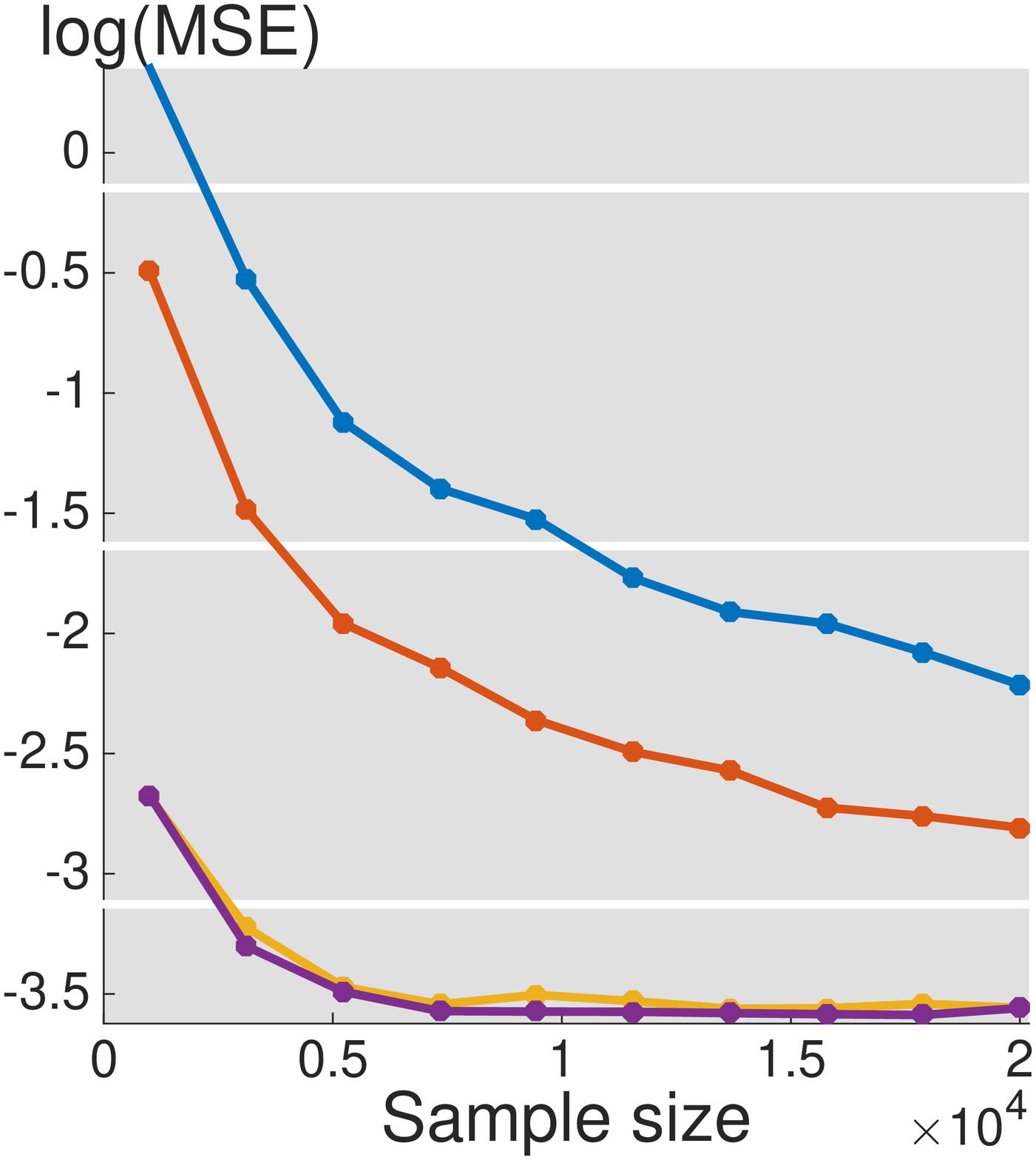} 
\\
      {\small (c)  $\beta = 0.5, \sigma^2 = 2 \times 10^{-2}$. } & {\small (d) $\beta = 1, \sigma^2 = 2 \times 10^{-2}$. }
 \\  
\end{tabular}
  \end{center}
  \caption{\label{fig:SW_Err} MSE comparison for the small-world graph
    for uniform sampling (in blue), leverage score based sampling (in
    orange), square root of the leverage score based sampling (in
    purple) and degree based sampling (in red).}
\end{figure}

\begin{figure}[htb]
  \begin{center}
    \begin{tabular}{cc}
\includegraphics[width=0.45\columnwidth]{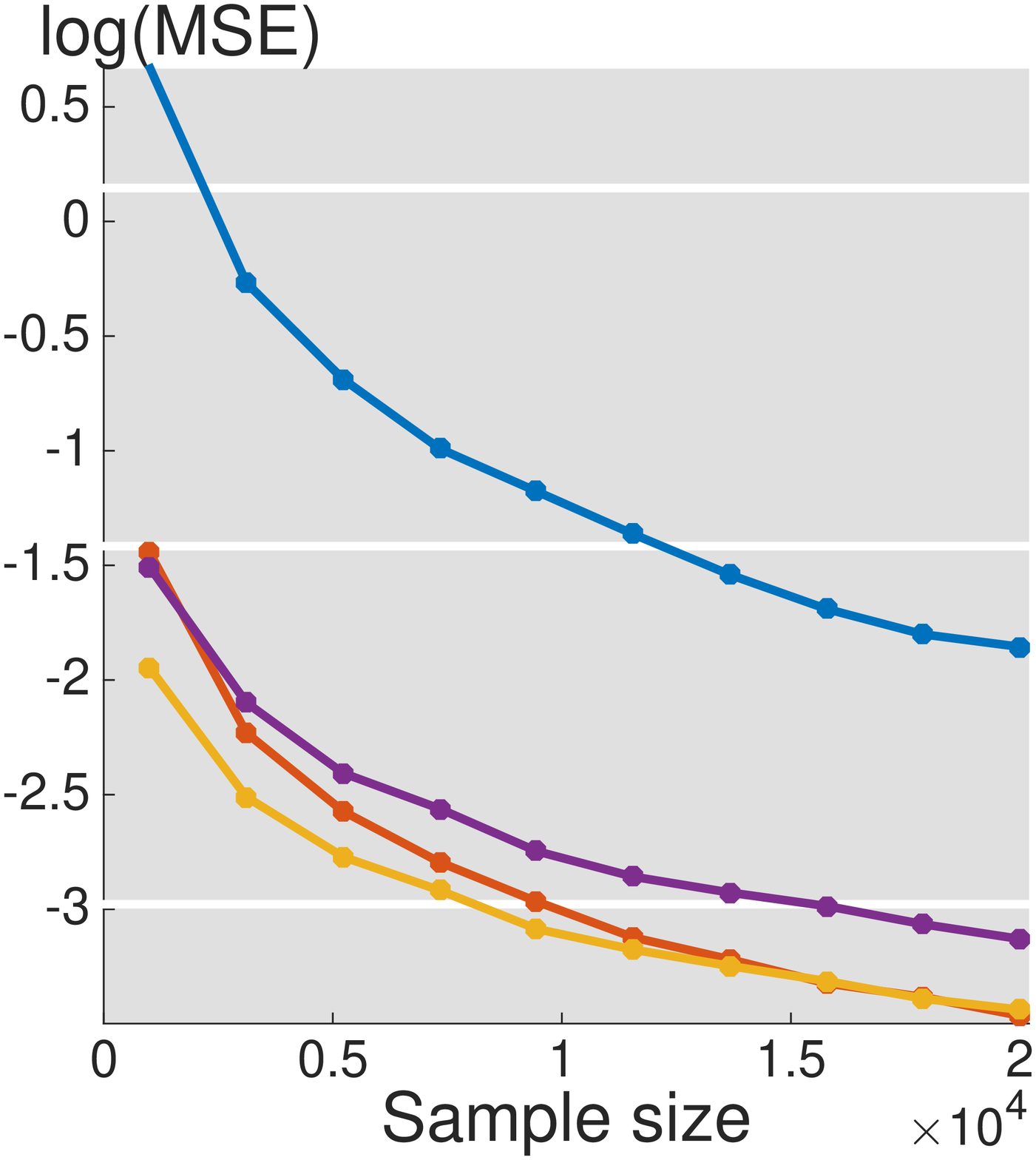}  & \includegraphics[width=0.45\columnwidth]{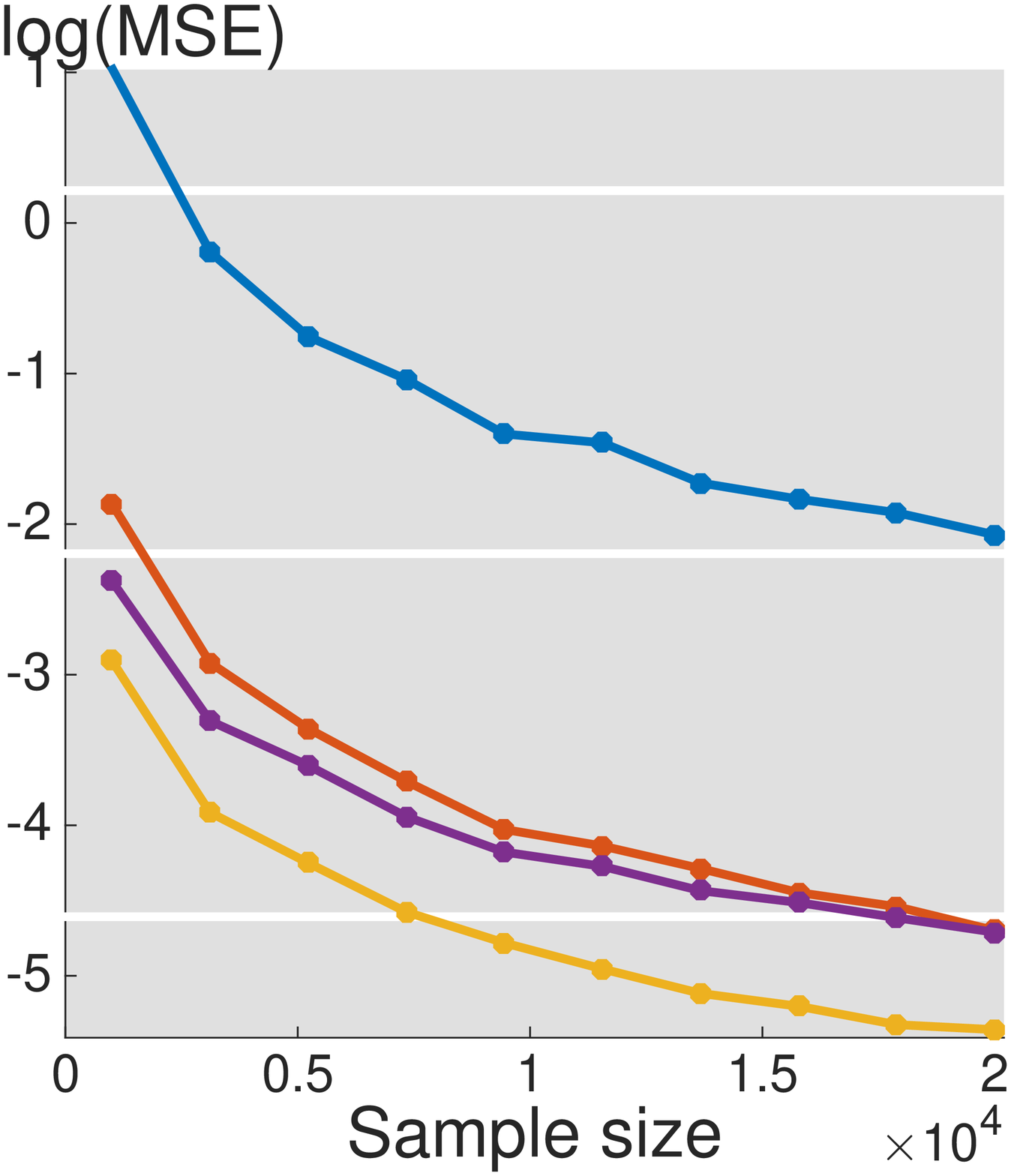} 
\\
      {\small (a)  $\beta = 0.5, \sigma^2 = 10^{-4}$. } & {\small (b) $\beta = 1, \sigma^2 = 10^{-4}$. } 
\\
\includegraphics[width=0.45\columnwidth]{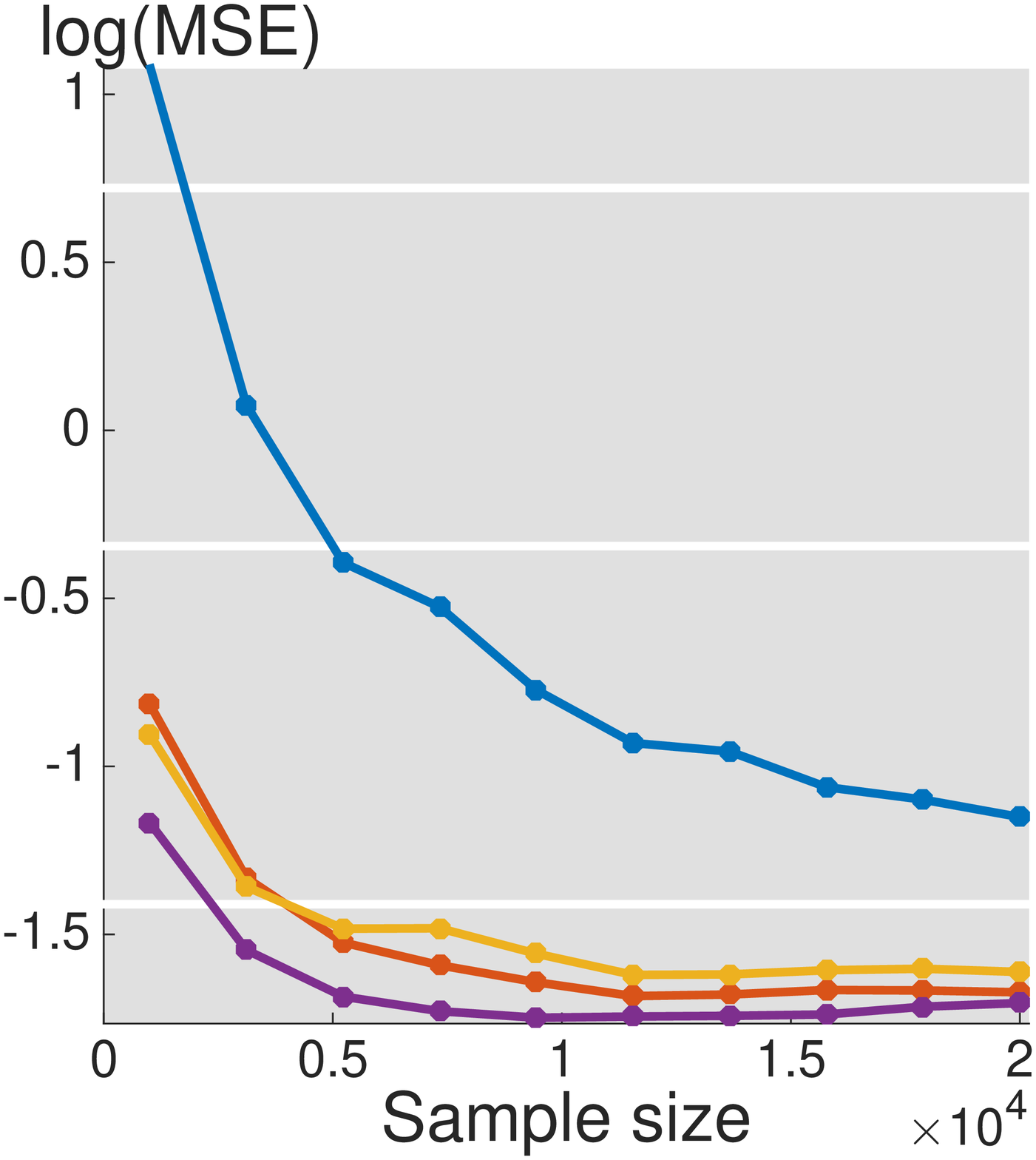}  & \includegraphics[width=0.45\columnwidth]{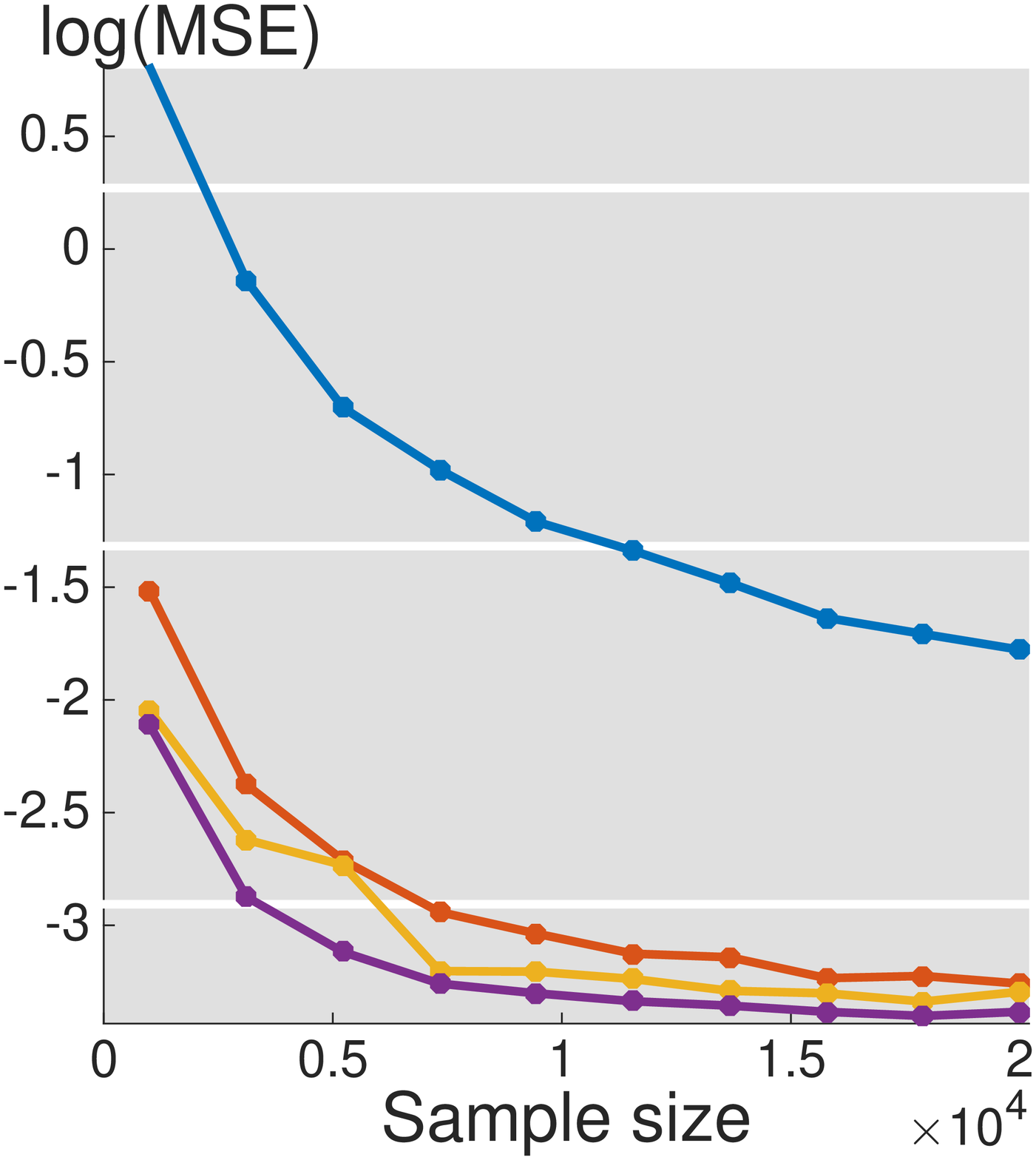} 
\\
      {\small (c)  $\beta = 0.5, \sigma^2 = 2 \times 10^{-2}$. } & {\small (d) $\beta = 1, \sigma^2 = 2 \times 10^{-2}$. }
 \\  
\end{tabular}
  \end{center}
  \caption{\label{fig:PL_Err} MSE comparison for the power-law graph
    for uniform sampling (in blue), leverage score based sampling (in
    orange), square root of the leverage score based sampling (in
    purple) and degree based sampling (in red).}
\end{figure}

\subsection{Discussion}

Graph Fourier basis is critical for understanding the graph
  structure. For example, given the graph Fourier basis,
  Theorem~\ref{thm:lower} shows that active sampling does not
  fundamentally outperform experimentally designed sampling while
  Corollary~\ref{cor:type2_opt} shows that experimentally designed
  sampling fundamentally outperforms uniform sampling on type-2
  graphs. In other words, graph Fourier basis is critical for
  choosing samples. 

For irregular (Type-2) graphs, using experimentally designed sampling
to choose anchor points can fundamentally aid semi-supervised learning
(not true for regular, Type-1 graphs), a technique for training
classifiers with both labeled and unlabeled data. Semi-supervised
learning assumes that unlabeled data can provide distribution
information to build a stronger classifier~\cite{Zhu:05}. Many
algorithms for semi-supervised learning are based on graphs that are
constructed from a given dataset~\cite{Zhu:05}, often by modeling each
node as a data sample and connecting two nodes by an edge if the
distance between their features is in a given range, which is similar
to the construction of random geometric graphs. Based on the
assumption that adjacent nodes have similar labels, semi-supervised
learning diffuses label probabilities from labeled data to unlabeled
data along the graph structure and classifies unlabeled data according
to those label probabilities. While in some
works~\cite{ZhuLG:03,BelkinN:03}, training data is selected uniformly
and randomly, in others, algorithms are designed adapted to the
structure~\cite{GuH:12, GaddeAO:14}, which is essentially equivalent
to experimentally designed sampling we propose. In other words,
experimentally designed sampling is used implicitly without being able
to articulate when and why it works; this paper, on the other hand,
provides a comprehensive explanation of why experimentally designed
sampling helps semi-supervised learning through showing the lower and
upper bounds of three sampling strategies.

\section{Conclusions}
\label{sec:conclusions}

We build a theoretical foundation for the recovery of a newly proposed class of smooth graph signals,
  approximately bandlimited signals, which generalizes the class of bandlimited graph signals, under uniform sampling,
experimentally designed sampling and active sampling. We show that experimentally designed sampling and
active sampling have the same fundamental limitations, and can
outperform uniform sampling on irregular graphs. We propose a recovery
strategy and analyze its statistical
properties. We show that the proposed recovery strategy attains the
optimal rates of convergence on two specific types of graphs. To
validate the recovery strategy, we test it on five specific types of graphs: a ring graph
with $k$ nearest neighbors, an Erd\H{o}s-R\'enyi graph, a random
geometric graph, a small-world graph and a power-law graph, and show
that experimental results match the proposed theory well. This work also gives a comprehensive explanation for why experimentally
designed sampling works for semi-supervised learning with graphs
and shows the critical role of the graph Fourier basis in
analyzing graph structures.

\section{Acknowledgment}
We gratefully acknowledge support from the NSF through awards 1130616,1421919,  the University Transportation Center grant (DTRT12-GUTC11) from the US Department of Transportation, and the CMU Carnegie Institute of Technology Infrastructure Award. We also thank the editor and the reviewers for comments that led to improvements in the manuscript. Initial parts of this work were presented at SampTA 2015~\cite{ChenVSK:15a}.

\bibliographystyle{IEEEbib}
\bibliography{bibl_jelena}

\pagebreak

\appendix
\section{Appendices}
\subsection{Proof of Theorem 3}
\label{sec:app1}
We aim to construct a typical set of vectors in $\F$, and use the
Fano's method. Let $\vv_k$ be the $k$th column of $\Vm$, $\vv^{(i)}$
be the $i$th row of $\Vm$, $ \mathcal{X}$ be a pruned hypercube and
\begin{displaymath}
  \F' = \{ \x^{(\w)} =  \Vm  \xhat \odot \w  = \sum_{k = \kappa_0}^{2 \kappa_0-1} w_k \psi_k, ~~\w \in \mathcal{X} \},
\end{displaymath}
where $\kappa_0$ is no smaller than the bandwidth $K$,
\begin{displaymath}
  \psi_k = \widehat{x}_k \vv_k  =  (\pm)^k \sqrt{ \frac{c \mu \left\| \x \right\|_2^2 }{ \kappa_0 (1+k^{2\beta} ) }}  \vv_k,
\end{displaymath}
and $0 < c < 1$. It is easy to check that $\F' \subseteq \F $.  Let $d
(\x, \y ) = \left\| \x - \y \right\|_2$; we thus have
\begin{eqnarray}
  d^2(\x^{(\w)} ,  \x^{(\u)}) & = & \left\|  \Vm  \xhat \odot ( \w  - \u )^2 \right\|_2^2
  \nonumber \\  \nonumber
  & = &  \sum_{k = \kappa_0}^{2 \kappa_0-1} (w_k - u_k)^2 \widehat{x}_k^2
  \nonumber \\  \nonumber
  & = &  \sum_{k = \kappa_0}^{2 \kappa_0-1} (w_k - u_k)^2 \frac{c \mu \left\| \x  \right\|_2^2 }{ \kappa_0 (1+k^{2\beta} ) }
  \nonumber \\  \nonumber
  & \stackrel{(a)}{\geq}  &   \sum_{k = \kappa_0}^{2 \kappa_0 -1} (w_k - u_k)^2  \cdot  \frac{c \mu  \left\| \x  \right\|_2^2 }{ \kappa_0 (1+ (2 \kappa_0 )^{2\beta} ) } 
  \nonumber \\  \nonumber
  & \stackrel{(b)}{\geq} &  \frac{  \kappa_0}{8} \cdot \frac{c \mu  \left\| \x  \right\|_2^2}{ \kappa_0 (1+ (2 \kappa_0 )^{2\beta} ) }
  \ \geq \  c_1 \mu \kappa_0^{-2\beta}   \left\| \x  \right\|_2^2,
\end{eqnarray}
where $(a)$ follows from $k \leq 2K$, and $(b)$ from the
Varshamov-Gilbert lemma.  To use the Fanno's method, we need to bound
the Kullback-Leibler divergence,
\begin{eqnarray*}
  && KL(p_{\w} , p_{\w_0} | \M) 
  \\
  & = & \sum_{i \in \M} \mathbb{E}_{\w} \left[  \log \frac{ p(y_i - x^{(\w)}_i) }{ p(y_i - x^{(\w_0)}_i) }   \right]
  \nonumber \\  \nonumber
  & \leq & \sum_{i \in \M} \left[  \frac{1}{2 \sigma^2} ( x^{(\w)}_i - x^{(\w_0)}_i )^2 \right]
  \\
  & = & \sum_{i \in \M} \left[  \frac{1}{2 \sigma^2} \left( {\vv^{(i)}}^* ( \xhat \odot \w ) \right)^2 \right]
  \nonumber \\  \nonumber
  & = & \frac{1}{2 \sigma^2} \sum_{i \in \M} \left(  \sum_{k= \kappa_0}^{2 \kappa_0 -1} \widehat{x}_k w_k \Vm_{ik}  \right)^2
  \nonumber \\  \nonumber
  & = & \frac{1}{2 \sigma^2} \sum_{i \in \M} \bigg(     \sum_{k= \kappa_0}^{2 \kappa_0-1}   \left( \widehat{x}_k w_k \Vm_{ik} \right)^2   
  \\
  && +  \sum_{k, k'= \kappa_0, k \neq k'}^{2 \kappa_0-1}   \widehat{x}_k \widehat{x}_{k'}   w_k w_{k'} \Vm_{ik} \Vm_{ik'}   \bigg)  
  \nonumber \\  \nonumber
  & \leq & \frac{1}{2 \sigma^2} \sum_{i \in \M} \bigg( \sum_{k= \kappa_0}^{2 \kappa_0-1}  \frac{c \mu \left\| \x \right\|_2^2  w_k^2 \Vm_{ik}^2 }{ 2\kappa_0 (1+k^{2\beta} ) }   
\end{eqnarray*}
\begin{eqnarray*}
  && +  \sum_{k, k'= \kappa_0, k \neq k'}^{2 \kappa_0-1}  (-1)^{k+k'} \frac{c \mu \left\| \x \right\|_2^2  w_k w_{k'} \Vm_{ik} \Vm_{ik'}  }{ \kappa_0 \sqrt{ 1+k^{2\beta} }  \sqrt{ 1+k'^{2\beta} }  }   \bigg)  
  \nonumber \\  \nonumber
  & \leq  &  \frac{c' \mu \left\| \x \right\|_2^2 }{\sigma^2} \sum_{i \in \M} \sum_{k= \kappa_0}^{2 \kappa_0-1}   \frac{\Vm_{ik}^2}{ \kappa_0 (1+\kappa_0^{2\beta} ) }    + \delta
  \nonumber \\  \nonumber
  & \asymp  & \frac{c' \mu \left\| \x \right\|_2^2}{ \sigma^2}  \kappa_0^{-(2\beta+1)} \sum_{i \in \M} \sum_{k= \kappa_0}^{2\kappa_0-1} \Vm_{ik}^2,
\end{eqnarray*}
where $\delta = \sum_{k, k'= \kappa_0, k \neq k'}^{2 \kappa_0-1}  (-1)^{k+k'} \frac{c \mu  \left\| \x \right\|_2^2 w_k w_{k'} \Vm_{ik} \Vm_{ik'}  }{ \kappa_0 \sqrt{ 1+k^{2\beta} }  \sqrt{ 1+k'^{2\beta} }} $ is small because of the cross signs.
For uniform sampling, the sampling set $\M$ is chosen randomly, thus, we have
\begin{eqnarray*}
  && KL(p_{\w} , p_{\w_0}  ) 
  \\
  & = &  \mathbb{E}_\M \left[ KL(p_{\w} , p_{\w_0}  | \M) \right]
  \\
  & \leq &  \frac{c \mu \left\| \x \right\|_2^2}{ \sigma^2}  \kappa_0^{-(2\beta+1)} \mathbb{E}_\M \left( \sum_{i \in \M} \sum_{k= \kappa_0}^{2 \kappa_0-1} \Vm_{ik}^2  \right) 
  \\
  & \stackrel{(a)}{=} &  \frac{c \mu \left\| \x \right\|_2^2}{ \sigma^2}  K^{-(2\beta+1)} \mathbb{E}_i \left( m \sum_{k= \kappa_0}^{2\kappa_0-1} \Vm_{ik}^2  \right) 
  \\
  & = & \frac{c \mu \left\| \x \right\|_2^2}{ \sigma^2}  K^{-(2\beta+1)} m \sum_{j=1}^N  \sum_{k= \kappa_0}^{2 \kappa_0-1} \Vm_{jk}^2 \mathbb{P} (j = i) 
  \\
  & = & \frac{c \mu \left\| \x \right\|_2^2}{ \sigma^2}  \kappa_0^{-(2\beta+1)} \frac{ m }{N} \sum_{j=1}^N  \sum_{k= \kappa_0}^{2\kappa_0-1} \Vm_{jk}^2 
  \\
  & \leq & \frac{c \mu \left\| \x \right\|_2^2 }{ \sigma^2 \kappa_0^{2\beta+1} N }   \left\| \Vm_{(2,\kappa_0)} \right\|_F^2  m,
\end{eqnarray*}
where $(a)$ follows from the independence of each sample,
$\mathbb{P}(j = i) $ denotes the probability to sample the $i$th node
that equals $j$, and $\left\| \Vm_{(2,\kappa_0)} \right\|_F^2 =
\sum_{j=1}^N \sum_{k= \kappa_0}^{2\kappa_0-1} \Vm_{jk}^2$.
 
For experimentally designed sampling, we can choose the sampling set
$\M$ to maximize $\sum_{i \in \M} \sum_{k= \kappa_0}^{2\kappa_0-1}
\Vm_{ik}^2$; we thus have
\begin{eqnarray*}
  KL(p_{\w} , p_{\w_0}) 
  & \leq  & \frac{c \mu \left\| \x \right\|_2^2 }{ \sigma^2}  \kappa_0^{-(2\beta+1)} \max_\M \sum_{i \in \M} \sum_{k= \kappa_0}^{2\kappa_0-1} \Vm_{ik}^2 
  \\
  & \leq  &   \frac{c \mu \left\| \x \right\|_2^2}{ \sigma^2}  \kappa_0^{-(2\beta+1)}  \left\| \Vm_{(2,\kappa_0)} \right\|_{2, \infty}^2 m.
\end{eqnarray*}
For active sampling, we cannot get more benefit from signal
coefficients, so the KL divergence is the same of the experimentally
designed sampling.r By Fanno's lemma, we finally have the lower bounds
for three sampling strategies.  \hfill$\blacksquare$

\subsection{Proof of Theorem 4}
\label{sec:app2}
We aim to bound the MSE by splitting to a bias term and a variance
term.
\begin{eqnarray*}
  && \mathbb{E} \left\| \x^* - \x \right\|_2^2
  \\
  & = & \mathbb{E} \left\| \x^* - \mathbb{E} \x^* + \mathbb{E} \x^* - \x \right\|_2^2
  \\
  & = & \mathbb{E} \left(  \left\| \x^* - \mathbb{E} \x^* \right\|_2^2  +  \left\| \mathbb{E} \x^* - \x \right\|_2^2  +  2 ( \x^* - \mathbb{E} \x^* )^T  (\mathbb{E} \x^* -  \x )  \right)
  \\
  & = & \left\| \mathbb{E} \x^* - \x \right\|_2^2  + \mathbb{E} \left\| \x^* - \mathbb{E} \x^* \right\|_2^2,
\end{eqnarray*}
where the first term is bias and the second term is variance.  For
each element in the bias term, we have
\begin{eqnarray*} 
  \mathbb{E}  x^*_i
  & = &   \sum_{k < \kappa} \Vm_{ik}  \mathbb{E}_{\M, \epsilon}  \left( \sum_{\M_j \in \M} \Um_{k \M_j}  \Dd^2_{\M_j,\M_j}  ( x_{\M_j} + \epsilon_{\M_j}) \right)
  \\
  & \stackrel{(a)}{=} &  \sum_{k < \kappa} \Vm_{ik} m  \mathbb{E}_{\ell}  \left( \Um_{k \ell}  \frac{1}{m \pi_{l}}  x_{\ell} \right)
  \\
  & = &  \sum_{k < \kappa} \Vm_{ik} m \sum_{\ell = 1}^N    \left( \Um_{k \ell}  \frac{1}{m \pi_{\ell}} x_{\ell} \right) \pi_{\ell}
  \\
  & = &  \sum_{k < \kappa} \Vm_{ik}  \sum_{\ell = 1}^N  \Um_{k \ell}  x_{\ell} 
  \\
  & = &  \sum_{k < \kappa} \Vm_{ik}  \widehat{x}_k,
\end{eqnarray*}
where $(a)$ follows from the independence of each sample. For all the
elements, we have
\begin{eqnarray*} 
  \mathbb{E}  \x^*  = \Vm_{(\kappa)} \widehat{\x}_{(\kappa)},
\end{eqnarray*}
which leads to Lemma 1.

We next bound the variance term by splitting into two parts, with and
without noise. For each element in the variance term, we have
\begin{eqnarray*} 
  &&  x_i^* - \mathbb{E} x_i^* 
  \\
  & = & \sum_{k < \kappa} \Vm_{ik} \left( \sum_{\M_j \in \M}  \Um_{k \M_j} \Dd^2_{\M_j, \M_j} y_{\M_j} \right) - \sum_{k < \kappa} \Vm_{ik} \widehat{x}_k
  \\
  & = & \sum_{k < \kappa} \Vm_{ik} \sum_{\M_j \in \M}  \Um_{k \M_j} \Dd^2_{\M_j, \M_j} \epsilon_{\M_j}  
  \\  
  && +  \sum_{k < \kappa} \Vm_{ik} \left( \sum_{\M_j \in \M}  \Um_{k \M_j} \Dd^2_{\M_j, \M_j} x_{\M_j}  -  \widehat{x}_k \right)
  \\
  & = & \Delta^{(1)}_i + \Delta^{(2)}_i, 
 \end{eqnarray*}
where  $\Delta^{(1)}$ is the variance from noise and $\Delta^{(2)}$ is the variance from sampling. To bound $\Delta^{(1)}_i $, we have
\begin{eqnarray*}
  && \mathbb{E} || \Delta^{(1)}_i ||^2  
  \\
  & = & \mathbb{E}  \bigg[  \bigg(  \sum_{k < \kappa} \Vm_{ik} \sum_{\M_j \in \M}  \Um_{k \M_j} \W_{\M_j, \M_j} \epsilon_{\M_j}  \bigg)  
  \\
  &&  \bigg(  \sum_{k' < \kappa} \Vm_{ik'} \sum_{\M_{j'} \in \M}  \Um_{k' \M_{j'}} \Dd^2_{\M_{j'}, \M_{j'}} \epsilon_{\M_{j'}}  \bigg)     \bigg]
  \\
  & = & \mathbb{E}  \bigg( \sum_{k, k' < \kappa}  \Vm_{ik} \Vm_{ik'}  \sum_{\M_j, \M_{j'} \in \M}  \Um_{k \M_j}  
  \\
  && \Um_{k' \M_{j'}} \Dd^2_{\M_j,\M_j} \Dd^2_{\M_{j'},\M_{j'}}  \epsilon_{\M_j} \epsilon_{\M_{j'}}  \bigg) 
  \\
  & = &  \sum_{k, k' < \kappa}  \Vm_{ik} \Vm_{ik'} m \mathbb{E}_{\ell, \epsilon} \left( \Um_{k \ell}^2  \frac{1}{m^2 \pi_{\ell}^2 }  \epsilon^2_{\ell} \right) 
  \\
  & = &  \sum_{k, k' < \kappa}  \Vm_{ik} \Vm_{ik'} m \sum_{\ell = 1}^N  \Um_{k \ell} \Um_{k' \ell}  \frac{1}{m^2 \pi_{\ell}^2 }  \mathbb{E}  \epsilon^2_{\ell}    \pi_{\ell}
  \\
  & = &  \sigma^2 \sum_{k, k' < \kappa}  \Vm_{ik} \Vm_{ik'}  \sum_{\ell = 1}^N   \frac{1 }{m \pi_{\ell} } \Um_{k \ell} \Um_{k' \ell}.
 \end{eqnarray*}
 
 To bound $\Delta^{(2)}_i $, we have
\begin{eqnarray*} 
  && \mathbb{E} \left\|   \Delta_i^{(2)}  \right\|_2^2
  \\
  & = &  \mathbb{E}  \Bigg[      \sum_{k < \kappa} \Vm_{ik} \left( \sum_{\M_j \in \M}  \Um_{k \M_j} \Dd^2_{\M_j, \M_j} x_{\M_j}  -  \widehat{x}_k \right)   
  \\
  &&~~~~   \sum_{k' < \kappa} \Vm_{ik'} \left( \sum_{\M_{j'} \in \M}  \Um_{k' \M_{j'}} \Dd^2_{\M_{j'}, \M_{j'}} x_{\M_{j'}}  -  \widehat{x}_{k'} \right)   \Bigg]
  \\
  & = &  \mathbb{E}  \Bigg[    \sum_{k, k' < K}  \Vm_{ik} \Vm_{ik'} \bigg(  \sum_{\M_j, \M_{j'} \in \M}  \Um_{k \M_j} \Um_{k' \M_{j'}}  \Dd^2_{\M_j, \M_j}  
  \\
  && \Dd^2_{\M_{j'}, \M_{j'} }   x_{\M_j} x_{\M_{j'}} -  \widehat{x}_k \widehat{x}_{k'}  \bigg)    \Bigg]
  \\
  & = &    \sum_{k, k' < \kappa}  \Vm_{ik} \Vm_{ik'}  \Bigg[   \mathbb{E}_\M \bigg( \sum_{\M_j \neq \M_{j'}, \M_j, \M_{j'} \in \M} \Um_{k \M_{j}} \Um_{k' \M_{j'}}  
  \\
  &&  x_{\M_{j}} x_{\M_{j'}} \bigg)  + \mathbb{E}_\M \bigg( \sum_{\M_{j} = \M_{j'}, \M_{j}, \M_{j'} \in \M} \Um_{k \M_{j}} \Um_{k' \M_{j'}}  
  \\
  &&   x_{\M_{j'}} x_{\M_{j'}} \bigg) -  \widehat{x}_k \widehat{x}_{k'}  \Bigg]
  \\
  & = &    \sum_{k, k' < \kappa}  \Vm_{ik} \Vm_{ik'}   \Bigg[  ( m^2 - m )\mathbb{E}_{\ell, \ell'}   \left( \frac{\Um_{k \ell} \Um_{k' \ell'} x_{\ell} x_{\ell'}  }{m^2 \pi_{\ell} \pi_{\ell '}   }  \right) 
  \\
  && + m \mathbb{E}_{\ell} \left( \frac{ \Um_{k \ell}  \Um_{k' \ell}  x_{\ell}^2 } { m^2 \pi_{\ell}^2 } \right) -  \widehat{x}_k \widehat{x}_{k'}  \Bigg]
  \\
  & = &    \sum_{k, k' < \kappa}  \Vm_{ik} \Vm_{ik'}   \Bigg[ ( m^2 -  m )  \sum_{\ell, \ell' = 1}^N  \frac{ \Um_{k \ell}  \Um_{k \ell'}   x_\ell x_{\ell'} }{ m^2 \pi_\ell \pi_{\ell'}  }   \pi_\ell \pi_{\ell'} 
\end{eqnarray*}
\begin{eqnarray*}
  && + m \sum_{\ell = 1}^N  \frac{ \Um_{k \ell}  \Um_{k' \ell}  x_{\ell}^2 } { m^2 \pi_{\ell}^2 } \pi_{\ell}  -  \widehat{x}_k \widehat{x}_{k'} \Bigg]
  \\
  & = &   \sum_{k, k' < \kappa}  \left( \Vm_{ik} \Vm_{ik'} \sum_{\ell = 1}^N  \Um_{k \ell}  \Um_{k' \ell}  \frac{ x_{\ell}^2 } { m \pi_{\ell} }  - 
    \frac{1}{m} \widehat{x}_k \widehat{x}_{k'} \right).
 \end{eqnarray*}
 
 We combine the bounds for both $\Delta^{(1)}_i $ and $\Delta^{(2)}_i
 $, and obtain the bounds for the variance term,
\begin{eqnarray*} 
  && \mathbb{E} \left\|     \x^* - \mathbb{E} \x^*  \right\|_2^2 
  \\
  & = & \sum_{i=1}^{N}    \mathbb{E} \left\|   x_i^* -  \mathbb{E}  x_i^*  \right\| 
  \\ 
  & = & \sum_{i=1}^{N}   \left(  \mathbb{E} \left\|   \Delta_i^{(1)}  \right\|_2^2 +  \mathbb{E} \left\|   \Delta_i^{(2)}  \right\|_2^2 \right)
  \\ 
  & = &   \sum_{i=1}^{N}  \sum_{k, k' < \kappa}  \Vm_{ik} \Vm_{ik'}  \sum_{\ell = 1}^N   \Um_{k \ell} \Um_{k \ell'}    \frac{ \sigma^2 + x_{\ell}^2 }{m \pi_{\ell} }  - \frac{1}{m} \sum_{k, k' < \kappa} \widehat{x}_k \widehat{x}_{k'}
  \\
  & = & {\rm Tr} \left( \Vm_{(\kappa)} \Um_{(\kappa)}  \W_C \Um_{(\kappa)}^T \Vm_{(\kappa)}^T \right) - \frac{1}{m} \left\|  \widehat{\x}_{(\kappa)} \right\|_2^2
  \\
  & = & {\rm Tr} \left( \Um_{(\kappa)}  \W_C \Vm_{(\kappa)} \right) - \frac{1}{m} \left\|  \widehat{\x}_{(\kappa)} \right\|_2^2,
  \end{eqnarray*}
  which leads to Lemma 2. Finally, we obtain the MSE by combine the
  bias term and the variance term,
  \begin{eqnarray*}
    && \mathbb{E} \left\| \x^* - \x \right\|_2^2
    \\
    & = & \left\| \mathbb{E} \x^* - \x \right\|_2^2  + \mathbb{E} \left\| \x^* - \mathbb{E} \x^* \right\|_2^2
    \\
    & = &  \left\| \Vm_{(-\kappa)} \widehat{x}_{(-\kappa)} \right\|_2^2 +   {\rm Tr} \left( \Um_{(\kappa)}  \W_C \Vm_{(\kappa)} \right) - \frac{1}{m} \left\|  \widehat{\x}_{(\kappa)} \right\|_2^2
    \\
    & \leq & \frac{1}{1 + \kappa^{2\beta}} \sum_{k \geq \kappa}  \widehat{x}_k^2 (1 + k^{2\beta}) + {\rm Tr} \left( \Um_{(\kappa)}  \W_C \Vm_{(\kappa)} \right)
    \\
    & \leq & \frac{\mu }{1 + \kappa^{2\beta}}  \left\| \x \right\|_2^2 + {\rm Tr} \left( \Um_{(\kappa)}  \W_C \Vm_{(\kappa)} \right).
\end{eqnarray*}
\hfill$\blacksquare$

\subsection{Proof of Corollaries 3 and 4}
\label{sec:app3}
For a Type-1 graph, we assume that each element in an approximately
bandlimited signal has a similar magnitude. Since all the elements in
$\Vm$ have the same magnitude, each element of a graph signal, $x_i =
\vv_i^T \widehat{\x}$, should have a similar magnitude. In other
words, $N\max_i x_i^2$ and $\left\| \x \right\|_2^2$ are of the same
order.

For uniform sampling, based on Corollary 1, we have
\begin{eqnarray*}
  && 
  \frac{\mu }{1 + \kappa^{2\beta}}  \left\| \x \right\|_2^2 + \frac{N}{m} \sum_{k=1}^{\kappa} \sum_{i=1}^N \Um_{k,i}^2  \left( x_i^2 + \sigma^2 \right) 
  \\
  & \leq & \left\| \x \right\|_2^2 \left(   \frac{\mu }{ \kappa^{2\beta}} + \frac{N \left( \max_i x_i^2 + \sigma^2 \right) }{m  \left\| \x \right\|_2^2} 
    \sum_{k=1}^{\kappa} \sum_{i=1}^N \Um_{k,i}^2  \right) 
  \\
  & \stackrel{(a)}{=}  &  \left\| \x \right\|_2^2 \left(   \frac{\mu }{ \kappa^{2\beta}} + \frac{ N (\max_i x_i^2 + \sigma^2) } {\left\| \x \right\|_2^2}  \frac{ \kappa}{m} \right) 
  \\
  & \leq & \left\| \x \right\|_2^2 \left(   \frac{\mu }{ \kappa^{2\beta}} + \frac{ c \kappa}{m}   \right) 
  \\
  & \asymp & C \left\| \x \right\|_2^2 m^{-\frac{2\beta}{2\beta+1}},
\end{eqnarray*}
where $(a)$ follows from $\Um$ being orthornormal, $\kappa$ being of the order of $m^{\frac{1}{2\beta+1}}$, and $C > 0$
some constant.  Since at least the sampled projection estimator
satisfies this rate of convergence, we have
\begin{eqnarray*}
  \inf_{(\x^*, \M) \in \Theta_{\rm u} } \sup_{\x \in \BLT_{\Adj}(K, \beta, \mu)} \mathbb{E}_{\x, \M} \left(  \frac{ \left\| \x^* - \x \right\|_2^2 } {\left\| \x \right\|_2^2  } \right) \leq  C m  ^{-\frac{2\beta}{2\beta+1}}.
\end{eqnarray*}

We next show the lower bound.  Based on Theorem 3, we have
\begin{eqnarray*}
  && \inf_{(\x^*, \M) \in \Theta_{\rm u} } \sup_{\x \in \BLT_{\Adj}(K, \beta, \mu)} \mathbb{E}_{\x, \M} \left(  \left\|  \x^* - \x \right\|_2^2  \right) 
  \\
  & \geq & \frac{c_1 \mu \left\|  \x \right\|_2^2  } {\kappa_0^{2\beta}}  \left( 1 -   \frac{ c \mu  \left\| \x \right\|_2^2 } { \sigma^2 \kappa_0^{2\beta+2} } \left\| \Vm_{(2,\kappa_0)} \right\|_{F}^2 m \right),
  \\
  & = & \frac{c_1 \mu \left\|  \x \right\|_2^2  } {\kappa_0^{2\beta}}  \left( 1 -   \frac{ c \mu  \left\| \x \right\|_2^2 } { \sigma^2 \kappa_0^{2\beta+1} N} m  \right)
  \\
  & \geq & \frac{c_1 \mu \left\|  \x \right\|_2^2  } {\kappa_0^{2\beta}}  \left( 1 -   \frac{ c \mu   \max_i  x_i^2  } { \sigma^2 \kappa_0^{2\beta+1} } m  \right)
  \\
  & \asymp & c  \left\|  \x \right\|_2^2 m^{-\frac{2 \beta}{2\beta+1}},
\end{eqnarray*}
where $\kappa_0$ is  of the order of $m^{\frac{1}{2\beta+1}}$. 

For optimal sampling scores based sampling, based on Corollary 2, we have
  \begin{eqnarray*}
    && \frac{\mu }{1 + \kappa^{2\beta}}  \left\| \x \right\|_2^2 +  \frac{1}{m} \left( \sum_{i=1}^N \sqrt{  \sum_{k=1}^{\kappa} \Um_{k,i}^2  \left( x_i^2 + \sigma^2 \right)  }  \right)^2
    \\
    & \leq &  \left\| \x \right\|_2^2  \left( \frac{\mu }{1 + \kappa^{2\beta}}  +  \frac{1}{m}  \frac{ \max_i x_i^2 + \sigma^2  }{\left\| \x \right\|_2^2}  \left\|  \Um_{(\kappa)} \right\|_{2,1} \right)
    \\
    & \stackrel{(a)}{ \leq }  &  \left\| \x \right\|_2^2  \left( \frac{\mu }{1 + \kappa^{2\beta}}  +   \frac{ \max_i x_i^2 + \sigma^2  }{\left\| \x \right\|_2^2} \frac{ N \kappa}{m}   \right)
    \\
    & \asymp & C \left\| \x \right\|_2^2 m^{-\frac{2\beta}{2\beta+1}},
\end{eqnarray*}
  where $(a)$ follows from that, based on Definition 5, $ \left\|
    \Um_{(\kappa)} \right\|_{2, 1}^2 = \left( N \sqrt{\kappa
      (\frac{c}{\sqrt{N}})^2} \right)^2 = O(N \kappa)$, which is of
  the same order of $N \left\| \Um_{(\kappa)} \right\|_{F}^2$. Since
  at least the sampled projection estimator satisfies this rate of
  convergence, we have
\begin{eqnarray*}
  \inf_{(\x^*, \M) \in \Theta_{\rm e} } \sup_{\x \in \BLT_{\Adj}(K, \beta, \mu)} \mathbb{E}_{\x, \M} \left( \frac{ \left\| \x^* - \x \right\|_2^2 } {\left\| \x \right\|_2^2} \right) \leq  C m^{-\frac{2\beta}{2\beta+1}}.
\end{eqnarray*}

Based on Definition 5, $ \left\| \Vm_{(2,\kappa_0)} \right\|_{\infty,
  2}^2 = \kappa_0 (\frac{c}{\sqrt{N}})^2 = c^2 \kappa_0/N$, which is
of the same order of $\left\| \Vm_{(2,\kappa_0)} \right\|_{F}^2/N$, we
have
\begin{eqnarray*}
  && \inf_{(\x^*, \M) \in \Theta_{\rm e} } \sup_{\x \in \BLT_{\Adj}(K, \beta, \mu)} \mathbb{E}_{\x, \M} \left(  \left\|  \x^* - \x \right\|_2^2  \right) 
  \\
  & \geq & \frac{c_1 \mu \left\|  \x \right\|_2^2  } {\kappa_0^{2\beta}}  \left( 1 -   \frac{ c \mu  \left\| \x \right\|_2^2 } { \sigma^2 \kappa_0^{2\beta+2} } \left\| \Vm_{(2,\kappa_0)} \right\|_{\infty, 2}^2 m \right),
  \\
  & = & \frac{c_1 \mu \left\|  \x \right\|_2^2  } {\kappa_0^{2\beta}}  \left( 1 -   \frac{ c \mu  \left\| \x \right\|_2^2 } { \sigma^2 \kappa_0^{2\beta+1} N} m \right)
  \\
  & \geq & \frac{c_1 \mu \left\|  \x \right\|_2^2  } {\kappa_0^{2\beta}}  \left( 1 -   \frac{ c \mu   \max_i  x_i^2  } { \sigma^2 \kappa_0^{2\beta+1} } m \right)
  \\
  & \asymp & c  \left\|  \x \right\|_2^2  m^{-\frac{2 \beta}{2\beta+1}},
\end{eqnarray*}
where $\kappa_0$ is of the order of $m^{\frac{1}{2\beta+1}}$.
\hfill$\blacksquare$

\subsection{Proof of Corollaries 5 and 6}
\label{sec:app4}
For a Type-2 graph, we assume that a few elements in an approximately
bandlimited signal have much higher magnitudes than the others.  The
intuition is that, since $\Vm$ is sparse, the energy concentrates in
$O(\kappa)$ rows of $\Vm_{(\kappa)}$, thus, $O(\kappa)$ components of
an approximately bandlimited signal, $x_i \approx \vv_{i,(\kappa)}^T
\widehat{\x}_{(\kappa)}$, have much higher magnitudes than the
others. In other words, $\kappa \max_i x_i^2$ and $\left\| \x
\right\|_2^2$ are of the same order.
   
For uniform sampling, based on Corollary 1, we have
\begin{eqnarray*}
    && 
    \frac{\mu }{1 + \kappa^{2\beta}}  \left\| \x \right\|_2^2 + \frac{N}{m} \sum_{k=1}^{\kappa} \sum_{i=1}^N \Um_{k,i}^2  \left( x_i^2 + \sigma^2 \right) 
    \\
    & \leq & \left\| \x \right\|_2^2 \left(   \frac{\mu }{ \kappa^{2\beta}} + \frac{N \left( \max_i x_i^2 + \sigma^2 \right) }{m  \left\| \x \right\|_2^2} 
      \sum_{k=1}^{\kappa} \sum_{i=1}^N \Um_{k,i}^2  \right) 
    \\
    & = &  \left\| \x \right\|_2^2 \left(   \frac{\mu }{ \kappa^{2\beta}} + \frac{ N (\max_i x_i^2 + \sigma^2) } { \kappa \max_i x_i^2}  \frac{ \kappa}{m} \right) 
    \\
    & \leq & \left\| \x \right\|_2^2 \left(   \frac{\mu }{ \kappa^{2\beta}} + \frac{ c N}{m}   \right) 
    \\
    & \asymp & C \left\| \x \right\|_2^2 m^{-\frac{2\beta}{2\beta+\gamma}},
\end{eqnarray*}
where $\gamma$ varies with $\kappa$ to satisfy $\kappa^\gamma \leq N$
($\gamma > 1$) and $\kappa$ is of the order of
$m^{\frac{1}{2\beta+\gamma}}$, and $C > 0 $ is some constant. Since at
least Algorithm 1 satisfies this rate of convergence, we thus have
\begin{eqnarray*}
  \inf_{(\x^*, \M) \in \Theta_{\rm u} } \sup_{\x \in \BLT_{\Adj}(K, \beta, \mu)} \mathbb{E}_{\x, \M} \left(  \left\| \x^* - \x \right\|_2^2 \right) \leq  C m  ^{-\frac{2\beta}{2\beta+\gamma}}.
\end{eqnarray*}

We next show the lower bound.  Based on Theorem 3, we have
\begin{eqnarray*}
  && \inf_{(\x^*, \M) \in \Theta_{\rm u} } \sup_{\x \in \BLT_{\Adj}(K, \beta, \mu)} \mathbb{E}_{\x, \M} \left(  \left\|  \x^* - \x \right\|_2^2  \right) 
  \\
  & \geq & \frac{c_1 \mu \left\|  \x \right\|_2^2  } {\kappa_0^{2\beta}}  \left( 1 -   \frac{ c \mu  \left\| \x \right\|_2^2 } { \sigma^2 \kappa_0^{2\beta+2} N } \left\| \Vm_{(2,\kappa_0)} \right\|_{F}^2 m \right),
  \\
  & = & \frac{c_1 \mu \left\|  \x \right\|_2^2  } {\kappa_0^{2\beta}}  \left( 1 -   \frac{ c \mu  \kappa_0 \max_i x_i^2 } { \sigma^2 \kappa_0^{2\beta+1} N} m  \right)
  \\
  & = & \frac{c_1 \mu \left\|  \x \right\|_2^2  } {\kappa_0^{2\beta}}  \left( 1 -   \frac{ c \mu  \max_i x_i^2 } { \sigma^2 \kappa_0^{2\beta+\gamma} } m  \right)
  \\
  & \asymp & c m^{-\frac{2 \beta}{2\beta+\gamma}},
\end{eqnarray*}
where $\kappa_0$ is of the order of $m^{\frac{1}{2\beta+\gamma}}$. 

For optimal sampling scores based sampling, based on Corollary 2, we have
\begin{eqnarray*}
     && \frac{\mu }{1 + \kappa^{2\beta}}  \left\| \x \right\|_2^2 +  \frac{1}{m} \left( \sum_{i=1}^N \sqrt{  \sum_{k=1}^{\kappa} \Um_{k,i}^2  \left( x_i^2 + \sigma^2 \right)  }  \right)^2
     \\
     & \leq &  \left\| \x \right\|_2^2  \left( \frac{\mu }{1 + \kappa^{2\beta}}  +  \frac{1}{m}  \frac{ \max_i x_i^2 + \sigma^2  }{\left\| \x \right\|_2^2}  \left\|  \Um_{(\kappa)} \right\|_{2,1}^2 \right)
     \\
     & \stackrel{(a)}{\leq}  &  \left\| \x \right\|_2^2  \left( \frac{\mu }{1 + \kappa^{2\beta}}  +   \frac{  \max_i x_i^2 + \sigma^2  }{\left\| \x \right\|_2^2} \frac{\kappa^2}{m}   \right)
     \\
     & = &  \left\| \x \right\|_2^2  \left( \frac{\mu }{1 + \kappa^{2\beta}}  +   \frac{  \max_i x_i^2 + \sigma^2  }{  \kappa \max_i x_i^2 } \frac{\kappa^2}{m}   \right)
     \\
     & \asymp & C \left\| \x \right\|_2^2 m^{-\frac{2\beta}{2\beta+1}},
\end{eqnarray*} 
where $(a)$ follows from the energy concentrating in $O(\kappa)$
columns of $\Um_{(\kappa)}$ as shown in Definition 6 and the upper
bound reaching the minimum when $\kappa$ is of the order of
$m^{\frac{1}{2\beta+1}}$.  Since at least Algorithm 1 satisfies this
rate of convergence, we have
\begin{eqnarray*}
  \inf_{(\x^*, \M) \in \Theta_{\rm e} } \sup_{\x \in \BLT_{\Adj}(K, \beta, \mu)} \mathbb{E}_{\x, \M} \left(  \left\| \x^* - \x \right\|_2^2 \right) \leq  C m  ^{-\frac{2\beta}{2\beta+1}}.
\end{eqnarray*}

Based on Definition 6, $ \left\| \Vm_{(2, \kappa_0)} \right\|_{\infty,
  2}^2 = c $, we have
\begin{eqnarray*}
  && \inf_{(\x^*, \M) \in \Theta_{\rm e} } \sup_{\x \in \BLT_{\Adj}(K, \beta, \mu)} \mathbb{E}_{\x, \M} \left(  \left\|  \x^* - \x \right\|_2^2  \right) 
  \\
  & \geq & \frac{c_1 \mu \left\|  \x \right\|_2^2  } {\kappa_0^{2\beta}}  \left( 1 -   \frac{ c \mu  \left\| \x \right\|_2^2 } { \sigma^2 \kappa_0^{2\beta+2} } \left\| \Vm_{(2, \kappa_0)} \right\|_{\infty, 2}^2 m \right),
  \\
  & = & \frac{c_1 \mu \left\|  \x \right\|_2^2  } {\kappa_0^{2\beta}}  \left( 1 -   \frac{ c \mu \kappa_0  \max_i x_i^2 } { \sigma^2 \kappa_0^{2\beta+2} } m  \right)
  \\
  & \asymp & c m^{-\frac{2 \beta}{2\beta+ 1}}.
\end{eqnarray*}
\hfill$\blacksquare$

\end{document}